\documentclass[11pt]{article}
\usepackage[T1]{fontenc}
\usepackage{lmodern}
\usepackage{amsfonts,amsmath,amssymb,amsthm}
\usepackage{yhmath}
\usepackage{mathtools}
\usepackage{fullpage}
\usepackage{xcolor}
\definecolor{maincolor}{HTML}{BA0C2F}
\usepackage[normalem]{ulem}
\usepackage{thmtools}
\usepackage{thm-restate}
\usepackage[colorlinks,allcolors=maincolor]{hyperref}
\usepackage[capitalize]{cleveref}
\usepackage{tikz,pgfplots}
\pgfplotsset{compat=1.18}
\usepackage{enumerate}
\usepackage{algorithm}
\usepackage{algpseudocode}

\newcommand{\coeff}{\mathrm{coeff}}
\newcommand{\res}{\mathrm{res}}
\newcommand{\pres}{\mathrm{proj.res}}

\newcommand{\ord}{\mathrm{ord}}
\newcommand{\id}[2]{\mathrm{Id}_{#1}(#2)}
\newcommand{\cont}[1]{\mathrm{cont}(#1)}
\newcommand{\disc}[1]{\mathrm{disc}(#1)}
\newcommand{\ideal}[1]{\left\langle #1 \right\rangle}
\newcommand{\vecX}{\mathbf{X}}
\newcommand{\vecY}{\mathbf{Y}}
\newcommand{\vecZ}{\mathbf{Z}}
\newcommand{\vecT}{\mathbf{T}}
\newcommand{\veca}{\mathbf{a}}
\newcommand{\vecb}{\mathbf{b}}
\newcommand{\vecc}{\mathbf{c}}
\newcommand{\vecu}{\mathbf{u}}
\newcommand{\vecv}{\mathbf{v}}
\newcommand{\vecx}{\mathbf{x}}
\newcommand{\ZZ}{\mathbb{Z}}
\newcommand{\RR}{\mathbb{R}}
\newcommand{\QQ}{\mathbb{Q}}
\newcommand{\NN}{\mathbb{N}}
\newcommand{\FF}{\mathbb{F}}
\newcommand{\KK}{\mathbb{K}}
\newcommand{\LL}{\mathbb{L}}
\newcommand{\EE}{\mathbb{E}}
\newcommand{\aff}{\mathbb{A}}
\newcommand{\proj}{\mathbb{P}}
\newcommand{\fr}[1]
{\mathrm{Frac}(#1)}
\newcommand{\order}{\mathcal{O}}
\newcommand{\pp}{\mathfrak{p}}

\newcommand{\mm}{\mathfrak{m}}
\newcommand{\lt}[1]{\mathrm{LT}(#1)}
\newcommand{\lc}[1]{\mathrm{LC}(#1)}
\newcommand{\lm}[1]{\mathrm{LM}(#1)}

\newcommand{\grobner}{Gr\"{o}bner }
\newcommand{\bezout}{B\'{e}zout}

\renewcommand{\epsilon}{\varepsilon}

\newtheorem{theorem}{Theorem}[section]
\newtheorem{corollary}[theorem]{Corollary}
\newtheorem{lemma}[theorem]{Lemma}
\newtheorem{fact}[theorem]{Fact}
\newtheorem{claim}[theorem]{Claim}
\newtheorem{definition}[theorem]{Definition}

\theoremstyle{remark}
\newtheorem{remark}{Remark}
\newtheorem*{remark*}{Remark}

\title{Deterministic Depth-4 PIT and Normalization}
 \author{Zeyu Guo\thanks{Department of Computer Science and Engineering, The Ohio State University. Email: \texttt{zguotcs@gmail.com}. Supported by the NSF CAREER award (CCF-2440926).}
 \and Siki Wang\thanks{Computing and Mathematical Sciences Department, Caltech. Email: \texttt{siki.wang@caltech.edu}}}
\date{}

\begin{document}

\maketitle

\begin{abstract}
In this paper, we initiate the study of deterministic PIT for $\Sigma^{[k]}\Pi\Sigma\Pi^{[\delta]}$ circuits over fields of any characteristic, where $k$ and $\delta$ are bounded. Our main result is a deterministic polynomial-time black-box PIT algorithm for $\Sigma^{[3]}\Pi\Sigma\Pi^{[\delta]}$ circuits, under the additional condition that one of the summands at the top $\Sigma$ gate is squarefree.

Our techniques are purely algebro-geometric: they do not rely on Sylvester--Gallai-type theorems, and our PIT result holds over arbitrary fields.

The core of our proof is based on the normalization of algebraic varieties. Specifically, we carry out the analysis in the integral closure of a coordinate ring, which enjoys better algebraic properties than the original ring.
\end{abstract}

\section{Introduction}

Polynomial Identity Testing (PIT) is the fundamental problem of deciding whether a given multivariate algebraic circuit computes the identically zero polynomial. Equivalently, it asks whether two arithmetic circuits compute the same polynomial, by checking whether their difference is identically zero. Despite its simple formulation, PIT is one of the central problems in the study of derandomization. On the one hand, it admits an efficient randomized algorithm: by the Schwartz--Zippel Lemma~\cite{Sch80,Zip79}, evaluating a degree-$d$ polynomial at a random point over a sufficiently large field yields a correct identity test with high probability. This leads to a fast Monte Carlo algorithm that treats the circuit as a black box. On the other hand, no efficient deterministic algorithm is known, even when the circuit’s structure is fully accessible. In complexity-theoretic terms, PIT lies in the class \textsf{coRP}, but whether it can be solved in \textsf{P} remains a major open question. Closing this gap is widely regarded as a central challenge in theoretical computer science. Indeed, PIT is often viewed as the algebraic analog of the classic \textsf{P} vs.\ \textsf{BPP} (or \textsf{P} vs.\ \textsf{RP}) question.

The importance of PIT lies both in its algorithmic applications and its deep connections to complexity theory. It has been used in a wide range of settings, such as primality testing \cite{AKS04}, polynomial factoring \cite{KSS14,BSV20,KRS24,DST24}, and perfect matching \cite{Lov79,MVV87,CRS95,FGT21,ST17}. More fundamentally, PIT plays a central role in the \emph{hardness vs.\ randomness} paradigm, and its derandomization is known to imply long-sought arithmetic circuit lower bounds \cite{HS80,KI04}.

Over the past few decades, there have been numerous results on deterministic PIT for various restricted models, including sparse polynomials, depth-3 circuits, low-depth circuits, and read-once arithmetic branching programs, among others. See, for example,~\cite{LST24,AF22} for recent breakthroughs, including subexponential-time deterministic PIT algorithms for low-depth circuits over fields of characteristic zero or large characteristic. For surveys of earlier and recent developments, see~\cite{Sax09,Sax14,SY10, DG24}. Nonetheless, despite decades of significant progress, a general polynomial-time deterministic PIT algorithm remains elusive. As such, PIT continues to serve as a central testing ground for our understanding of algebraic structure, pseudorandomness, and computational hardness.

\paragraph{Depth-4 PIT.}
A major breakthrough by Agrawal and Vinay~\cite{AV08} revealed that PIT for \mbox{depth-4}
circuits essentially captures the full complexity of general PIT.
In particular, for $n$-variate polynomials of degree $O(n)$, subexponential-size arithmetic circuits can be reduced to depth-4 circuits of subexponential size.
Thus, proving exponential lower bounds for depth-4 circuits would imply
exponential lower bounds for general circuits.
Combined with hardness-vs-randomness connections~\cite{NW94, KI04},
they further showed that a complete derandomization of depth-4 PIT,
even for circuits with slightly unbounded second-layer fan-in
and $O(\log n)$ bottom fan-in,
would yield a quasi-polynomial-time deterministic PIT algorithm
for general circuits.
As a result, depth-4 PIT has become an important frontier in algebraic derandomization.
Notably, no analogous constant-depth reduction is known in the Boolean setting.

Despite this equivalence, the expressive power of general depth-4 circuits
makes deterministic PIT for them highly challenging.
This has motivated a refined line of work focusing on the syntactic subclass $\Sigma^{[k]}\Pi\Sigma\Pi^{[\delta]}$, consisting of depth-4 circuits with top fan-in at most $k$ and bottom fan-in
at most $\delta$, where $k$ and $\delta$ are bounded.
The study of deterministic PIT for these circuits has led to the development
of new techniques grounded in Sylvester--Gallai-type arguments
and tools from algebraic geometry.

\subsection{Previous Work}

\paragraph{Sylvester--Gallai-based approach.}
A $\Sigma^{[k]}\Pi\Sigma\Pi^{[\delta]}$ circuit consists of a top $\Sigma$ gate of fan-in at most $k$, followed by alternating layers of unbounded $\Pi$, $\Sigma$, and $\Pi$ gates, where the bottom $\Pi$ gates have fan-in at most $\delta$. This model naturally generalizes bounded top fan-in depth-3 circuits ($\Sigma^{[k]}\Pi\Sigma$), for which efficient deterministic PIT algorithms are known. 
Over fields of characteristic zero, Sylvester--Gallai (SG)-type theorems \cite{EK66,BW90} from combinatorial geometry have been used to establish constant-rank bounds for simple minimal $\Sigma^{[k]}\Pi\Sigma$ identities,
leading to deterministic polynomial-time black-box PIT algorithms~\cite{DS07,KS09,SS13}.

Gupta~\cite{Gup14} conjectured a nonlinear Sylvester--Gallai-type statement which, if true, would yield deterministic polynomial-time black-box PIT algorithms for $\Sigma^{[k]}\Pi\Sigma\Pi^{[\delta]}$ circuits with bounded $k$ and $\delta$ over fields of characteristic zero. Motivated by this conjecture, Shpilka proved an SG-type theorem for quadratic polynomials, initiating a substantial line of work~\cite{Shp20,PS22,PS21,GOS22,OS22,GOPS23,OS24,GOS25,GOS25-3, GOSSS26} aimed at fully resolving the conjecture.

Peleg and Shpilka~\cite{PS21}, building on earlier work~\cite{Shp20,PS22}, proved a quadratic SG-type theorem that yields a deterministic polynomial-time black-box PIT algorithm for $\Sigma^{[3]}\Pi\Sigma\Pi^{[2]}$ circuits. Subsequent works~\cite{OS22,GOPS23,OS24,GOS25} extended these results and introduced new techniques, including those based on Stillman uniformity~\cite{AH20,ESS19,DLL19}, culminating in the work of Garg, Oliveira, and Sengupta~\cite{GOS25-3}, which established constant-rank bounds for nonlinear
Edelstein--Kelly configurations and gave a deterministic polynomial-time black-box PIT algorithm for $\Sigma^{[3]}\Pi\Sigma\Pi^{[\delta]}$ circuits over fields of characteristic zero for any bounded $\delta$. More recently, Garg, Oliveira, Sengupta, Shalmon, and Shpilka~\cite{GOSSS26} established similar results for $\Sigma^{[k]}\Pi\Sigma\Pi^{[2]}$ circuits over fields of characteristic zero for any bounded $k$.

However, SG-based methods remain largely confined to characteristic zero. Kayal and Saxena~\cite{KS07} showed that the constant-rank bounds implied by SG-type arguments do not hold in positive characteristic, even for the simpler $\Sigma^{[k]}\Pi\Sigma$ model.

\paragraph{Other work.}
Dutta, Dwivedi, and Saxena~\cite{DDS21} gave a deterministic quasipolynomial-time
black-box PIT algorithm for $\Sigma^{[k]}\Pi\Sigma\Pi^{[\delta]}$ circuits over fields
of characteristic zero or sufficiently large characteristic. Their approach uses the
Jacobian criterion and faithful homomorphisms~\cite{ASSS16} to reduce the problem to
black-box PIT for $\Sigma\wedge\Sigma\Pi^{[\delta]}$ circuits (sums of powers of
low-degree polynomials), for which only quasipolynomial-time hitting sets are
known~\cite{For15}.

We also mention the breakthrough of superpolynomial lower bounds for low-depth circuits by Limaye, Srinivasan, and Tavenas~\cite{LST24}. This result, together with the earlier work of Chou, Kumar, and Solomon \cite{CKS18} or the work of Andrews and Forbes \cite{AF22}, yields subexponential-time black-box PIT algorithms for low-depth circuits, including $\Sigma^{[k]}\Pi\Sigma\Pi^{[\delta]}$ circuits.  These algorithms also rely on properties that only hold when the characteristic is zero or large. Subsequently, Forbes~\cite{For24} extended the superpolynomial lower bound for low-depth circuits to positive characteristics. However, whether the PIT algorithms themselves can be extended to this setting remains an open question~\cite{For24}.

In summary, all of the depth-4 PIT algorithms discussed above require characteristic zero or sufficiently large characteristic. This naturally raises the following question: Is it possible to design a nontrivial, or even polynomial-time, deterministic PIT algorithm for $\Sigma^{[k]}\Pi\Sigma\Pi^{[\delta]}$ circuits over arbitrary fields? We believe the answer is yes. In this work, we make progress toward answering this question affirmatively.

As a reality check, any such result must first address the simpler case of $\Sigma^{[k]}\Pi\Sigma$ circuits. Deterministic polynomial-time algorithms for this class are indeed known in positive characteristics: they were first obtained by Kayal and Saxena in the white-box setting~\cite{KS07}, and later by Saxena and Seshadhri in the black-box setting~\cite{SS12}. Notably, these results do not rely on SG-type theorems or other arguments that require assumptions on the characteristic of the base field. Our work is inspired by these approaches and seeks to extend them to the depth-4 setting.

\subsection{Our Results}

In this paper, we initiate the study of deterministic PIT for $\Sigma^{[k]}\Pi\Sigma\Pi^{[\delta]}$ circuits over fields of any characteristic.
Our main result is a deterministic polynomial-time black-box PIT algorithm for $\Sigma^{[3]}\Pi\Sigma\Pi^{[\delta]}$ circuits under the condition that one of the summands at the top $\Sigma$ gate is squarefree.

Our techniques are purely algebro-geometric. Notably, the proof does not rely on Sylvester--Gallai-type theorems, and the result remains valid even over fields of small positive characteristic.

We begin by stating the homogeneous version of our result.

\begin{theorem}[Main theorem, homogeneous version]\label{thm:main-intro}
Let $C_{n,d,k,\delta,\FF}$ be the set of polynomials $F\in\FF[\vecX]=\FF[X_1,\dots,X_n]$ over a field $\FF$ satisfying the following conditions:
\begin{enumerate}[(1)]
    \item $F$ can be expressed as a sum $F=\sum_{i=0}^{k_0-1} F_i$, where $k_0\leq k$, $F_i=\prod_{j=1}^{m_i} f_{i,j}$ for $i\in \{0,1,\dots,k_0-1\}$, and each $f_{i,j}\in \FF[\vecX]$ is a nonzero homogeneous polynomial of degree at most $\delta$.
    \item $\deg(F_i)=d_0$ for some $d_0\leq d$ and all $i\in \{0,1,\dots,k_0-1\}$.
    \item\label{item:hom3-intro} There exists $i\in \{0,1,\dots,k_0-1\}$ such that $F_i$ is squarefree over $\overline{\FF}$, meaning that the irreducible factors of $F_i$ over $\overline{\FF}$ are distinct.
\end{enumerate}
Then there exists an explicit\footnote{Here explicitness means that $\mathcal{H}$ can be constructed deterministically in time $(nd)^{O_{\delta}(1)}$. Moreover, $\mathcal{H}$ can be chosen to be a subset of $\QQ^n$ if $\mathrm{char}(\FF)=0$, or a subset of $\KK^n$ if $\FF$ is a finite field, where $\KK$ is an extension field of $\FF$ of size polynomial in $(nd)^{O_{\delta}(1)}$ and $|\FF|$. The same explicitness assumption holds for \cref{cor:inhom}.} $(nd)^{O_{\delta}(1)}$-sized hitting set $\mathcal{H}\subseteq \overline{\FF}^n$ for $C_{n,d,3,\delta,\FF}$. Equivalently, there exists a deterministic black-box PIT algorithm for  $C_{n,d,3,\delta,\FF}$ that makes at most $(nd)^{O_\delta(1)}$ evaluation queries.
\end{theorem}

The above homogeneous result easily extends to an analogous statement for inhomogeneous polynomials. We now explain how this extension works.

Let $F=\sum_{i=0}^{k_0-1} F_i\in \FF[X_1,\dots,X_n]$, where $F_0,\dots,F_{k_0-1}$ are nonzero and possibly inhomogeneous polynomials. Define the homogenization of $F$ with respect to $F_0,\dots,F_{k_0-1}$ as 
\[
\mathrm{H}(F,F_0,\dots,F_{k_0-1})
:=\sum_{i=0}^{k_0-1} F_i(X_1/X_0, \dots,X_n/X_0)X_0^{d_0}, \quad \text{where}~d_0=\max_{0\leq i\leq k_0-1}(\deg(F_i)),
\]
which is a homogeneous polynomial living in $\FF[X_0,X_1,\dots,X_n]$.

We now state the inhomogeneous version of our result.

\begin{corollary}[Main theorem, inhomogeneous version]\label{cor:inhom}
        Let $C_{n,d,k,\delta,\FF}^*$ be the set of polynomials $F\in \FF[\vecX]=\FF[X_1,\dots,X_n]$ over a field $\FF$ that can be written as a sum $F=\sum_{i=0}^{k_0-1} F_i$ for some $k_0\leq k$ such that $\mathrm{H}(F,F_0,\dots,F_{k_0-1})\in C_{n+1,d,k,\delta,\FF}$. This includes all those $F=\sum_{i=0}^{k_0-1} F_i$ with $k_0\leq k$ for which the following conditions hold:
\begin{enumerate}[(1)]
    \item\label{item:inhom1} $F_i=\prod_{j=1}^{m_i} f_{i,j}$ for $i\in \{0,1,\dots,k_0-1\}$, where each $f_{i,j}\in \FF[\vecX]$ is a nonzero 
    polynomial of degree at most $\delta$.
    \item\label{item:inhom2} $d_0:=\max_{0\leq i\leq k_0-1}(\deg(F_i))$ is at most $d$.
    \item\label{item:inhom3} $F_i$ is squarefree over $\overline{\FF}$ for some $i\in \{0,1,\dots,k_0-1\}$ satisfying $\deg(F_i)=d_0$, or more generally, satisfying $\deg(F_i)\geq d_0-1$. 
\end{enumerate}
And there exists an explicit $(nd)^{O_{\delta}(1)}$-sized hitting set $\mathcal{H}\subseteq \overline{\FF}^n$ for $C_{n,d,3,\delta,\FF}^*$. Equivalently, there exists a deterministic black-box PIT algorithm for  $C_{n,d,3,\delta,\FF}^*$ that makes at most $(nd)^{O_\delta(1)}$ evaluation queries.
\end{corollary}

\begin{proof}
Consider $i\in\{0,1,\dots,k_0-1\}$.
Note that
\[
F_i(X_1/X_0, \dots,X_n/X_0)X_0^{d_0}=\left(\prod_{j=1}^{m_i} \left(f_{i,j}(X_1/X_0, \dots,X_n/X_0)X_0^{\deg(f_{i,j})}\right)\right) X_0^{d_0-\deg(F_i)}.
\]
Here, each factor $f_{i,j}(X_1/X_0, \dots,X_n/X_0)X_0^{\deg(f_{i,j})}$ is not divisible by $X_0$.
Therefore, if $F_i$ is squarefree over $\overline{\FF}$ and $\deg(F_i)\geq d_0-1$, then $F_i(X_1/X_0, \dots,X_n/X_0)X_0^{d_0}$ is also squarefree over $\overline{\FF}$. So if $F=\sum_{i=0}^{k_0-1} F_i$ satisfies \cref{item:inhom3} in \cref{cor:inhom}, then $\mathrm{H}(F,F_0,\dots,F_{k_0-1})$ satisfies \cref{item:hom3-intro} in \cref{thm:main-intro}.
It is straightforward to verify that \cref{item:inhom1} and \cref{item:inhom2} in \cref{cor:inhom} imply the corresponding items in \cref{thm:main-intro} as well.

It remains to construct an explicit $(nd)^{O_{\delta}(1)}$-sized hitting set $\mathcal{H}^*\subseteq \overline{\FF}^n$ for $C_{n,d,3,\delta,\FF}^*$.
First, apply \cref{thm:main-intro} over
$\overline{\FF}$ to obtain an explicit hitting set
$\mathcal H\subseteq\overline{\FF}^{n+1}$ for
$C_{n+1,d,3,\delta,\overline{\FF}}$.

Let us first assume that $a_0\neq 0$ for all  $(a_0,\dots,a_n)\in \mathcal{H}$.
Construct the set
\[
\mathcal{H}^*=\{(a_1/a_0,\dots,a_n/a_0): (a_0,\dots,a_n)\in\mathcal{H}\}.
\]
Consider a nonzero polynomial $F=\sum_{i=0}^{k_0-1} F_i\in C_{n,d,3,\delta,\FF}^*$. Then $\mathrm{H}(F,F_0,\dots,F_{k_0-1})$ is a nonzero polynomial in $C_{n+1,d,3,\delta,\FF}$. 
Choose $\veca=(a_0,\dots,a_n)\in \mathcal{H}$  such that $\mathrm{H}(F,F_0,\dots,F_{k_0-1})(\veca)\neq 0$.
By definition,
\[
\mathrm{H}(F,F_0,\dots,F_{k_0-1})(\veca)=\sum_{i=0}^{k_0-1} F_i(a_1/a_0, \dots,a_n/a_0)a_0^{d_0}=F(a_1/a_0,\dots,a_n/a_0)a_0^{d_0}.
\]
Therefore, $F(a_1/a_0,\dots,a_n/a_0)=\mathrm{H}(F,F_0,\dots,F_{k_0-1})(\veca)a_0^{-d_0}\neq 0$, where $(a_1/a_0, \dots,a_n/a_0)\in\mathcal{H}^*$ by construction. So $\mathcal{H}^*$ is a hitting set for $C_{n,d,3,\delta,\FF}^*$, whose size is at most $|\mathcal{H}|=(nd)^{O_{\delta}(1)}$. 

Finally, we remove the assumption that $a_0 \neq 0$ for all $(a_0, \dots, a_n) \in \mathcal{H}$. Observe that the class $C_{n+1,d,3,\delta,\overline{\mathbb F}}$ is closed under invertible linear transformations of the coordinates. 

We may assume that $\mathbf{0} = (0, \dots, 0) \notin \mathcal{H}$, since any non-constant homogeneous polynomial always vanishes at $\mathbf{0}$. Hence, every point in $\mathcal{H}$ has at least one nonzero coordinate.

We can deterministically and efficiently find an invertible linear transformation $\phi : \overline{\FF}^{n+1} \to \overline{\FF}^{n+1}$ such that every point $(a_0, \dots, a_n) \in \phi(\mathcal{H})$ satisfies $a_0 \neq 0$. We then replace $\mathcal{H}$ by $\phi(\mathcal{H})$ and construct $\mathcal{H}^*$ as before.
\end{proof}

We remark that removing \cref{item:inhom3} from \cref{cor:inhom} (the squarefreeness condition) recovers the class $\Sigma^{[3]}\Pi\Sigma\Pi^{[\delta]}$.

For general constant $\delta$, no subexponential-time
deterministic black-box PIT algorithm was previously
known for these classes over fields of arbitrary
small positive characteristic.

\paragraph{Subsequent Work.}

The squarefreeness condition has subsequently been removed in \cite{Guo26} using the Mason--Stothers theorem. We refer the interested reader to that note for details.

\subsection{Proof Overview}

To explain our ideas, we begin with a solved case: consider a nonzero polynomial $F \in \Sigma^{[2]}\Pi\Sigma\Pi^{[\delta]}$, i.e., a depth-4 circuit with top fan-in two. In this case, we can write $F = F_1 + F_2 \neq 0$, where $F_1$ and $F_2$ are products of nonzero polynomials, each of degree at most $\delta$.

Suppose
\begin{equation}\label{eq:factor}
F_1 = g_1 \cdots g_r \quad \text{and} \quad F_2 = h_1 \cdots h_s
\end{equation}
are factorizations of $F_1$ and $F_2$ into irreducible polynomials, respectively. If $F_2 = c F_1$ for some $c \in \FF^\times$, then $F = (c+1)F_1$, which is again a product of degree-$(\leq \delta)$ nonzero polynomials. PIT for such polynomials is straightforward. So assume this is not the case. Then the factorizations in \eqref{eq:factor} do not ``match.'' That is, there does not exist a bijection $\sigma : [r] \to [s]$ such that $g_i$ and $h_{\sigma(i)}$ are scalar multiples of each other for all $i \in [r]$.

It is natural to choose an affine line $\ell$ such that restricting to $\ell$ reduces the ambient dimension while preserving the non-matching structure of the factorizations. In other words, we want $F_1|_\ell$ and $F_2|_\ell$ to still have non-matching factorizations. This ensures that $(F_1 + F_2)|_\ell \neq 0$, reducing the problem to PIT for univariate polynomials, which is easy.

To ensure the factorizations remain non-matching after restriction, it suffices to preserve the degree of every irreducible
factor and to ensure that any two coprime irreducible factors $g, h \in \{g_1, \dots, g_r, h_1, \dots, h_s\}$ remain coprime after restricting to $\ell$, i.e., they do not share a common root on the line. Geometrically, the latter condition amounts to ensuring that $\ell$ avoids the codimension-two variety $V(g, h)$. In~\cite{Guo24}, it is shown how to construct a polynomial-sized set of lines such that most lines avoid all such varieties. This effectively (re)solves the problem.

We now move to the class $C_{n,d,3,\delta,\FF}$. Let $F$ be a nonzero polynomial in this class. Then we may write
\[
F = F_0 + F_1 + F_2,
\]
where $F_0, F_1, F_2$ are homogeneous of the same degree and are products of nonzero homogeneous polynomials, each of degree at most $\delta$. Furthermore, one of the summands, say $F_0$, is squarefree over $\overline{\FF}$.

In the study of Boolean circuits, one common technique is applying a restriction (i.e. partial assignment) to simplify the circuit. An analogous idea works here: we restrict $F$ to the zero locus of an irreducible factor $\theta$ of $F_0$ to eliminate $F_0$. Algebraically, this corresponds to working in the quotient ring $\FF[X_1, \dots, X_n]/(\theta)$, where each $F_i$ is replaced by $\overline{F}_i := F_i \bmod \theta$. Since $\theta$ divides $F_0$, we have $\overline{F}_0 = 0$. Thus, we have effectively reduced to the $k = 2$ case, though now over the quotient ring.

Why can't we directly reuse the $k = 2$ argument? The issue is that $\FF[X_1, \dots, X_n]/(\theta)$ is not, in general, a unique factorization domain (UFD), so the factorizations of $\overline{F}_1$ and $\overline{F}_2$ are not well-defined.

But is unique factorization truly necessary? We argue that it is not: even in the absence of unique factorization, working in rings with the weaker property of \emph{normality} still enables us to obtain meaningful results.

\paragraph{Normality.}  Let $A$ be an integral domain, whose field of fractions is denoted by $\mathrm{Frac}(A)$. We say $A$ is \emph{integrally closed} if for any monic polynomial $P(X)\in A[X]$, all roots of $P(X)$ in $\mathrm{Frac}(A)$ are in $A$.
An irreducible affine variety is said to be \emph{normal} if its coordinate ring is integrally closed.

It is not easy to give a purely geometric definition of normality. However, its usefulness lies in the fact that, if $V$ is a normal variety with coordinate ring $A$, then for each codimension-one irreducible subvariety  $\mathcal{Z}\subseteq V$, there is a well-behaved ``order function'' $\ord_\mathcal{Z}: \mathrm{Frac}(A)\to \mathbb{Z}\cup\{\infty\}$ indicating the order of zeros or poles of every $g\in \mathrm{Frac}(A)$ along $\mathcal{Z}$. 
For example, for $A=\FF[X,Y]$ and $g=(X+Y)^2/X^3\in\mathrm{Frac}(A)$, we have $\ord_{V(X+Y)}(g)=2$ and $\ord_{V(X)}(g)=-3$. 

If $A$ is normal, then for $g\in A\setminus\{0\}$, one can define a ``generalized factorization'' of $g$, where the ``irreducible factors'' are not polynomials, but codimension-one irreducible subvarieties $\mathcal{Z}$ of $V$, each with multiplicity $\ord_{\mathcal{Z}}(g)$. (This is called the Weil divisor associated with $g$, written additively as $\mathrm{div}(g):=\sum_{\mathcal{Z}} \ord_\mathcal{Z}(g) \cdot \mathcal{Z}$.)

With this generalized notion of factorization, one can carry out an argument analogous to (and in fact generalizing) the one for the $k = 2$ case, assuming the variety defined by the factor $\theta$ of $F_0$ is normal.

Thus, normality may be viewed as a useful weakening of unique factorization. In general, however, the variety in question may even fail to be normal. To address this, we apply a standard technique from algebraic geometry known as \emph{normalization}.

\paragraph{Normalization.}
Conceptually, the normalization of a variety $V$ produces a normal variety $\widetilde{V}$ that best approximates $V$ among all normal varieties. Algebraically, if $V$ is affine, normalization corresponds to taking the integral closure $\widetilde{\FF[V]}$ of the coordinate ring $\FF[V]$ in its field of fractions. Our key idea is to work within the ``nicer'' ring $\widetilde{\FF[V]}$ in place of $\FF[V]$.

The strategy of enlarging a ring to recover a weak form of unique factorization has deep roots in number theory. It began with Kummer’s introduction of \emph{ideal numbers} to address the failure of unique factorization for cyclotomic integers, and was later formalized by Dedekind through the theory of \emph{ideals}, recovering unique factorization at the level of ideal decomposition in number rings. In modern terms, this philosophy is embodied in the process of normalization—passing to the integral closure of a ring in its field of fractions---which yields an integrally closed ring that better approximates a unique factorization domain.

As an illustrative example, work over $\RR$ and consider the plane curve $C$ defined by $Y^2 - X^2(X+1) = 0$ (see \cref{fig:normalization}). The rational function $Z = Y/X$ satisfies the monic polynomial $Z^2 - (X+1)$ on $C$, but is not a regular function in $\RR[C]$, meaning that it is not well-defined on the curve. Geometrically, this reflects the fact that $C$ has two branches at the point $(0,0)$, where the limits of $Y/X$ approach $1$ and $-1$, respectively. Introducing $Z$ as a new coordinate function separates these branches. One can show that $X$, $Y$, and $Z$ generate an integrally closed ring in $\mathrm{Frac}(\RR[C])$, and that their defining relations are generated by $Y^2 - X^2(X+1)$, $XZ = Y$, and $Z^2 - (X+1)$. These equations define the normalization $\widetilde{C} \subseteq \mathbb{A}^3$.

\begin{figure}[htb]
\centering
\begin{tikzpicture}
\begin{axis}[ticks = none, axis lines=none, unit vector ratio=1 1 1, xmin=-1.2, xmax=1.5, ymin=-1.8, ymax=1.8, smooth, samples=200, xlabel=$X$,ylabel=$Y$, at={(0,0)}]
\draw[domain=-1.45:1.45, variable=\t, thick] plot ({\t*\t-1}, {(\t*\t-1)*\t});
\node[thick, fill, circle,inner sep=1.8pt, label=right:{$(0,0)$}] (p) at (0,0){};
\end{axis}
\begin{axis}[view={15}{10}, smooth,
                axis lines=center, 
                xlabel=$X$,ylabel=$Y$,zlabel=$Z$, y label style={anchor=south},
                xmin=-1.2, xmax=1.5, ymin=-1.8, ymax=1.8, zmin=-1.5, zmax=1.5,
                ticks=none,
                unit vector ratio=1 1 1, at={(140pt,0)}]
\addplot3[variable=t,domain=-1.45:1.45,thick,samples y=1] ({t*t-1},{(t*t-1)*t}, t);
\node[thick, fill, circle,inner sep=1.8pt, label={[yshift=-0.2cm]above left:{$(0,0,1)$}}] (p1) at (0,0,1){};
\node[thick, fill, circle,inner sep=1.8pt, label=right:{$(0,0,-1)$}] (p2) at (0,0,-1){};
\end{axis}
\node at (2.2,1.6) {$C$};
\node at (6.3,1.6) {$\widetilde{C}$};
\draw [thick, -{>[scale=1.0]}]  (5.0,2.85) -- node[above] {$\pi$} (3.5,2.85);
\end{tikzpicture}
\caption{Normalization $\widetilde{C}$ of the curve $C$ defined by $Y^2 - X^2(X+1) = 0$. The map $\pi : \widetilde{C} \to C$ sends $(x, y, z) \mapsto (x, y)$.}
\label{fig:normalization}
\end{figure}
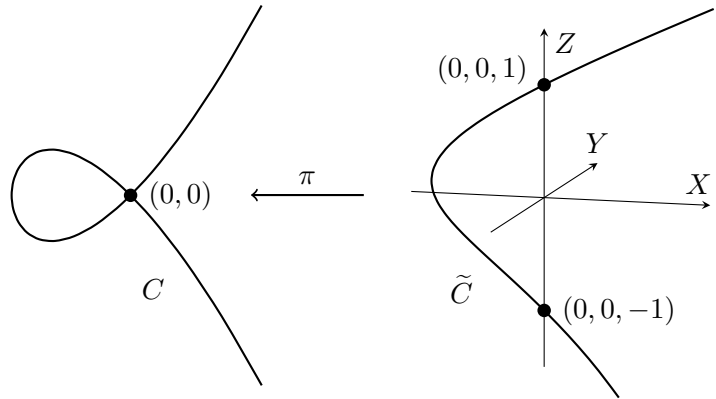

This example illustrates why normalization is useful for defining a generalized notion of factorization. On the non-normal curve $C$, the singular point $(0,0)$ behaves like two overlapping points, with functions like $Y/X$ exhibiting distinct limiting behavior along each branch. After normalization, these branches are separated into two points, $p_1 = (0,0,1)$ and $p_2 = (0,0,-1)$, each admitting a well-behaved order function $\ord_{p_i}(\cdot)$.

A knowledgeable reader may recognize the above example as the resolution of singularities of $C$. Indeed, a normal variety is regular in codimension one
\cite[\href{https://stacks.math.columbia.edu/tag/0345}
{Tag 0345}]{stacks-project}.
Consequently, the normalization of a curve is regular;
over a perfect field, it is smooth
\cite[\href{https://stacks.math.columbia.edu/tag/0B8Y}
{Tag 0B8Y}]{stacks-project}
and hence provides a resolution of singularities.

In higher dimensions, by contrast, resolution of singularities is highly complex in characteristic zero~\cite{Hir64,Hir64-2}, and remains an open problem in positive characteristic. Normalization, on the other hand, is significantly more tractable. A classical theorem of Emmy Noether shows that the integral closure of a coordinate ring $\FF[V]$ is finitely generated as a module over $\FF[V]$ (see, e.g.,~\cite{Eis13}).

Constructive normalization has also been extensively studied; see~\cite{Sto74,Sei70,Sei75,Tra84,deJ98}. In this paper, we only require normalization for curves, and we follow the framework of Trager~\cite{Tra84} to perform this task.

That said, normalization is generally regarded as a computationally expensive problem, often requiring tools such as \grobner bases. In the context of Geometric Complexity Theory, a key challenge is the non-normality of certain orbit closures related to the determinant and permanent, a fact established by Kumar~\cite{Kum13}. In that setting, the potential utility of normalization remains unclear due to its high complexity.

In our setting, however, this complexity is not a barrier. We apply a dimensionality reduction technique that restricts the input to a carefully chosen constant-dimensional subspace, and specifically, a plane. After restriction to a plane, the number of variables and the relevant degree bounds depend only on $\delta$. Controlling the coefficient complexity in the remaining parameters then allows us to obtain
a polynomial-time hitting set construction.

\paragraph{Dimensionality reduction.}
We reduce the dimension of the ambient space by restricting to a constant-dimensional subspace (specifically, a plane) while preserving all the structural properties required for our analysis. This naturally raises the question: How should one choose such a plane?

It is not hard to show that a randomly chosen plane will work with high probability. However, using randomness would contradict the goal of this work: derandomization. Instead, we restrict to a \emph{generic plane}. That is, we treat the parameters defining the plane as indeterminates $\vecY = (Y_1, \dots, Y_\ell)$, where $\ell = \Theta(n)$, and work over the function field $\FF(\vecY)$ in place of the original base field $\FF$. The analysis is carried out symbolically over this function field, and we eventually specialize $\vecY$ to a tuple $\veca \in \FF^\ell$.

This specialization step is constructive. We identify a small collection of polynomials \( Q_1, \dots, Q_m \) in the ring \( \FF[\vecY] \), each of \emph{bounded degree}, such that any assignment \( \vecY \gets \veca \) satisfying \( Q_i(\veca) \ne 0 \) for all \( i \) ensures that the chosen plane preserves the desired properties. Verifying such an assignment reduces to testing whether each \( Q_i \) is nonzero at \( \veca \), which is a bounded-degree PIT problem and can be solved deterministically in polynomial time~\cite{KS01}.

Such generic-to-specific arguments are common in algebraic complexity. Notably, they appear in Kaltofen's work on multivariate polynomial factorization~\cite{Kal95} (see also~\cite{KSS14}), and in Sylvester--Gallai-based approaches to PIT for $\Sigma^{[k]}\Pi\Sigma\Pi^{[\delta]}$ circuits, e.g.,~\cite{Gup14}. 
The general philosophy---prove that a generic geometric object satisfies a given property, and then deduce that the property holds for a dense open set of specific instances---is also a standard principle in algebraic geometry.

Carrying out this analysis introduces several technical challenges, which we address in this paper. First, the function field $\FF(\vecY)$ is not algebraically closed, so arguments must be written more carefully and sometimes more abstractly. Second, we need to compute \grobner bases over $\FF(\vecY)$. While degree bounds on the polynomials in a \grobner basis are well-known~\cite{MM82,Du90}, we additionally require bounds on the \emph{coefficient complexity}. Specifically, we need to bound the degrees of numerators and denominators of the rational functions in $\FF(\vecY)$ that appear as coefficients. Since such bounds are not readily available in the literature, we establish them ourselves. These bounds may be of independent interest for other problems concerning algebraic pseudorandomness.

Third, although we may assume that $\FF$ is a \emph{perfect field} by passing to its algebraic closure, the function field $\FF(\vecY)$ may still be non-perfect in characteristic $p > 0$. In particular, field extensions over $\FF(\vecY)$ can be inseparable, complicating tasks such as computing radicals or primary decompositions. To address this, we adjoin $p^e$-th roots of the variables $\vecY$ for sufficiently large $e$, thereby passing to the extended function field $\FF(\vecY^{1/p^e}) := \FF(Y_1^{1/p^e}, \dots, Y_\ell^{1/p^e})$, over which the relevant extensions become separable.

\paragraph{The condition of squarefreeness.}

Finally, we explain where and why the squarefreeness condition is needed in our paper. This condition is specifically used in \cref{lem:justify}. To explain why it is needed, consider again a nonzero polynomial $F = F_0 + F_1 + F_2$, where $F_0, F_1, F_2$ are homogeneous of the same degree and are products of nonzero homogeneous polynomials, each of degree at most $\delta$. 
When $F_0$ is squarefree, then, after factoring out all common divisors of $F_0$, $F_1$, and $F_2$ and disposing of the
case $F_1+F_2=cF_0$, we can ensure that $F_0$ admits an irreducible factor $\theta$ such that $F \equiv F_1 + F_2 \not\equiv 0 \mod \theta$, allowing the approach sketched above to proceed.

In general, however, it may happen that $F \equiv F_1 + F_2 \not\equiv 0 \mod \theta^e$ but $F \equiv 0 \mod \theta$, where $\theta^e$ is the largest power of $\theta$ dividing $F_0$. The issue is that the geometric space associated with the ring $\FF[X_1, \dots, X_n]/(\theta^e)$ is not a variety, but rather a “variety with multiplicity $e$.” Normalization recovers only the reduced structure, namely the underlying variety without multiplicities, and $F$ may vanish identically on this variety.

In other words, while we can conveniently track and compare multiplicities along codimension-one subvarieties on the normalization of $V(\theta)$ (the variety defined by $\theta$), the multiplicity of $V(\theta)$ itself is lost in the process. This suggests that other techniques are needed to address the general case, and we leave this important question for future work.

\paragraph{Organization of this paper.}  
We begin with preliminaries in \cref{sec:prelim}, where we introduce our notations and review relevant notions and facts from commutative algebra and algebraic geometry.

In \cref{sec:comp}, we develop computational tools, with a focus on Gröbner bases. Specifically, we define a measure of coefficient complexity for polynomials over the function field $\FF(\vecY)$ in terms of the degrees of numerators and denominators of their coefficients. Using Gröbner bases, we analyze the constructivity of various fundamental computational tasks and derive bounds on coefficient complexity. These tasks include solving systems of linear equations, computing Gröbner bases, variable elimination, computing ring homomorphisms, and computing the generators of maximal ideals, among others.

In \cref{sec:normalization}, we use the tools from \cref{sec:comp} to give an essentially constructive normalization of an affine curve---equivalently, the computation of the integral closure of an affine curve's coordinate ring. We begin in \cref{sec:ant-notions} with basic notions from algebraic number theory, including orders, integral bases, and discriminants. In \cref{sec:ic}, we construct the normalization following the framework of Trager~\cite{Tra84}, which is a function field analog of an earlier algorithm by Ford and Zassenhaus~\cite{For78} for computing the integral closure of a number ring.

Finally, in \cref{sec:pit}, we prove our main result (\cref{thm:main-intro}). The proof relies on the computational tools from \cref{sec:comp} and the constructive normalization from \cref{sec:normalization}.

%Finally, in \cref{sec:conclusion}, we conclude the paper and discuss directions for future work.

\section{Preliminaries}\label{sec:prelim}

Denote by $[n]$ the set $\{1, 2, \dots, n\}$.
For a polynomial $f$ and a monomial $m$, let $\coeff_f(m)$ denote the coefficient of $m$ in $f$.
For polynomials $f$ and $g \neq 0$, we write $g \mid f$ to indicate that $g$ divides $f$.

It is well-known that designing deterministic black-box PIT algorithms is equivalent to constructing explicit hitting sets. We now define this notion formally.

\begin{definition}[Hitting set] A finite collection $\mathcal{H} \subseteq \FF^n$ of points is said to be a \emph{hitting set} for a polynomial $f \in \FF[X_1, \dots, X_n]$ if either $f = 0$ or $f(\veca) \neq 0$ for some $\veca \in \mathcal{H}$.
For a family $\mathcal{C} \subseteq \FF[X_1, \dots, X_n]$ of polynomials, we say $\mathcal{H} \subseteq \FF^n$ is a hitting set for $\mathcal{C}$ if it is a hitting set for every $f \in \mathcal{C}$.

Let $\varepsilon \in [0,1]$.
We say $\mathcal{H} \subseteq \FF^n$ is an \emph{$\varepsilon$-hitting set} for $f \in \FF[X_1, \dots, X_n]$ if either $f = 0$ or $\Pr_{\veca \in \mathcal{H}}[f(\veca) \neq 0] \geq 1 - \varepsilon$.
We say $\mathcal{H}$ is an $\varepsilon$-hitting set for $\mathcal{C}$ if it is an $\varepsilon$-hitting set for every $f \in \mathcal{C}$. \end{definition}

We focus on constructing explicit hitting sets, although the stronger notion of $\varepsilon$-hitting sets will also be used.
Given a hitting set $\mathcal{H}$ for polynomials of degree at most $d$,
one can boost it to an $\varepsilon$-hitting set for any $\varepsilon>0$,
assuming that $\FF$ is large enough:
For $|\mathcal{H}|\geq 2$, interpolate a polynomial curve $C$ of degree
at most $|\mathcal{H}|-1$ through $\mathcal{H}$, and choose more than
$d(|\mathcal{H}|-1)/\varepsilon$ distinct points on $C$.
The resulting set is an $\varepsilon$-hitting set.
The case $|\mathcal{H}|=1$ is immediate.

\subsection{Commutative Algebra}

All rings are assumed to be commutative with unity. The algebraic closure of a field $\FF$ is denoted by $\overline{\FF}$.

The set of $n\times m$ matrices over $A$ is denoted by $A^{n\times m}$, or more generally by $A^{S\times T}$, where $S$ and $T$ are the index sets for rows and columns, respectively. We often use $M[i,j]$ to denote the $(i,j)$-th entry of a matrix $M$. 

The ideal of a ring $A$ generated by a set $S\subseteq A$ is denoted $\ideal{S}$. And the ideal of $A$ generated by $f_1,\dots,f_k\in A$ is denoted $\ideal{f_1,\dots,f_k}$. 

\paragraph{Prime, maximal, and radical ideals.}

A proper ideal $\pp$ of a ring $A$ is a \emph{prime ideal} if $ab\in \pp \implies a\in\pp$ or $b\in \pp$. Equivalently, $\pp$ is prime if $A/\pp$ is an integral domain.

A proper ideal $\pp$ of $A$ is a \emph{maximal ideal} if it is maximal among all ideals properly contained in $A$.
Equivalently, $\pp$ is maximal if $A/\pp$ is a field. All maximal ideals are prime.

For an ideal $\pp$ of $A$, the \emph{radical} of $\pp$ is the ideal $\sqrt{\pp}:=\{a\in A: a^k\in \pp \text{ for some } k\geq 1\}$.
An ideal $\pp$ is \emph{radical} if $\pp=\sqrt{\pp}$. Prime ideals and maximal ideals are both radical.

\paragraph{Absolute irreducibility.} A polynomial $f\in\FF[\vecX]$ is \emph{absolutely irreducible} over $\FF$ if it is irreducible over the algebraic closure $\overline{\FF}$ of $\FF$. This is equivalent to $f$ being irreducible over every algebraic extension of $\FF$. By convention, the zero polynomial is not considered irreducible.

\paragraph{Separability.}

A polynomial $f(X)\in\FF[X]$ is said to be \emph{separable} over $\FF$ if its roots are distinct in the algebraic closure of $\FF$. 
This is equivalent to the statement that $f(X)$ and its derivative $f'(X)$ are coprime.
An algebraic element over $\FF$ is said to be separable over $\FF$ if its minimal polynomial over $\FF$ is separable.
The set of separable elements in an algebraic extension $\KK/\FF$ forms a subfield of $\KK$ containing $\FF$, called the \emph{separable closure} of $\FF$ in $\KK$. 
An algebraic extension $\KK/\FF$ is separable if the separable closure of $\FF$ in $\KK$ equals $\KK$. Otherwise, it is inseparable.

A field $\FF$ is \emph{perfect} if every irreducible polynomial over $\FF$ is separable. Examples of perfect fields include fields of characteristic zero, finite fields,  and algebraically closed fields. 

An irreducible polynomial can be inseparable over non-perfect fields, which must have positive characteristic. We will need the following lemma to address this issue.

\begin{lemma}\label{lem:pth-root}
Suppose $\FF$ is a field of characteristic $p>0$, and $\KK/\FF$ is a finite extension. Let $e\geq 0$ be the largest integer such that $p^{e}$ divides $[\KK:\FF]$.
For $i\geq 0$, let $\FF^{(i)}$ be the field $\{a^{1/p^i}: a\in\FF\}$.
Then for $e'\geq e$, every $b\in\KK$ is separable over $\FF^{(e')}$.
\end{lemma}

\begin{proof}
Let \( b \in \KK \), and let \( f(X) \) be its minimal polynomial over \( \FF \). Then there exists \( e_0 \leq e' \) such that over \( \FF^{(e_0)} \), we may write
\[
f(X) = g(X)^{p^{e_0}},
\]
such that \( g \in \FF^{(e_0)}[X] \) has a monomial whose degree is not divisible by \( p \). In particular, the derivative \( g' \) of \( g \) is nonzero. It suffices to show that \( b \) is separable over \( \FF^{(e_0)} \). Since \( b \) is a root of \( g \), and \( g' \neq 0 \), it remains only to show that \( g \) is irreducible over \( \FF^{(e_0)} \).

Assume, for contradiction, that \( g(X) = u(X)v(X) \) for some nonconstant \( u(X), v(X) \in \FF^{(e_0)}[X] \). Then \( f \) factors over \( \FF \) as
$f(X) = U(X) V(X)$ where $U(X)=u(X)^{p^{e_0}}$ and $V(X)=v(X)^{p^{e_0}}$,
contradicting the irreducibility of \( f \) over \( \FF \).
\end{proof}

\paragraph{Localization.} Localization is a construction in algebra that allows us to formally invert a chosen set of elements in a ring, effectively turning them into units. For example, $\QQ$ is a localization of $\ZZ$, where every nonzero integer has been made invertible. We now give the formal definition.

\begin{definition}[Localization]\label{def:localization}
Let $M$ be a module over a ring $A$. Let $S$ be a \emph{multiplicatively closed subset} of $A$, i.e., it holds that $1\in S$ and $ab\in S$ for $a,b\in S$.
Define $S^{-1} M$ to be the set of representations of pairs $(a,s)\in M\times S$ under the equivalence relation
\[
(a,s)\equiv (b,t)\quad \Longleftrightarrow \quad (at-bs)u=0 \text{ for some } u\in S.
\]
Write $a/s$ or $\frac{a}{s}$ for $(a,s)$. Call $S^{-1}M$ the \emph{localization} of $M$ with respect to $S$.

$S^{-1} A$ is a ring equipped with addition $(a/s)+(b/t)=(at+bs)/(st)$ and multiplication $(a/s)\cdot (b/t)=(ab)/(st)$.
And $S^{-1} M$ is an $S^{-1} A$-module equipped with addition $(a/s)+(b/t)=(at+bs)/(st)$ and scalar multiplication $(a/s)\cdot (b/t)=(ab)/(st)$. 
\end{definition}

Intuitively, $S^{-1}A$ is the ring obtained from $A$ by making the elements in $S$ invertible.
Note that if $0\not\in S$ and $A$ is an integral domain, the condition that $(at-bs)u=0 \text{ for some } u\in S$ in \cref{def:localization} is equivalent to $at-bs=0$ since $A$ does not have a nonzero zerodivisor.

\begin{definition}\label{def:localizaiton-examples}
When $S=\{1,f,f^2,\dots\}$ for some $f\in A$, denote $S^{-1}M$ by $M_f$. When $S=A\setminus\pp$ for some prime ideal $\pp$ of $A$, denote $S^{-1}M$ by $M_{\pp}$. When $A$ is an integral domain and $\pp$ is the zero ideal of $A$, denote $A_{\pp}$ by $\fr{A}$, called the \emph{field of fractions} of $A$.
\end{definition}

We also have the following fact. See, e.g., \cite{Art18} for a proof.
\begin{fact}\label{fact:inverse}
Let $A$ be a ring and $f\in A$. Then $A_f$ is isomorphic to $A[X]/\ideal{X f-1}$ via the map that sends $a/f^i$ to $aX^i+\ideal{Xf-1}$ for $a\in A$ and $i\in \NN$, with the inverse map sending $X+\ideal{Xf-1}$ to $1/f$.
\end{fact}

\paragraph{Local rings.} A ring $A$ is \emph{local} if it has a unique maximal ideal $\mm$. For a ring $A$ and a prime ideal $\pp$ of $A$, the ring $A_\pp$ (see \cref{def:localizaiton-examples}) is a local ring with the unique maximal ideal $\pp_\pp$.

The following statement is well-known.

\begin{lemma}[{\cite[Theorem~4.7]{Mat89}}]\label{lem:intersection-of-localization}
Let $A$ be an integral domain. Then
\[
A=\bigcap_{\textup{prime ideal }\pp\subseteq A} A_\pp=\bigcap_{\textup{maximal ideal }\mm\subseteq A} A_\mm,
\]
where the intersections are taken within $\fr{A}$.
\end{lemma}

\paragraph{Krull dimension.} The \emph{(Krull) dimension} of a ring $A$, denoted by $\dim A$, is the supremum of the lengths of chains of prime ideals $\mathfrak{p}_0\subsetneq\cdots\subsetneq\mathfrak{p}_\ell$ of $A$. 

If $A$ is a finitely generated integral domain over a field $\KK$, then $\dim A$ equals the transcendence degree of $\fr{A}$ over $\KK$.
Under the same assumption, $\dim A=\dim A_\mm$ holds for every maximal ideal $\mm$ of $A$ \cite[Corollary~11.27]{AM94}.

We say an ideal $I$ of a ring $A$ has dimension $k$ if $\dim(A/I)=k$.

\paragraph{Integrality and finiteness.}
Let $A\subseteq B$ be commutative rings. We say $b\in B$ is \emph{integral} over $A$ if there exists a monic polynomial $f(X)\in A[X]$ such that $f(b)=0$.
We say $B$ is integral over $A$ if every $b\in B$ is integral over $A$.
The set of elements in $B$ integral over $A$ is called the \emph{integral closure} of $A$ in $B$, and is a ring \cite{AM94}.
The integral closure of an integral domain $A$ in its field of fractions $\fr{A}$ is simply called the integral closure of $A$, and denoted $\widetilde{A}$.

\begin{lemma}[{\cite[Corollary~5.4]{AM94}}]\label{lem:integral-transitivity}
Suppose $C$ is integral over $B$, and $B$ is integral over $A$. Then $C$ is integral over $A$.
\end{lemma}

\begin{corollary}\label{cor:integral-subring}
Suppose $A\subseteq B\subseteq C$ are commutative rings and $B$ is integral over $A$. Then $c\in C$ is integral over $A$ if and only if it is integral over $B$. In other words, the integral closure of $A$ in $C$ equals the integral closure of $B$ in $C$.
\end{corollary}
\begin{proof}
Suppose $c$ is integral over $A$. Then as $A\subseteq B$, by definition, $c$ is also integral over $B$. 

To see the converse, let $C^*$ be the integral closure of $B$ in $C$. Suppose $c$ is integral over $B$, i.e., $c\in C^*$. As $C^*$ is integral over $B$ and $B$ is integral over $A$, by \cref{lem:integral-transitivity}, $C^*$ is integral over $A$. So $c$ is integral over $A$. 
\end{proof}

A module or algebra $M$ over a ring $A$ is said to be \emph{finite} over $A$ if $M$ is a finitely generated $A$-module. Finiteness is closely related to integrality, as indicated by the following lemma.

\begin{lemma}[{\cite[Proposition~5.1 and Corollary~5.2]{AM94}}]\label{lem:finite-vs-integral}
A finitely generated algebra over a ring
$A$ is a finite module over $A$ iff it is integral over $A$.
\end{lemma}

We also include the following lemma, which states that integrally closedness can be checked on the localizations at prime ideals or at maximal ideals.
\begin{lemma}[{\cite[Proposition~5.13]{AM94}}]\label{lem:integral-local}
Let $A$ be an integral domain. The following are equivalent:
\begin{enumerate}[(1)]
\item $A$ is integrally closed.
\item $A_\pp$ is integrally closed for every prime ideal $\pp$ of $A$.
\item $A_\mm$ is integrally closed for every maximal ideal $\mm$ of $A$.
\end{enumerate}
\end{lemma}

\paragraph{Regular local rings.} Let $A$ be a Noetherian local ring with the maximal ideal $\mm$. Use $\dim_k \mm/\mm^2$ to denote the dimension of $\mm/\mm^2$ as a vector space over $k = A/\mm$. We say a local ring $A$ is \emph{regular} if its Krull dimension $\dim A$ equals $\dim_k \mm/\mm^2$.

\paragraph{Discrete valuation rings.} A (nontrivial) \emph{discrete valuation} on a field $K$ is a mapping $v: K \to \ZZ\cup\{\infty\}$ satisfying $v(K^\times)\neq \{0\}$ and the following conditions:
\begin{enumerate}[(1)]
\item $v(xy) = v(x) + v(y)$,
\item $v(x + y) \geq \min (v(x), v(y))$, and
\item $v(x)=\infty \Longleftrightarrow x=0$, 
\end{enumerate}
for all $x,y\in K$.
We say a discrete valuation $v$ is \emph{normalized} if $v(K^\times)=\ZZ$.

Given a discrete valuation $v: K \to \ZZ\cup\{\infty\}$, the \emph{valuation ring} of $v$ is defined to be $\{x\in K: v(x)\geq 0\}$.
An integral domain $A$ is a \emph{discrete valuation ring} if there is a discrete valuation $v$ on $\fr{A}$ such that $A$ is the valuation ring of $v$. 

We include the following useful characterizations of discrete valuation rings.

\begin{lemma}[{\cite[Proposition~9.2]{AM94}}]\label{eq:equivalence}
Let $A$ be a Noetherian\footnote{A ring $A$ is Noetherian if every ideal of $A$ is finitely generated. All rings considered in this paper are Noetherian.} local domain of dimension one with the maximal ideal $\mm$.
Then the following are equivalent:
\begin{enumerate}[(1)]
\item $A$ is a discrete valuation ring; 
\item $A$ is integrally closed;
\item $A$ is a regular local ring;
\item $\mm$ is a principal ideal;
\item There exists $t \in A$ such that every nonzero ideal of $A$ is of the form $\ideal{t^k}$ for some $k \geq 0$.
\end{enumerate}
\end{lemma}
We call such $t$ a \emph{uniformizer}. 
Given a uniformizer $t\in A$, every nonzero element $r \in A$ can be written as $r = t^k u$ where $u$ is a unit; the (normalized) valuation $v$ on $A$ is determined by $v(t^k u) = k$.

\paragraph{Valuation at a point.}
Let $A$ be a (Noetherian) integrally closed domain of dimension one, not necessarily local. Let $\mm$ be a maximal ideal of $A$. By \cref{lem:integral-local}, the localization $A_\mm$ is an integrally closed local domain of dimension one.
So by \cref{eq:equivalence}, it is also a regular local ring and a discrete valuation ring. We often denote the corresponding (normalized) discrete valuation by $\ord_\mm(\cdot)$.

From a geometric point of view, the maximal ideal $\mm$ corresponds to a point $p$, and $\ord_\mm(f)$ indicates the order of zero or pole of $f$ at $p$. If $\ord_\mm(f) \geq 0$, then $\ord_\mm(f)$ is the order of zero of $f$ at $p$; otherwise, $-\ord_\mm(f)$ is the order of the pole of $f$ at $p$.

\paragraph{Base change.}

Let $A$ be a ring. For an $A$-module $M$ and an $A$-algebra $B$, their tensor product $M \otimes_A B$ over $A$ is defined to
be the $A$-module generated by the set of elements 
$\{a\otimes b : a \in M, b \in B\}$ subject to the $A$-bilinear relations $a \otimes b + a' \otimes b = (a + a') \otimes b$, 
$a \otimes b + a \otimes b' = a \otimes (b + b')$, 
and $c(a \otimes b) = (ca) \otimes b = a \otimes (cb)$ 
for $a, a'\in M$, $b, b' \in B$,
and $c \in A$. The $A$-module $M \otimes_A B$ is also a $B$-module.  Furthermore, if $M$ is an $A$-algebra, then $M\otimes_A B$ is a $B$-algebra.

Intuitively, $M \otimes_A B$ is obtained from $M$ by changing the ring of scalars from $A$ to $B$. For example, if $M$ is an $A$-algebra $A[X_1,\dots,X_n]/\ideal{r_1,\dots,r_m}$ and $A\to B$ is a ring extension, then $M \otimes_A B$ is simply $B[X_1,\dots,X_n]/\ideal{r_1,\dots,r_m}$.

\paragraph{Noether normalization.} The Noether normalization lemma is an important lemma in commutative algebra. Roughly speaking, it states that a finitely generated algebra $R$ over a field is not too far from a polynomial ring, in the sense that $R$ is a finite module over a subring that is isomorphic to a polynomial ring. 

\begin{lemma}[Noether normalization lemma]\label{lem:noether}
Let $\KK/\FF$ be a field extension.
Let $I$ be an ideal of $A=\KK[X_1,\dots,X_n]$ such that the Krull dimension of $A/I$ is $k$. 
Then for almost all\footnote{``Almost all'' here means that there exists a nonzero polynomial $Q\in\overline{\FF}[X_{1,1},\dots,X_{k,n}]$ such that the claim holds for all $\vecc\in\FF^{kn}$ satisfying $Q(\vecc)\neq 0$.} $\vecc=(c_{1,1},\dots,c_{k,n})\in\FF^{kn}$, it holds that:
\begin{enumerate}[(1)]
\item the elements $a_1,\dots,a_k\in A/I$ defined by $a_i=(\sum_{j=1}^n c_{i,j} X_j)+I$ are algebraically independent over $\KK$, and
\item  $A/I$ is finite over $\KK[a_1,\dots,a_k]$.
\end{enumerate}
In particular, such a vector $\vecc$ exists if $\FF$ is infinite.
\end{lemma} 

See, e.g., \cite[Lemma~4.8]{GVJZ23} and its proof.

\subsection{Algebraic Geometry}

Let $\FF$ be a field.
We use $\aff^n_{\FF}$ (or simply $\aff^n$) to denote the affine $n$-space over $\FF$, and $\proj^n_{\FF}$ (or simply $\proj^n$) to denote the projective $n$-space over $\FF$.
Since $\FF$ is not necessarily algebraically closed, affine and projective spaces over $\FF$---as well as varieties defined over $\FF$ within them---must be carefully formalized, rather than viewed simply as sets of solutions.
We refer the reader to \cite{Mum99} and \cite{Vak24} for formal definitions of these objects.

An affine variety $V$ over a field $\FF$ is equipped with a ring $\FF[V]$, called the \emph{coordinate ring} of $V$. The elements of $\FF[V]$ are called \emph{regular functions} on $V$. Intuitively, these are algebraic functions that are well-defined on all of $V$. For example, the coordinate ring of $\mathbb{A}^n_\FF$ is simply $\FF[X_1, \dots, X_n]$.

The (closed) points of $V$ correspond bijectively to the maximal ideals of $\FF[V]$. For a point $p \in V$, the set of regular functions vanishing at $p$ is precisely the maximal ideal corresponding to $p$.

Any radical ideal $I$ of $\FF[\vecX]=\FF[X_1,\dots,X_n]$ defines an affine variety in $\aff^n_\FF$, which we denote by $V(I)$. 
The coordinate ring of $V(I)$ is simply $\FF[\vecX]/I$.
A regular function $f$ on $\aff^n_\FF$ can be restricted to the variety $V(I)$, yielding the regular function $f|_{V(I)}:=f+I\in \FF[\vecX]/I$ on $V(I)$. 

Define $V(S):=V(\sqrt{\ideal{S}})$ for a set $S$ and $V(f_1,\dots,f_k):=V(\sqrt{\ideal{f_1,\dots,f_k}})$ for $f_1,\dots,f_k\in \FF[\vecX]$. Affine varieties of the form $V(f)$ with a single polynomial $f$ are called (affine) hypersurfaces.

A projective space $\mathbb{P}^n_{\FF}$ has $n+1$ homogeneous coordinates $X_0, \dots, X_n$. A projective variety in $\mathbb{P}^n_{\FF}$ is defined by a set of homogeneous polynomials in $\FF[X_0,\dots,X_n]$. If the variety is defined by a single homogeneous polynomial, it is called a \emph{(projective) hypersurface}.

The space $\mathbb{P}^n_{\FF}$ can be covered by $n+1$ standard affine open subsets $U_0, \dots, U_n$, where
$U_i$ is defined as the complement of the projective hyperplane defined by $X_i=0$,
and each $U_i$ is isomorphic to $\mathbb{A}^n_\FF$. The coordinate ring of $U_i$ is $\FF\left[\frac{X_0}{X_i}, \dots, \frac{X_{i-1}}{X_i}, \frac{X_{i+1}}{X_i}, \dots, \frac{X_n}{X_i} \right]$.
More generally, a projective variety $V \subseteq \mathbb{P}^n_{\FF}$ can be covered by the affine varieties $U_0 \cap V, \dots, U_n \cap V$.

A \emph{rational function} on $\mathbb{P}^n_{\FF}$ is an expression of the form
$F = \frac{P}{Q}$,
where $P$ and $Q$ are homogeneous polynomials of the same degree in $X_0, \dots, X_n$, with $Q \neq 0$. Such a function can be restricted to the affine chart $U_i$ by substituting $X_j \mapsto \frac{X_j}{X_i}$ for all $j = 0, \dots, n$, which yields a rational function on $U_i \cong \mathbb{A}^n_{\FF}$.

\paragraph{Irreducibility.}

An (affine or projective) variety is said to be \emph{irreducible} if it cannot be written as the union of two proper closed subvarieties. Otherwise, it is  \emph{reducible}. Every variety $V$ can be uniquely expressed as a finite union of maximal irreducible subvarieties, called the \emph{irreducible components} of $V$.

Suppose $V$ is an irreducible affine variety. Then $\FF[V]$ is an integral domain. In this case, the field of fractions $\mathrm{Frac}(\FF[V])$ is called the \emph{function field} of $V$, and is denoted by $\FF(V)$.

Given an irreducible projective variety $V \subseteq \mathbb{P}^n_{\FF}$, we can restrict a rational function $F = \frac{P}{Q}$ on $\mathbb{P}^n_{\FF}$ to $V$, provided that $Q$ does not vanish identically on $V$. In fact, such a restriction can be computed locally on any nonempty open subset of $V$, such as $V \cap U_i$ for some $i$ \cite[Exercise~5.2.I]{Vak24}.

Let $p \in V$ correspond to a maximal ideal $\mm \subset \FF[V]$. A rational function $f \in \FF(V)$ is said to be \emph{regular at $p$} if $f$ lies in the local ring $\FF[V]_{\mm}$. Intuitively, this means $f$ is well-defined at (and near) the point $p$.

Indeed, we can evaluate $f \in \FF[V]_{\mm}$ at $p$ via the natural quotient map
\[
f \mapsto f + \mm_{\mm} \in \kappa_{\mm} := \FF[V]_{\mm} / \mm_{\mm} \cong \FF[V]/\mm.
\]
This value lies in the field $\kappa_{\mm}$, called the \emph{residue field} of the point $p$ (or of the ideal $\mm$). It is a finite field extension of $\FF$.

\paragraph{Normality.}

An irreducible affine variety is said to be \emph{normal} if its coordinate ring is an integrally closed domain.\footnote{More generally, a (possibly reducible) affine variety is normal if its coordinate ring is a finite product of integrally closed domains \cite[\href{https://stacks.math.columbia.edu/tag/030C}{Tag 030C}]{stacks-project}.}
An irreducible projective variety $V \subseteq \mathbb{P}^n_\FF$ is normal if each of the affine pieces $V \cap U_0, \dots, V \cap U_n$ is a normal affine variety.

\paragraph{Dimension.}
Let $V$ be an (affine or projective) variety over an algebraically closed field $\FF$. The \emph{dimension} of $V$, denoted $\dim V$, is the largest integer $d$ such that there exists a chain of irreducible closed subvarieties
$Z_0 \subsetneq Z_1 \subsetneq \cdots \subsetneq Z_d \subseteq V$, where each $Z_i$ is a proper irreducible subvariety of $Z_{i+1}$.
In other words, the dimension of $V$ is the length of the longest strictly increasing chain of irreducible closed subsets in $V$.

If $V$ is a variety over a field $\FF$ that is not algebraically closed, its dimension is defined as
$\dim V := \dim V_{\overline{\FF}}$,
where $V_{\overline{\FF}}$ denotes the base change of $V$ to the algebraic closure $\overline{\FF}$ of $\FF$.

In the case of an affine variety $V \subseteq \mathbb{A}^n_\FF$, the dimension of $V$ coincides with the Krull dimension of its coordinate ring $\FF[V]$.

Varieties of dimension one are called curves.

\paragraph{Degree.}

Suppose $V \subseteq \mathbb{A}^n_\FF$ is an irreducible affine variety over an algebraically closed field $\FF$. The \emph{degree} of $V$, denoted $\deg(V)$, is defined as the number of isolated points in the intersection of $V$ with a general affine linear subspace of codimension $d$, where $d = \dim V$.

Now suppose $V \subseteq \mathbb{P}^n_\FF$ is an irreducible projective variety over an algebraically closed field $\FF$. The degree of $V$ is defined as the number of points in the intersection of $V$ with a general linear subspace $L \subseteq \mathbb{P}^n_\FF$ of codimension $d$, where $d = \dim V$.

For a reducible variety, we define its degree to be the sum of the
degrees of all its irreducible components, including
components of different dimensions.

If $\FF$ is not algebraically closed, we define $\deg(V) := \deg(V_{\overline{\FF}})$, where $V_{\overline{\FF}}$ denotes the base change of $V$ to the algebraic closure $\overline{\FF}$.%
\footnote{The base change $V_{\overline{\FF}}$ is not necessarily a variety and must be interpreted as a \emph{scheme} \cite{Mum99,Vak24} to correctly define the degree. For example, consider the affine variety defined by $X^p - T$ in $\mathbb{A}^1_{\FF_p(T)}$. After base change to $\overline{\FF_p(T)}$, it becomes a ``point of multiplicity $p$,'' since $X^p - T = (X - T^{1/p})^p$.}

It can be shown that the degree of a point $p$ (corresponding to a maximal ideal $\mathfrak{m}$) equals the degree of its residue field extension, that is, $[\kappa_{\mathfrak{m}} : \FF]$.

\paragraph{\bezout's inequality.} 

We need the following version of \bezout’s inequality.
\begin{lemma}[{\bezout’s inequality \cite{HS80,Hei83}}]
Let $V$ and $V'$ be closed affine subvarieties of $\aff^n_\FF$.
Then $\deg(V \cap V')\leq \deg(V)\cdot \deg(V')$.
\end{lemma}

The following statement follows from \cite[Proposition~3.5]{BM93}.
\begin{lemma}\label{lem:bezout2}
Let $I=\ideal{f_1,\dots,f_k}$ be a zero-dimensional ideal of $\FF[X_1,\dots,X_n]$, and let $d=\max(\deg(f_1),\dots,\deg(f_k))$.
Then $\dim_\FF(\FF[\vecX]/I)\leq d^n$.
In particular, $[\kappa_\mm: \FF]\leq d^n$ for any maximal ideal $\mm$ containing $I$.
\end{lemma}

\section{Computational Tools}\label{sec:comp}

The main purpose of this section is to give bounds on the complexity of various algebraic procedures. The bounds will be given in terms of the degrees of the input polynomials, as well as the degrees of their coefficients $\frac{c}{c'}$, since the coefficients come from the field of rational functions $\FF(Y_1,\dots,Y_\ell)$. To do this, we develop various computational tools, with a focus on \grobner bases, and establish bounds of the form $O_{s_1,\dots,s_r}(1)$, where $s_1,\dots,s_r$ are parameters. While such bounds obscure the specific dependence on $s_1,\dots,s_r$, they suffice for our purposes, as these parameters will be treated as constants in our application. Additionally, although it is theoretically possible to make these bounds explicit, doing so would be tedious and distracting, as it involves repeated composition of intermediate results.

That said, one can verify that all bounds in this section are at most constant-height towers of exponentials in the degrees and the number of variables. In most cases, the bounds are at least doubly exponential due to the use of \grobner bases. As mentioned, this is acceptable since the degrees and number of variables will eventually be fixed constants depending only on the bottom fan-in $\delta$.

Throughout this section, let $\KK=\FF(\vecY)=\FF(Y_1,\dots,Y_\ell)$, where $\FF$ is an infinite field. 

To control the complexity of coefficients in $\KK$, we introduce the following families of subsets of $\KK$ and $\KK[\vecX]$.

\begin{definition}[Field elements and polynomials with bounded complexity]
For an integer $d \geq 0$, let $\KK_{\leq d}$ denote the set of $a\in \KK$ such that $a=\frac{c}{c'}$ for some $c\in \FF[\vecY]$ and $c'\in \FF[\vecY]\setminus\{0\}$ satisfying $\deg_{\vecY}(c),\deg_{\vecY}(c')\leq d$.

For integers $n,d,d'\geq 0$ and variables $\vecX=(X_1,\dots,X_n)$, let $\KK_{\leq d'}[\vecX]_{\leq d}$ denote the set of polynomials $f\in\KK[\vecX]$ of degree at most $d$ such that the coefficients of $f$ are all in $\KK_{\leq d'}$.
\end{definition}

\begin{lemma}\label{lem:frac-deg-bound} The following hold:
\begin{enumerate}[(1)]
\item  \label{item:deg1} 
$ca\in \KK_{\leq d}$ for $a\in \KK_{\leq d}$ and $c\in\FF$.
\item\label{item:deg2} $a+b,ab\in \KK_{\leq d + d'}$
for $a\in \KK_{\leq d}$ and $b\in \KK_{\leq d'}$.
\item\label{item:deg3} $1/a\in \KK_{\leq d}$ for $a\in \KK_{\leq d}\setminus\{0\}$.
\end{enumerate}
\end{lemma}

\begin{proof}
The claims follow straightforwardly by definition. 
\end{proof}

\begin{lemma}\label{lem:prod-deg-bound}
Let $\KK[\vecX]=\KK[X_1,\dots,X_n]$.
For $f\in \KK_{\leq d}[\vecX]_{\leq d_1}$ and $g\in\KK_{\leq d}[\vecX]_{\leq d_2}$, it holds that $fg\in \KK_{\leq d'}[\vecX]_{\leq d_1+d_2}$ for some $d'=O_{d_1+d_2,d,n}(1)$.
\end{lemma}

\begin{proof}
Each coefficient of $fg$ is the sum of $O_{d_1+d_2,n}(1)$ terms of the form $\coeff_f(m_1) \cdot \coeff_g(m_2)$, where $m_1$ is a monomial of $f$ and $m_2$ is a monomial of $g$. The claim follows by Lemma \ref{lem:frac-deg-bound}\,\eqref{item:deg2}.
\end{proof}

To address the issue of inseparable extensions, we introduce the following definitions.

\begin{definition}\label{defi:p-powers}
Let \( p = \mathrm{char}(\FF) \), and let \( e \geq 0 \).  
We use the shorthand
\[
\FF(\vecY^{1/p^e}) := \FF(Y_1^{1/p^e}, \dots, Y_\ell^{1/p^e}) \quad \text{and} \quad \FF[\vecY^{1/p^e}] := \FF[Y_1^{1/p^e}, \dots, Y_\ell^{1/p^e}].
\]
For convenience, if \( \mathrm{char}(\FF) = 0 \), we interpret \( p^e = 1 \), even though this is not formally correct.

Define \( \KK^{(e)} := \FF(\vecY^{1/p^e}) \).  
Let \( \KK^{(e)}_{\leq d} \) be the set of elements \( a \in \KK^{(e)} \) such that \( a = \frac{c}{c'} \) for some \( c \in \FF[\vecY^{1/p^e}] \), \( c' \in \FF[\vecY^{1/p^e}] \setminus \{0\} \), and
$\deg_{\vecY^{1/p^e}}(c), \deg_{\vecY^{1/p^e}}(c') \leq d$.
That is, both $c$ and $c'$ have degree at most $d$ in the variables $Y_i^{1/p^e}$.

For integers \( n, d, d' \geq 0 \) and variables \( \vecX = (X_1, \dots, X_n) \), define \(  \KK_{\leq d'}^{(e)}[\vecX]_{\leq d} \) to be the set of polynomials \( f \in \KK^{(e)}[\vecX] \) of degree at most \( d \), whose coefficients are all in \( \KK^{(e)}_{\leq d'} \).
\end{definition}

Note that the $\mathbb{F}$-isomorphism $\mathbb{K}^{(e)}\to\mathbb{K}$ sending $Y_i^{1/p^e}$ to $Y_i$ induces a bijection
\[
\KK_{\leq d'}^{(e)}[\vecX]_{\leq d} \to \KK_{\leq d'}[\vecX]_{\leq d}.
\]
In this way, statements about polynomials in \( \KK_{\leq d'}[\vecX]_{\leq d} \) apply to those in \( \KK_{\leq d'}^{(e)}[\vecX]_{\leq d} \) via this transformation. On the other hand, there is also a natural inclusion
\[
\KK_{\leq d'}[\vecX]_{\leq d} \hookrightarrow \KK_{\leq p^e d'}^{(e)}[\vecX]_{\leq d}
\]
given by the identity map \( f \mapsto f \).

\paragraph{Solving a system of linear equations.}

The following lemma bounds the complexity of a solution of a system of linear equations.

\begin{lemma}\label{lem:linsol}
Let $A=(a_{i,j})_{i\in [N], j\in [M]}\in \KK^{N\times M}$ and $\mathbf{b}=(b_1,\dots,b_N)\in \KK^N$ such that all entries of $A$ and those of $\mathbf{b}$ are in $\KK_{\leq d}$.
Suppose the system of linear equations $A\mathbf{x}=\mathbf{b}$ has a solution.
Then it has a solution $\mathbf{x}\in\KK^M$ whose entries are all in $\KK_{\leq d'}$, where $d'=O_{N,M,d}(1)$.
\end{lemma}

\begin{proof}
Suppose $A$ has rank $k$, and without loss of generality, we may assume the top-leftmost $k\times k$ minor $B$ of $A$ has full rank. Let $\widetilde{A}=(a_{i,j})_{i\in [k], j\in[M]}$ and $\widetilde{\mathbf{b}}=(b_1,\dots,b_k)$.
As $A\mathbf{x}=\mathbf{b}$ has a solution and the first $k$ rows of $A$ span the others, we know $\mathbf{x}=(x_1,\dots,x_M)$ is a solution of $A\mathbf{x}=\mathbf{b}$ if and only if it is a solution of $\widetilde{A}\mathbf{x}=\widetilde{\mathbf{b}}$. By Cramer's rule, the latter has a solution given by 
\begin{equation}\label{eq:xi}
x_i=\begin{cases}
    \frac{\det(B_i)}{\det(B)} & 1\leq i\leq k\\
    0 & k<i\leq M,
\end{cases}
\end{equation}
where $B_i$ is the matrix formed by replacing the $i$-th column of $B$ by $\widetilde{\mathbf{b}}$.
As the entries of $B$ and $B_i$ are in $\KK_{\leq d}$, by \cref{lem:frac-deg-bound}, the entries $x_i$ given by \eqref{eq:xi} are in $\KK_{\leq d'}$ for some $d'=O_{N,M,d}(1)$.
\end{proof}

The following lemma bounds the complexity of the coefficients of a factor of a polynomial $f$ in terms of its degree and the complexity of its coefficients.

%\begin{lemma}\label{lem:factor}
\begin{restatable}{lemma}{factor}\label{lem:factor}
Let $f\in \KK_{\leq d'}[\vecX]_{\leq d}$ be a nonzero polynomial in $n$ variables.
Let $g$ be a factor of $f$.
Then there exists $c\in\KK^\times$ such that $cg\in \KK_{\leq d''}[\vecX]_{\leq d}$ for some $d''=O_{d,d',n}(1)$.
\end{restatable}
%\end{lemma}

The proof of \cref{lem:factor} is deferred to \cref{sec:appendix}.

\paragraph{\grobner bases.}

For $\vecX=\{X_1,\dots,X_n\}$, denote by $\mathcal{M}_\vecX$ the set of monomials in $\vecX$.
 
Let $\preceq$ be a total order on $\mathcal{M}_\vecX$.
We say $\preceq$ is a \emph{monomial order} if (1) $m_1\preceq m_2\implies m_1 m\preceq m_2 m$, and (2) $1 \preceq m$ for all $m\in\mathcal{M}_\vecX$.
Write $m_1\prec m_2$ if $m_1\preceq m_2$ and $m_1\neq m_2$.

A monomial order is \emph{degree-compatible} if $m_1\preceq m_2\implies \deg(m_1)\leq \deg(m_2)$. Examples of degree-compatible monomial orders include graded lexicographic and graded reverse lexicographic orders.

Fix a monomial order $\preceq$, for $0\neq f\in \KK[\vecX]$, the \emph{leading monomial} of $f$, denoted by $\lm{f}$, is the  monomial appearing in $f$ that is greatest under $\preceq$. Its coefficient is called the \emph{leading coefficient} of $f$ and denoted $\lc{f}$.
The term $\lc{f}\cdot \lm{f}$ is called the \emph{leading term} of $f$ and denoted $\lt{f}$.
Define $\lm{0}=\lc{0}=\lt{0}=0$.

For a set $S\subseteq \KK[\vecX]$, denote by $\lt{S}$ the ideal of $\KK[\vecX]$ generated by the set $\{\lt{f}: f\in S\}$.
It is called
\emph{ideal of leading terms} of $S$, also known as the \emph{initial ideal} of $S$. 

\begin{definition}[\grobner basis]
A finite generating set $G$ of an ideal $I$ of $\KK[\vecX]$ is said to be a \grobner basis for $I$ if $\lt{G}=\lt{I}$.
\end{definition}

\paragraph{Degree bounds for \grobner bases.} 
 
We have the following degree bound on \grobner bases.
\begin{theorem}[{\cite[Corollary~8.3]{Du90}}]\label{thm:grobner-degree-bound}
Let $\KK=\FF(\vecY)$.
Let $I$ be an ideal of $\KK[\vecX]=\KK[X_1,\dots,X_n]$ generated by polynomials of degree at most $d$.
Then for any monomial order $\preceq$, $I$ has a \grobner basis with respect to $\preceq$ consisting only of polynomials of degree at most $2\left(\frac{d^2}{2}+d\right)^{2^{n-1}}$.
\end{theorem}

We also need the following degree bound for the ideal membership problem. For a proof, see \cite[Appendix]{MM82}.

\begin{lemma}[{\cite{Her26, MM82}}]\label{lem:ideal-membership}
Let $f_1,\dots,f_k\in \KK[\vecX]=\KK[X_1,\dots,X_n]$ be polynomials of degree at most $d$. Let $f\in\ideal{f_1,\dots,f_k}$ be a polynomial of degree at most $d'$.
Then $f=\sum_{i=1}^k h_i f_i$ for some $h_1,\dots,h_k\in \KK[\vecX]$ of degree at most $d'+(kd)^{2^n}$.
\end{lemma}

Now, we strengthen \cref{thm:grobner-degree-bound} and \cref{lem:ideal-membership} to give the following lemma. In the case where the coefficients of $f_1, ..., f_k$ come from $\KK=\FF(\vecY)$, we establish complexity bounds on the \grobner basis of an ideal $\ideal{f_1, ..., f_k}$ by also bounding the degree of the coefficients of such a \grobner basis.  

\begin{lemma}\label{lem:grobner-bounds}
Suppose $\KK=\FF(\vecY)$.
Let $I$ be an ideal of $\KK[\vecX]=\KK[X_1,\dots,X_n]$ generated by $f_1,\dots,f_k\in \KK_{\leq d'}[\vecX]_{\leq d}$.
Then for any monomial order $\preceq$, $I$ has a \grobner basis with respect to $\preceq$ consisting only of polynomials in $\KK_{\leq d''}[\vecX]_{\leq D}$, where $D=2\left(\frac{d^2}{2}+d\right)^{2^{n-1}}$ and $d''=O_{d,d',n}(1)$.
\end{lemma}

\begin{proof}
Note that we may assume $k\leq \binom{n+d}{d}=O_{d,n}(1)$.
By \cref{thm:grobner-degree-bound}, $I$ has a \grobner basis $G$ consisting of polynomials of degree at most $D$. 
Consider nonzero $g\in G$ and let $m_g=\lm{g}$. By replacing $g$ with $\frac{1}{\lc{g}} g$, we may assume $\lt{g}=m_g$.
By \cref{lem:ideal-membership}, we have $g=\sum_{i=1}^k h_i f_i$ for some $h_1,\dots,h_k$ of degree at most $D':=D+(kd)^{2^n}=O_{d,n}(1)$. 

For every integer $a\in \NN$, let $S_a$ be the set of monomials of degree at most $a$.
For each monomial $m'$, as $g=\sum_{i=1}^k h_i f_i$, we have 
\begin{equation}\label{eq:coeff-grobner}
\coeff_{g}(m')=\sum_{i=1}^k\sum_{m\in S_{D'}: m|m'} \coeff_{h_i}(m)\coeff_{f_i}(m'/m).
\end{equation}
The condition that $\lt{g}=m_g$ is equivalent to that  $\coeff_g(m_g)=1$ and $\coeff_g(m')=0$ for all $m'\in S_{d+D'}$ strictly greater than $m_g$ with respect to $\preceq$. (Here, we only need to consider monomials $m'\in S_{d+D'}$ since $\deg(f_i)\leq d$ and $\deg(h_i)\leq D'$ for $i\in [k]$.) We formulate this condition as a system of linear equations as follows. 
Let $S$ be the set of $m'\in S_{d+D'}$ satisfying either $m_g\preceq m'$ or $\deg(m')>\deg(g)$. Let $T=[k]\times S_{D'}$. 
Define a matrix $A\in\KK^{S\times T}$ by
\[
A[m',(i,m)]=\begin{cases}
    \coeff_{f_i}(m'/m) & m|m'\\
    0 & \text{otherwise}.
\end{cases}
\]
And define the column vector $\mathbf{b}\in\KK^S$ by
\[
\mathbf{b}[m']=\begin{cases}
    1 & m'=m_g\\
    0 & \text{otherwise}.
\end{cases}
\]
By assumption, we have $A[m',(i,m)],\mathbf{b}[m']\in \KK_{\leq d'}$ for all $(m',(i,m))\in S\times T$. By \cref{lem:linsol}, the system of linear equations $A\mathbf{x}=\mathbf{b}$ has a solution $(c_{i,m})_{(i,m)\in T}$ whose entries are all in $\KK_{\leq d_0}$ for some 
\[
d_0=O_{|S|, |T|, d'}(1)=O_{d,d',n}(1),\]
where the last equality holds since
$|S|\leq |S_{d+D'}|=\binom{n+d+D'}{d+D'}=O_{d,n}(1)$ and  $|T|=k|S_{D'}|=k\binom{n+D'}{D'}=O_{d,n}(1)$.
Let $\widetilde{h}_i:=\sum_{m\in S_{D'}} c_{i,m} m$ for $i\in[k]$.
As $A \mathbf{x}=\mathbf{b}$, we have $\lt{\sum_{i=1}^k \widetilde{h}_i f_i}=m_g$ and $\deg(\sum_{i=1}^k \widetilde{h}_i f_i)\leq \deg(g)$, where the second condition holds since $S$ includes all monomials in $S_{d+D'}$ of degree greater than $\deg(g)$, and their coefficients are therefore forced to be zero.

By \cref{lem:frac-deg-bound}, the coefficients of $\sum_{i=1}^k \widetilde{h}_i f_i$ belong to $\KK_{\leq d''}$ for some $d''=O_{d_0,d,d',k,n}(1)=O_{d,d',n}(1)$.
Replacing $h_i$ by $\widetilde{h}_i$ for $i\in[k]$, and replacing $g$ by $\sum_{i=1}^k \widetilde{h}_i f_i$, do not change the leading term of $g$ or increase its degree. Performing this replacement for each $g\in G$ preserves the fact that $G$ is a \grobner basis for $I$, since $\lt{G}$ remains unchanged. 
After the replacement, the new \grobner basis satisfies the requirement of the lemma.
\end{proof}

Similarly, we establish degree bounds on the coefficients of an instance of the ideal membership problem.
\begin{lemma}\label{lem:ideal-membership-coeff}
Let $\KK[\vecX]=\KK[X_1,\dots,X_n]$ and $f_1,\dots,f_k\in \KK_{\leq D}[\vecX]_{\leq d}$. Let $f\in\ideal{f_1,\dots,f_k}$ be a polynomial in $\KK_{\leq D'}[\vecX]_{\leq d'}$.
Then $f=\sum_{i=1}^k h_i f_i$ for some $h_1,\dots,h_k\in \KK_{\leq D''}[\vecX]_{\leq d''}$ where $d'':=d'+(kd)^{2^n}$ and $D''=O_{n, d,d',D,D'}(1)$.
\end{lemma}

\begin{proof}
We may assume $k\leq \binom{n+d}{d}$.
By \cref{lem:ideal-membership}, there exist polynomials $h_1,\dots,h_k$ of degree at most $d''$
such that $f=\sum_{i=1}^k h_i f_i$. Viewing the coefficients of $h_1,\dots,h_k$ as variables, we may build a system of linear equations over $\KK$ that expresses the relation $f=\sum_{i=1}^k h_i f_i$. It has at most $k\binom{n+d''}{d''}=O_{n,d,d'}(1)$
variables and at most
$\binom{n+d+d''}{d+d''}=O_{n,d,d'}(1)$ linear equations,
obtained by comparing all coefficients through degree $d+d''$. The coefficients of these linear equations live in $\KK_{\leq\max\{D, D'\}}$. So by \cref{lem:linsol}, we may choose $h_1,\dots,h_k$ such that they live in $\KK_{\leq D''}[\vecX]_{\leq d''}$ for some $D''=O_{n,d,d',D,D'}(1)$.
\end{proof}

In the following, we introduce the \emph{elimination order} on monomials in a polynomial ring. We use it to obtain degree bounds on the \grobner basis of the intersections of ideals. 
\paragraph{Elimination of variables.}

Let $\preceq_\vecX$ and $\preceq_\vecZ$ be monomial orders on $\mathcal{M}_\vecX$ and $\mathcal{M}_\vecZ$, respectively.
Denote by $\preceq=(\preceq_\vecX, \preceq_\vecZ)$ the following monomial order on the monomials in both $\vecX$ and $\vecZ$:
\[
m_1 m_2 \prec m_1' m_2' \Longleftrightarrow 
    \begin{cases}
      m_1\prec m_1'\\
      \text{or}\\
     m_1=m_1' \text{ and } m_2\prec m_2'  
    \end{cases}  
    \quad \text{for $m_1,m_1'\in\mathcal{M}_\vecX$ and $m_2,m_2'\in\mathcal{M}_\vecZ$}.
\]
Call $\preceq$ an \emph{elimination order} for $\vecX$.\\

The following is a well-known fact about eliminating variables using an elimination order. See, e.g., \cite[Theorem~2.3.4]{AL94}.

\begin{lemma}[Elimination of variables]\label{lem:elim-ideal}
Let $\preceq_\vecX$ and $\preceq_\vecZ$ be monomial orders on $\mathcal{M}_\vecX$ and $\mathcal{M}_\vecZ$, respectively.
Let $\preceq=(\preceq_\vecX, \preceq_\vecZ)$.
Let $I$ be an ideal of $\KK[\vecX, \vecZ]$. Let $G$ be a \grobner basis for the ideal $I$ with respect to $\preceq$.
Then $G \cap \KK[\vecZ]$ is a \grobner basis for the ideal $I \cap \KK[\vecZ]$ of $\KK[\vecZ]$ with respect to $\preceq_\vecZ$.
\end{lemma}

We combine \cref{lem:grobner-bounds} and \cref{lem:elim-ideal} to give the following statement, which provides degree bounds for the \grobner basis of ideals of the form $I\cap \KK[\vecZ]$, as well as that of its coefficients.  

\begin{lemma}\label{lem:elimination}
Let $I$ be an ideal of $\KK[\vecX,\vecZ]=\KK[X_1,\dots,X_n,Z_1,\dots,Z_m]$ generated by polynomials $f_1,\dots,f_k\in \KK_{\leq d'}[\vecX,\vecZ]_{\leq d}$. Then for any monomial order $\preceq_\vecZ$ on $\mathcal{M}_\vecZ$, $I\cap \KK[\vecZ]$ has a \grobner basis with respect to $\preceq_\vecZ$ consisting only of polynomials in $\KK_{\leq D'}[\vecZ]_{\leq D}$, where $D=O_{d,n+m}(1)$ and $D'=O_{d,d',n+m}(1)$.
\end{lemma}

\begin{proof}
Let $\preceq_\vecX$ be a monomial order on $\mathcal{M}_\vecX$ and let $\preceq =(\preceq_\vecX,\preceq_\vecZ)$.
Let $G$ be a \grobner basis for $I$ with respect to $\preceq$.
By \cref{lem:grobner-bounds}, we may assume that $G$ consists only of polynomials in $\KK_{\leq D'}[\vecX,\vecZ]_{\leq D}$, where $D=2\left(\frac{d^2}{2}+d\right)^{2^{n+m-1}}=O_{d,n+m}(1)$ and $D'=O_{d,d',n+m}(1)$.
By \cref{lem:elim-ideal}, $G\cap\KK[\vecZ]$ is a \grobner basis for $I\cap \KK[\vecZ]$ with respect to $\preceq_\vecZ$. The lemma follows.
\end{proof}

\paragraph{Reduction algorithm.}

The reduction algorithm, also known as the division algorithm, generalizes both row reduction in Gaussian elimination and long division of univariate polynomials.
Running this algorithm on a polynomial $f$ with respect to a \grobner basis yields a unique ``remainder'' of $f$.

\begin{definition}
Let $G$ be a subset of $\KK[\vecX]$.
$f\in \KK[\vecX]$ is said to be \emph{reducible modulo $G$} if $f$ has a nonzero term that is in $\lt{G}$. Otherwise, we say $f$ is \emph{reduced modulo $G$}.
\end{definition}

In the following lemma, we bound the complexity of a polynomial $f$ modulo a \grobner basis $G$. We again give bounds on the degree of the remainder $f_G$ as well as that of $f_G$'s coefficients. 

\begin{lemma}\label{lem:reduction-alg}
Let $f\in\KK[\vecX]=\KK[X_1,\dots,X_n]$ and let $G=\{g_1,\dots,g_k\}$ be a \grobner basis for an ideal $I$ of $\KK[\vecX]$ with respect to a monomial order $\preceq$. 
Then:
\begin{enumerate}[(1)]
\item\label{item:reduction1} There exists a unique remainder $f_G\in\KK[\vecX]$ such that $f_G$ is reduced mod $G$ and $f-f_G=\sum_{i=1}^k h_i g_i\in I$, where $h_1,\dots,h_k\in \KK[\vecX]$.
\item\label{item:reduction2} Suppose $f,g_1,\dots,g_k\in \KK_{\leq d'}[\vecX]_{\leq d}$. Further assume that for $i\in [k]$, $\deg(\lt{g_i})=\deg(g_i)$, which holds if $\preceq$ is degree-compatible.
Then in \cref{item:reduction1}, $f_G\in \KK_{\leq d''}[\vecX]_{\leq d}$ and $h_1,\dots,h_k$ may be chosen in $\KK_{\leq d''}[\vecX]_{\leq d}$ as well, where $d''=O_{d,d',n}(1)$.
\end{enumerate}
\end{lemma}

The reduction algorithm outputting $f_G$ and $h_1,\dots,h_k$ is given as follows.

\begin{algorithm}[H]
\caption{The reduction algorithm}\label{alg:reduction}
\begin{algorithmic}
\Require $f$, $G=\{g_1,\dots,g_k\}$
\Ensure $f_G$, $h_1,\dots,h_k$
\State $f_G\gets f$; $h_1,\dots,h_k\gets 0$
\While{$f_G$ has a term $T$ divisible by $\lt{g_i}$ for some $i\in [k]$}
\State $f_G\gets f_G - \frac{T}{\lt{g_i}} g_i$
\State $h_i\gets h_i+\frac{T}{\lt{g_i}}$
\EndWhile
\State \Return $f_G$, $h_1,\dots,h_k$
\end{algorithmic}
\end{algorithm}

\cref{alg:reduction} terminates in general, even without the assumption that $\deg(\lt{g_i}) = \deg(g_i)$ for $i \in [k]$. This fact follows from Hilbert's basis theorem or Dickson's lemma~\cite{BW98}. We make the additional assumption that $\deg(\lt{g_i}) = \deg(g_i)$ for $i \in [k]$, which leads to a simpler proof that bounds the degrees of the coefficients of $f_G, h_1, \dots, h_k$.

\begin{proof}[Proof of \cref{lem:reduction-alg}]
\cref{item:reduction1} is standard (see, e.g., \cite{AL94}).

We claim $\deg(f_G)\leq \deg(f)$ and $\deg(h_ig_i)\leq \deg(f)$ for $i\in [k]$. This follows from the following induction: At the beginning of the algorithm, these bounds hold since $f_G=f$ and $h_1,\dots,h_k=0$.
In each iteration, the degree of $\Delta:=\frac{T}{\lt{g_i}} g_i$ is $\deg(T)-\deg(\lt{g_i})+\deg(g_i)=\deg(T)\leq\deg(f_G)\leq \deg(f)$. We subtract $\Delta$ from $f_G$, which does not increase the degree of $f_G$ since the degrees of the newly added monomials are bounded by $\deg(\Delta)\leq \deg(f_G)$.
Similarly, we add $\Delta/g_i$ to $h_i$, which preserves the property that $\deg(h_ig_i)\leq \deg(f)$ since $\deg((\Delta/g_i)\cdot g_i)=\deg(\Delta)\leq \deg(f)$. So the degree bounds hold during and at the end of the algorithm.

Now we prove \cref{item:reduction2}. By the degree bounds just proved, we have
$\deg(f_G),\deg(h_1),\dots,\deg(h_k)\leq\deg(f)\leq d$.

We have 
\[
f_G=f-\sum_{i=1}^k h_i g_i=f-\sum_{i=1}^k\sum_{m\in S_i} \coeff_{h_i}(m) m g_i,
\]
where $S_i$ is the set of monomials $m$ such that $\deg(m\cdot g_i)\leq\deg(f)$.
Therefore, for each monomial $m'$ of degree at most $\deg(f)$, we have
\begin{equation}\label{eq:coeff-f-reduction}
\coeff_{f_G}(m')=\coeff_{f}(m')-\sum_{i=1}^k\sum_{m\in S_i: m|m'} \coeff_{h_i}(m)\coeff_{g_i}(m'/m).
\end{equation}

By keeping only one polynomial for each distinct leading monomial in $G$,
we may assume that $k\leq\binom{n+d}{d}$.
Let $S$ be the set of monomials $m'$ of degree at most $\deg(f)$
such that $m'\in\lt{G}$, and let
$T=\{(i,m):i\in[k],\ m\in S_i\}$.
Note that $|S|\leq \binom{n+d}{d}$ and $|T|=\sum_{i=1}^k |S_i|\leq \binom{n+d}{d}^2$.
By definition, a polynomial $g$ satisfying $\deg(g)\leq \deg(f)$ is reduced mod $G$ if and only if $\coeff_{g}(m')=0$ for all $m'\in S$. 
So the polynomial $f_G=f-\sum_{i=1}^k h_i g_i$ with $h_i=\sum_{m\in S_i} \coeff_{h_i}(m) m$ is reduced mod $G$ if and only if the column vectors $\mathbf{x}=(\coeff_{h_i}(m))_{(i,m)\in T}$ and $\mathbf{b}=(\coeff_{f}(m'))_{m'\in S}$ satisfy $A\mathbf{x}=\mathbf{b}$, where the matrix $A\in \KK^{S\times T}$ is given by
\[
A[m',(i,m)]=\begin{cases}
    \coeff_{g_i}(m'/m) & m|m'\\
    0 & \text{otherwise}.
\end{cases}
\]
By assumption, we have $A[m',(i,m)],\mathbf{b}[m']\in \KK_{\leq d'}$ for all $(m',(i,m))\in S\times T$. By \cref{lem:linsol}, we may choose the coefficients of $h_1,\dots,h_k$ to be in $\KK_{\leq d_0}$ for some $d_0=O_{|S|, |T|, d'}(1)=O_{d,d',n}(1)$.

By \eqref{eq:coeff-f-reduction} and \cref{lem:frac-deg-bound}, for such $h_1,\dots,h_k$, the coefficients of $f_G$ are all in $\KK_{\leq d_1}$, where $d_1=O_{d,d',n}(1)$.
Choose $d'':=\max(d_0,d_1)=O_{d,d',n}(1)$.
Then $f_G,h_1,\dots,h_k\in \KK_{\leq d''}[\vecX]_{\leq d}$ by definition.
\end{proof}

\begin{definition}\label{def:remainder}
$f_G$ in \cref{lem:reduction-alg} is called the \emph{remainder} of $f$ modulo $G$ with respect to $\preceq$.
\end{definition}

The following is a variant of \cref{lem:reduction-alg} for elimination orders. It states that the degree of the remainder $f_G$ and the degrees of its coefficients are bounded. 

\begin{lemma}\label{lem:reduction-alg2}
Let $f\in\KK[\vecX,\vecZ]=\KK[X_1,\dots,X_n,Z_1,\dots,Z_m]$.
Let $\preceq_\vecX$ and $\preceq_\vecZ$ be degree-compatible monomial orders on $\mathcal{M}_\vecX$ and on $\mathcal{M}_\vecZ$, respectively, and let $\preceq=(\preceq_\vecX, \preceq_\vecZ)$.
Let $G=\{g_1,\dots,g_k\}$ be a \grobner basis for an ideal $I$ of $\KK[\vecX,\vecZ]$ with respect to $\preceq$. 
Suppose $f,g_1,\dots,g_k\in \KK_{\leq d'}[\vecX,\vecZ]_{\leq d}$.
Then the remainder $f_G$ of $f$ is in $\KK_{\leq d_2}[\vecX,\vecZ]_{\leq d_1}$, where $d_1=O_{d,n+m}(1)$ and $d_2=O_{d,d',n+m}(1)$.
\end{lemma}

\begin{proof}
Let $N=\binom{n+d}{d}$ and $B=d(N+1)$.
Run the reduction algorithm with respect to $G$.

Whenever a term is reduced, the $\vecX$-part of every newly
introduced monomial is no larger, with respect to
$\preceq_{\vecX}$, than the $\vecX$-part of the term being reduced.
Since $\preceq_{\vecX}$ is degree-compatible, the degree in
$\vecX$ never exceeds $d$. Thus at most $N$ different
$\vecX$-monomials can occur.

If a reduction leaves the $\vecX$-part unchanged, then the
degree in $\vecZ$ does not increase, because
$\preceq_{\vecZ}$ is degree-compatible. If the $\vecX$-part
strictly decreases, the degree in $\vecZ$ increases by at
most $d$. Along any chain of terms produced by successive
reductions, the latter event occurs at most $N-1$ times.
Consequently, every intermediate monomial has degree in
$\vecZ$ at most $dN$ and total degree at most $B$.
The multipliers $h_i$ produced by the algorithm also have
degree at most $B$.

For the coefficient bound, we first replace $G$ by a subset of $G$ with the
same leading-term ideal and with at most
$\binom{n+m+d}{d}$ elements. This does not change the remainder.
Introduce unknown coefficients for the polynomials $h_i$
of degree at most $B$ and for a remainder supported on
standard monomials of degree at most $B$.
The identity
\[
 f=f_G+\sum_i h_i g_i
\]
is a consistent linear system, obtained by comparing
coefficients through degree $B+d$.
Its size is $O_{d,n+m}(1)$, and its entries have complexity
bounded in terms of $d'$.
By \cref{lem:linsol}, it has a solution with coefficient
complexity $O_{d,d',n+m}(1)$.
Uniqueness of the reduced remainder identifies the resulting
polynomial with $f_G$.
\end{proof}

\paragraph{Ring homomorphisms.}

The following lemma gives quantitative, effective degree bounds on ring homomorphisms and the induced isomorphisms between quotient rings. While Part~(1) is the standard First Isomorphism Theorem, Part~(2) bounds the complexity of the \grobner basis of the kernel of a homomorphism, and Part~(3) shows that if an element in the image is represented by a polynomial with bounded complexity, then so is its preimage.

\begin{lemma}\label{lem:ring-iso}
Let $\KK[\vecX]=\KK[X_1,\dots,X_n]$
and $\KK[\vecZ]=\KK[Z_1,\dots,Z_m]$.
Let $f_1,\dots,f_k,g_1,\dots,g_m\in \KK_{\leq d'}[\vecX]_{\leq d}$.
Let $I=\ideal{f_1,\dots,f_k}\subseteq\KK[\vecX]$.
Let $\phi$ be the $\KK$-linear ring homomorphism 
\begin{align*}
\phi:\KK[\vecZ]&\to \KK[\vecX]/I\\
Z_i&\mapsto g_i + I, \quad i=1,2,\dots, m.
\end{align*} 
Let $A$ be the image of $\phi$, i.e., $A$ is the subring of $\KK[\vecX]/I$ generated by $g_1+I,\dots,g_m+I$.
Then:
\begin{enumerate}[(1)]
\item\label{item:poly1} $\phi$ induces an isomorphism $\KK[Z_1,\dots,Z_m]/\ker(\phi)\cong A$.
\item\label{item:poly2} For any monomial order $\preceq_\vecZ$ on $\mathcal{M}_\vecZ$, $\ker(\phi)$ has a \grobner basis with respect to $\preceq_\vecZ$ consisting only of polynomials in 
$\KK_{\leq D'}[\vecZ]_{\leq D}$, where $D=O_{d,n+m}(1)$ and $D'=O_{d,d',n+m}(1)$.
\item\label{item:poly3}
Let $f\in\KK[\vecX]$ such that $f+I\in A$ and $f\in \KK_{\leq d_2}[\vecX]_{\leq d_1}$. Then there exists $h\in \KK[\vecZ]$ such that $\phi(h)=f+I$ and $h\in \KK_{\leq d_4}[\vecZ]_{\leq d_3}$, where $d_3=O_{d,d_1,n+m}(1)$ and $d_4=O_{d,d',d_1,d_2,n+m}(1)$.
\end{enumerate}
\end{lemma}

\begin{proof}
\cref{item:poly1} holds by the first isomorphism theorem.
Let $J=\ideal{f_1,\dots,f_k, Z_1-g_1,\dots,Z_m-g_m}\subseteq\KK[\vecX, \vecZ]$.
Then $\ker(\phi)=J\cap \KK[\vecZ]$ \cite[Theorem~2.4.10]{AL94}. \cref{item:poly2} then follows from \cref{lem:elimination}.

It remains to prove \cref{item:poly3}.
Let $\preceq_\vecX$ and $\preceq_\vecZ$ be degree-compatible monomial orders on $\mathcal{M}_\vecX$ and on $\mathcal{M}_\vecZ$, respectively.
Let $\preceq=(\preceq_\vecX, \preceq_\vecZ)$.
Let $G$ be a \grobner basis of $J$ with respect to $\preceq$.
By \cref{lem:grobner-bounds}, we may assume that $G$ consists only of polynomials in $\KK_{\leq d_6}[\vecX,\vecZ]_{\leq d_5}$, where $d_5=O_{d,n+m}(1)$ and $d_6=O_{d,d',n+m}(1)$.

By assumption, $f+I$ is in $A$ and hence in the image of $\phi$.
Let $h$ be the remainder of $f$ modulo $G$ with respect to $\preceq$. Then $h\in\KK[\vecZ]$ and $\phi(h)=f+I$ by \cite[Theorem~2.4.11]{AL94}.
Apply \cref{lem:reduction-alg2} with degree parameter
$\max\{d_1,d_5\}$ and coefficient-complexity parameter
$\max\{d_2,d_6\}$.
It follows that
$h\in\KK_{\leq d_4}[\vecZ]_{\leq d_3}$, where
$d_3=O_{d,d_1,n+m}(1)$ and
$d_4=O_{d,d',d_1,d_2,n+m}(1)$.
\end{proof}

\paragraph{Primitive element theorem.} We now prove a quantitative form of the primitive element theorem, modifying a proof in \cite{ZS75}.

\begin{theorem}[Primitive element theorem, quantitative form]\label{thm:primitive}
Let $\KK[\vecX]=\KK[X_1,\dots,X_n]$.
Suppose $f_1,f_2,\dots,f_k\in\KK_{\leq d'}[\vecX]_{\leq d}$ generate a zero-dimensional ideal $I$.
Suppose $\mm$ is a maximal ideal of $\KK[\vecX]$ containing $I$ and $\LL:=\KK[\vecX]/\mm$ is a finite separable extension of $\KK$. 
Let $\alpha_i:=X_i+\mm\in\LL$ for $i\in [n]$, so that $\LL=\KK(\alpha_1,\dots,\alpha_n)$.
Then there exists a nonzero polynomial $Q\in\LL[\vecZ]=\LL[Z_1,\dots,Z_n]$ such that for every nonzero $\vecc=(c_1,\dots,c_n)\in\FF^n$ satisfying $Q(\vecc)\neq 0$, there exist  $P_0^\vecc,P_1^\vecc,\dots,P_n^\vecc\in\KK[T]$ such that the following hold:
\begin{enumerate}[(1)]
\item\label{item:PET1} 
Let $\beta_\vecc=\sum_{i=1}^n c_i \alpha_i\in\LL$.
Then $P_0^\vecc(\beta_\vecc)\neq 0$ and $\alpha_i=\frac{P_i^\vecc(\beta_\vecc)}{P_0^\vecc(\beta_\vecc)}$ for $i\in [n]$. In particular, $\LL=\KK(\beta_\vecc)$. 
\item\label{item:PET2} 
$P_0^\vecc,\dots, P_n^\vecc\in
\KK_{\leq d_0}[T]_{\leq d^n}$
for some $d_0=O_{d,d',n}(1)$.
\end{enumerate}
In particular, the above hold for almost all $\vecc\in\FF^n$.
\end{theorem}

\begin{proof}
Let $\widetilde{\KK}=\KK(\vecZ)=\KK(Z_1,\dots,Z_n)$ and $\widetilde{\LL}=\LL(\vecZ)=\LL(Z_1,\dots,Z_n)=\widetilde{\KK}(\alpha_1,\dots,\alpha_n)$. 
By assumption, $\alpha_1,\dots,\alpha_n$ are separable over $\KK$. So they are also separable over $\widetilde{\KK}$.
It follows that $\widetilde{\LL}$ is a finite separable extension of $\widetilde{\KK}$.

Let $\widehat{\alpha_i}=X_i+I\in \KK[\vecX]/I$ for $i\in [n]$.
Let $\widehat{\beta}(\vecZ)=Z_1\widehat{\alpha_1}+Z_2\widehat{\alpha_2}+\dots+Z_n\widehat{\alpha_n}\in (\KK[\vecX]/I)\otimes_{\KK}\widetilde{\KK}$.
We have $\dim_{\KK} (\KK[\vecX]/I)\leq d^n$ by \cref{lem:bezout2}.
So $\dim_{\widetilde{\KK}}((\KK[\vecX]/I)\otimes_{\KK}\widetilde{\KK})\leq d^n$.
Therefore, $\widehat{\beta}(\vecZ)$ has a minimal polynomial over $\widetilde{\KK}$, which we denote by $\widehat{F}$, and its degree is at most $d^n$.

We may construct $\widehat{F}$ as follows.
Let $\preceq_{T,\vecZ}$ be an elimination order for $T$ on $\mathcal{M}_{T,\vecZ}$, and let $\preceq=(\preceq_\vecX, \preceq_{T,\vecZ})$ be an elimination order for $\vecX$ on $\mathcal{M}_{\vecX,T,\vecZ}$.
Let $J$ be the ideal 
\[
J=\ideal{f_1,\dots,f_k,T-(Z_1X_1+\dots+Z_nX_n)}
\] of $\KK[\vecX, T, \vecZ]$.
The ideal $J$ is the preimage of the ideal $\ideal{T-\widehat{\beta}(\vecZ)}$ of $(\KK[\vecX]/I)[T,\vecZ]$ under the natural quotient map $ \KK[\vecX, T, \vecZ]\to (\KK[\vecX]/I)[T,\vecZ]$.
Let $G$ be a \grobner basis for $J$ with respect to $\preceq$.
Then $G\cap \KK[T,\vecZ]$ is a \grobner basis for $J\cap \KK[T,\vecZ]$ with respect to $\preceq_{T,\vecZ}$ by \cref{lem:elim-ideal}.
Choose $\widehat{g}\in G\cap \KK[T,\vecZ]$ such that $\lm{\widehat{g}}$ has the form $m T^{e}$, where $m\in \mathcal{M}_\vecZ$ and $e$ is minimized. 
By the choice of $\preceq$, we have $\widehat{g}\in \KK[T,\vecZ]$. Note that such $\widehat{g}$ exists since it can be obtained by clearing the denominators of the coefficients of $\widehat{F}(T)$. 
By \cref{lem:grobner-bounds}, we may assume $\widehat{g}\in \KK_{\leq d_1}[T,\vecZ]_{\leq D}$ with $D=O_{d,n}(1)$ and $d_1=O_{d,d',n}(1)$.

Write $\widehat{g}=\sum_{i=0}^e h_i T^i$ with $h_i\in \KK[\vecZ]$ for $i=0,1,\dots,e$. Then $h_e\neq 0$. As $e$ is minimized and $J$ is the preimage of $\ideal{T-\widehat{\beta}(\vecZ)}$ in $\KK[\vecX, T, \vecZ]$, $\widehat{g}/h_e=\sum_{i=0}^e (h_i/h_e) T^i$ is precisely the minimal polynomial $\widehat{F}(T)$ of $\widehat{\beta}(\vecZ)$ over $\widetilde{\KK}$. 

Consider
\[
\beta(\vecZ):=Z_1 \alpha_1+Z_2 \alpha_2+\dots+Z_n \alpha_n\in \LL[\vecZ]\subseteq \widetilde{\LL}.
\]
Let $F(T)\in\widetilde{\KK}[T]$ be the minimal polynomial of $\beta(\vecZ)$ over $\widetilde{\KK}$. As $\mm\supseteq I$, $F$ is a factor of $\widehat{F}$. In particular, $\deg(F)\leq \deg(\widehat{F})\leq d^n$. 
As $\widehat{g}/\widehat{F}=h_e\in \widetilde{\KK}^\times$, $F$ is a factor of $\widehat{g}$, both viewed as polynomials over $\widetilde{\KK}$. So $\widehat{g}$ has a factor $g\in \KK[T, \vecZ]$ such that $g/F\in \widetilde{\KK}^\times$. Moreover, by \cref{lem:factor}, after replacing $g$ with $cg$ for a suitable
$c\in\KK^\times$, we may assume that
$g\in \KK_{\le d_2}[T,\vecZ]_{\le D}$, where
$d_2=O_{D,d_1,n}(1)=O_{d,d',n}(1)$.

As $g/F\in \widetilde{\KK}^\times$ and $F=F(T)$ is the minimal polynomial of $\beta(\vecZ)$ over $\widetilde{\KK}$, we have
\begin{equation}\label{eq:g-is-zero}
g(Z_1 \alpha_1+Z_2 \alpha_2+\dots+Z_n \alpha_n, Z_1,\dots,Z_n)=g(\beta(\vecZ), \vecZ)=0.
\end{equation}
Taking the partial derivatives of \eqref{eq:g-is-zero} with respect to $Z_1,\dots,Z_n$ and applying the chain rule for multivariate polynomials, we obtain
\begin{equation}\label{eq:relations}
\alpha_i g_0(\beta(\vecZ), \vecZ)+g_i(\beta(\vecZ),\vecZ)=0,\quad i=1,2,\dots,n,
\end{equation}
where $g_0(T,\vecZ):=\frac{\partial g(T,\vecZ)}{\partial T}$ and $g_i(T,\vecZ):=\frac{\partial g(T,\vecZ)}{\partial Z_i}$ for $i\in [n]$. 

As $F=F(T)$ is the minimal polynomial of $\beta(\vecZ)\in\widetilde{\LL}$ over $\widetilde{\KK}$ and $\widetilde{\LL}$ is a finite separable extension of $\widetilde{\KK}$, we have $F'(\beta(\vecZ))\neq 0$. 
And as $g_0(T,\vecZ)=\frac{\partial g(T,\vecZ)}{\partial T}$ and $g/F\in \widetilde{\KK}^\times$, we have
\[
g_0(\beta(\vecZ),\vecZ)= (g/F) \cdot F'(\beta(\vecZ))\neq 0.
\]
Define $Q(\vecZ)=g_0(\beta(\vecZ),\vecZ)\in\LL[\vecZ]$.
For every $\vecc=(c_1,\dots,c_n)\in\FF^n$, define  $P_0^\vecc(T)=g_0(T,\vecc)\in \KK[T]$ and $P_i^\vecc(T)=-g_i(T,\vecc)\in\KK[T]$ for $i\in [n]$.

Consider $\vecc\in\FF^n$ such that $Q(\vecc)\neq 0$.
Note that $\beta(\vecc)=\sum_{i=1}^n c_i \alpha_i=\beta_{\vecc}$ and $P_0^\vecc(\beta_\vecc)=g_0(\beta(\vecc), \vecc)=Q(\vecc)\neq 0$.
We have by \eqref{eq:relations} that
\[
\alpha_i=-\frac{g_i(\beta(\vecc),\vecc)}{g_0(\beta(\vecc),\vecc)}=\frac{P_i^\vecc(\beta_\vecc)}{P_0^\vecc(\beta_\vecc)},\quad i=1,2,\dots,n,
\]
By definition, $P_0^\vecc(T)=g_0(T,\vecc)$ and $P_i^\vecc(T)=-g_i(T,\vecc)$ for $i\in [n]$.
So we have $\deg(P_i^\vecc)\leq \deg_T(g_i)\leq \deg_T(g)\leq e\leq d^n$ for $i= 0,1,\dots,n$. 
As already noted above, $g\in \KK_{\le d_2}[T,\vecZ]_{\le D}$, so the coefficients of $g$ are all in $\KK_{\le d_2}$.
For $i=0,1,\dots,n$ and $j=0,1,\dots,\deg(P_i^\vecc)$,
the coefficient of $T^j$ in $P_i^\vecc$ is a linear combination of the coefficients of the monomials of degree $j$ in $T$ of $g_i$ over $\FF$, 
and $g_i$ has at most $\binom{n+D}{D}$ such monomials since $\deg(g_i)\leq \deg(g)\leq D$.
Therefore, by \cref{lem:frac-deg-bound}, the coefficients of $P_0^\vecc,\dots,P_n^\vecc$ are in $\KK_{\le d_0}$, where $d_0=\binom{n+D}{D}d_2=O_{d,d',n}(1)$.
It follows that $P_0^\vecc,\dots, P_n^\vecc\in \KK_{\leq d_0}[T]_{\leq d^n}$.
\end{proof}

\paragraph{Extracting a maximal ideal.} For $I$ a zero-dimensional ideal of $\KK[\vecX]$, we need to solve the problem of finding the generators of a maximal ideal $\mm$ containing $I$ or, more precisely, to bound the complexity of these generators. This problem is a special case of finding the radical and the primary decomposition of a zero-dimensional ideal, with coefficient bounds. The following lemma provides a direct solution, providing the required bounds on the complexity of the generators of $\mm$.

\begin{lemma}\label{lem:extract}
Let $\KK[\vecX]=\KK[X_1,\dots,X_n]$ and $f_1,\dots,f_k\in\KK_{\leq d'}[\vecX]_{\leq d}$.
Suppose $I=\ideal{f_1,\dots,f_k}$ is a zero-dimensional ideal of $\KK[\vecX]$.
Let $\mm$ be a maximal ideal of $\KK[\vecX]$ containing $I$ such that $\KK[\vecX]/\mm$ is a finite separable extension of $\KK$. 
Then $\mm$ is generated by polynomials in 
$\KK_{\leq d''}[\vecX]_{\leq d^n+1}$, where $d''=O_{d,d',n}(1)$.
\end{lemma}

\begin{proof}
We follow the approach of \cite{GTZ88}.
Let $\alpha_i=X_i+\mm\in \KK[\vecX]/\mm$ for $i\in [n]$. 
By \cref{thm:primitive}, there exist $\vecc=(c_1,\dots,c_n)\in\FF^n$ and $P_0^\vecc,P_1^\vecc,\dots,P_n^\vecc\in\KK_{\leq d_0}[T]_{\leq d^n}$  for some $d_0=O_{d,d',n}(1)$ such that
$\beta_\vecc=\sum_{i=1}^n c_i \alpha_i$ 
satisfies $P_0^\vecc(\beta_\vecc)\neq 0$ and $\alpha_i=\frac{P_i^\vecc(\beta_\vecc)}{P_0^\vecc(\beta_\vecc)}$ for $i\in [n]$.

Let $T=\sum_{i=1}^n c_i X_i\in \KK[\vecX]$.
Let $\mm_0=\mm\cap \KK[T]$, which is a prime ideal of $\KK[T]$. Let $I_0=I\cap \KK[T]$.
As $\KK[T]$ is a PID, $\mm_0$ and $I_0$ are generated by polynomials $g,h\in\KK[T]$ over $\KK$, respectively, and $g|h$. 
The inclusion $\KK[T]\hookrightarrow \KK[\vecX]$ induces an inclusion $\KK[T]/I_0\hookrightarrow\KK[\vecX]/I$.
So we have
\[
\deg(g)\leq \deg(h)=\dim_\KK(\KK[T]/I_0)\leq \dim_\KK(\KK[\vecX]/I)\leq d^n,
\]
where the last inequality holds by \cref{lem:bezout2}.
By \cref{lem:elimination} and the fact that $I_0=\ideal{h}$, we may assume that $h\in \KK_{\leq d_1}[T]_{\leq d^n}$ for some $d_1=O_{d,d',n}(1)$.\footnote{As $c_1,\dots,c_n\in\FF$, we may assume that $T$ is one of $X_1,\dots,X_n$ by performing an $\FF$-linear change of coordinates.
This transformation may increase the degrees of the numerators and denominators of the coefficients, but only by a factor of $O_{d,n}(1)$, by \cref{lem:frac-deg-bound}.}
As $g|h$, by \cref{lem:factor}, we may assume that $g\in \KK_{\leq d_2}[T]_{\leq d^n}$ for some $d_2=O_{d,d',n}(1)$.

For $i\in [n]$, let $g_i=P_0^\vecc(T) X_i-P_i^\vecc(T)\in\KK[\vecX]$, whose degree is at most $d^n+1$.
And the coefficients of each $g_i$ are in $\KK_{\leq d_3}$ for some $d_3=O_{d_0,d,n}(1)=O_{d,d',n}(1)$ since $P_0^\vecc, P_i^\vecc\in  \KK_{\leq d_0}[T]_{\leq d^n}$. 
Let $d''=\max(d_2,d_3)$.
Then $g,g_1,\dots,g_n\in  \KK_{\leq d''}[\vecX]_{\leq d^n+1}$.

For $i\in [n]$, by the definition of $\alpha_i$, the image of $g_i$ in $\KK[\vecX]/\mm$ is $P_0^\vecc(\beta_\vecc)\alpha_i-P_i^\vecc(\beta_\vecc)$, which is zero since $\alpha_i=\frac{P_i^\vecc(\beta_\vecc)}{P_0^\vecc(\beta_\vecc)}$.
So $g_1,\dots,g_n\in\mm$.

Note that $\KK[\vecX]/\ideal{g,g_1,\dots,g_n}\cong \KK[T]/\ideal{g}=\KK[T]/\mm_0$.
This follows by noting that $P_0^\vecc(T)\in\KK[T]$ is invertible modulo $\ideal{g}=\mm_0$ (since $P_0^\vecc(\beta_\vecc)\neq 0$ and $\mm_0=\mm\cap \KK[T]$ is a maximal ideal, modulo which every nonzero element is invertible), and therefore, we can use the relations $g_i=P_0^\vecc(T) X_i-P_i^\vecc(T)=0$, $i=1,\dots,n$, to eliminate $X_1,\dots,X_n$.
As $\KK[T]/\mm_0$ is a field, $\ideal{g,g_1,\dots,g_n}$ is a maximal ideal of $\KK[\vecX]$.
As $g,g_1,\dots,g_n\in\mm$ and $\mm$ is also a maximal ideal of $\KK[\vecX]$, we have $\mm=\ideal{g,g_1,\dots,g_n}$.
The lemma follows.
\end{proof}

Moving on, we introduce the notion of an \emph{idealizer} in an integral domain, which will be used in the next section to compute the normalization of curves. We begin by introducing several related definitions and then prove a lemma that bounds the complexity of the preimage of an idealizer.

\paragraph{Idealizer.}

Let $A$ be an integral domain and $J$ be an ideal of $A$. The \emph{idealizer} of $J$ is defined as
\[
\id{A}{J}:=\{a\in\fr{A}: aJ\subseteq J\}.
\]
It is the largest subring of $\fr{A}$ in which $J$ is still an ideal.

A closely related notion is the \emph{ideal quotient} of two ideals.
For ideals $I$ and $J$ of a ring $R$, define the ideal quotient $(I:J):=\{a\in R: aJ\subseteq I\}$, which is an ideal of $R$.

\begin{lemma}\label{lem:find-idealizer}
Let $\KK[\vecX]=\KK[X_1,\dots,X_n]$ and $f_1,\dots,f_k,g_1,\dots,g_m\in\KK_{\leq d'}[\vecX]_{\leq d}$.
Suppose $I:=\ideal{f_1,\dots,f_k}$ is a prime ideal of $\KK[\vecX]$, or equivalently, $A:=\KK[\vecX]/I$ is an integral domain.
Let $J$ be the ideal $\ideal{g_1+I,\dots,g_m+I}$ of $A$.

Assume that $c\cdot \id{A}{J}\subseteq A$, where $c=f+I\in A\setminus\{0\}$ and $f\in \KK_{\leq d'}[\vecX]_{\leq d}$.
Then with respect to any monomial order, the preimage of $c\cdot \id{A}{J}$ in $\KK[\vecX]$ under the natural quotient map $\KK[\vecX]\to A$ has a \grobner basis consisting only of polynomials in $\KK_{\leq D'}[\vecX]_{\leq D}$, where $D=O_{d,n}(1)$ and $D'=O_{d,d',n}(1)$.
\end{lemma}

\begin{proof}
If $J$ is the zero ideal, simply choose $\{1\}$ as the \grobner basis. So assume $J\neq \{0\}$.
By keeping linearly independent subsets of the
nonzero polynomials among the $f_i$ and among the $g_i$, we may assume $k,m\leq\binom{n+d}{d}$ and $g_1,\dots,g_m\neq 0$.
For any ideal $I_0$ of $A$, denote by $\widehat{I_0}$ the preimage of $I_0$ in $\KK[\vecX]$ under the natural quotient map, i.e., $\widehat{I_0}=\{a\in \KK[\vecX]: a+I\in I_0\}$.

As $c\neq 0$ and $A$ is an integral domain, we know $c$ is a non-zero-divisor of $A$.
As $c\cdot \id{A}{J}\subseteq A$, we know 
\begin{equation*}%\label{eq:idealizer1}
c\cdot \id{A}{J}=\{b\in A:b=ca, aJ\subseteq J\}
=\left\{b\in A: \frac{b}{c}J\subseteq J\right\}
=\{b\in A: bJ\subseteq cJ\}=(cJ:J).
\end{equation*}
It follows that 
\begin{equation}\label{eq:idealizer2}
\widehat{c\cdot \id{A}{J}}=\widehat{(cJ:J)}=(\widehat{cJ}:\widehat{J}).
\end{equation}
By definition, $\widehat{J}=\ideal{f_1,\dots,f_k,g_1,\dots,g_m}$ and $\widehat{cJ}=\ideal{f_1,\dots,f_k,fg_1,\dots,fg_m}$. It follows from \cite[Lemmas~2.3.10 and 2.3.11]{AL94} that
\begin{equation}\label{eq:idealizer3}
\begin{aligned}
(\widehat{cJ}:\widehat{J})&=\bigcap_{h\in\{f_1,\dots,f_k,g_1,\dots,g_m\}} \frac{1}{h}(\ideal{f_1,\dots,f_k,fg_1,\dots,fg_m}\cap\ideal{h})\\
&=\bigcap_{i=1}^m \frac{1}{g_i}(\ideal{f_1,\dots,f_k,fg_1,\dots,fg_m}\cap\ideal{g_i})\\
&=\bigcap_{i=1}^m \frac{1}{g_i} I_i,
\end{aligned}
\end{equation}
where $I_i:=\ideal{f_1,\dots,f_k,fg_1,\dots,fg_m}\cap\ideal{g_i}$ for $i\in [m]$.
By \cite[Proposition~2.3.5]{AL94}, we can compute this intersection by introducing a new variable $T$ and then eliminating it:
\begin{equation}\label{eq:idealizer5}
\begin{aligned}
I_i&=(T\cdot \ideal{f_1,\dots,f_k,fg_1,\dots,fg_m}+(1-T)\cdot \ideal{g_i})\cap \KK[\vecX]\\
&=\ideal{Tf_1,\dots,Tf_k,Tfg_1,\dots,Tfg_m, (1-T)g_i}\cap \KK[\vecX].
\end{aligned}
\end{equation}
As $f,f_1,\dots,f_k,g_1,\dots,g_m\in \KK_{\leq d'}[\vecX]_{\leq d}$, by \cref{lem:prod-deg-bound}, we have $Tf_1,\dots,Tf_k,Tfg_1,\dots,Tfg_m, (1-T)g_1,\dots, (1-T)g_m\in \KK_{\leq d_0}[\vecX, T]_{\leq 2d+1}$ for some $d_0=O_{d,d',n}(1)$.
By \eqref{eq:idealizer5} and \cref{lem:elimination}, for $i\in [m]$, $I_i$ has a \grobner basis $G_i\subseteq \KK_{\leq d_2}[\vecX]_{\leq d_1}$, where $d_1=O_{d,n}(1)$ and $d_2=O_{d,d',n}(1)$.

Consider $i\in [m]$. The polynomials in $G_i$ are all divisible 
by $g_i$ since $I_i\subseteq\ideal{g_i}$.
Let $\frac{1}{g_i}G_i=\{\frac{g}{g_i}: g\in G_i\}$.
Note that $\lt{\frac{1}{g_i}G_i}=\lt{\frac{1}{g_i}I_i}$ since $\lt{G_i}=\lt{I_i}$.
After rescaling the elements of $\frac1{g_i}G_i$
by nonzero elements of $\KK$, \cref{lem:factor}
gives a Gr\"obner basis of the ideal $\frac{1}{g_i}I_i$ contained in
$\KK_{\leq d_3}[\vecX]_{\leq d_1}$,
where $d_3=O_{d,d',n}(1)$.

By \cite[Exercise~2.3.8]{AL94}, we have
\begin{equation}\label{eq:idealizer4}
\bigcap_{i=1}^m \frac{1}{g_i} I_i=I^*\cap \KK[\vecX]
\end{equation}
where 
\[
I^*:=\ideal{1-\sum_{i=1}^m T_i}+T_1 \cdot \frac{1}{g_1} I_1+\dots+T_m\cdot \frac{1}{g_m}I_m\subseteq \KK[\vecX, T_1,\dots,T_m].
\]
The ideal $I^*$ has the set of generators $\{1-\sum_{i=1}^m T_i\}\cup\{T_i\cdot g: g\in \frac{1}{g_i}G_i, i\in [m]\}$.
These generators are in $\KK_{\leq d_3}[\vecX,\vecT]_{\leq d_1+1}$ as $\frac{1}{g_i}G_i\subseteq \KK_{\leq d_3}[\vecX]_{\leq d_1}$.
By \eqref{eq:idealizer4} and \cref{lem:elimination}, $\bigcap_{i=1}^m \frac{1}{g_i} I_i$ has a \grobner basis contained in $\KK_{\leq D'}[\vecX]_{\leq D}$, where $D=O_{d,n}(1)$ and $D'=O_{d,d',n}(1)$.
Finally, by \eqref{eq:idealizer2} and \eqref{eq:idealizer3}, we know that $\bigcap_{i=1}^m \frac{1}{g_i}I_i$ is the preimage $\widehat{c\cdot \id{A}{J}}$ of $c\cdot \id{A}{J}$ under the natural quotient map $\KK[\vecX]\to A$, concluding the proof.
\end{proof}

\section{Normalization of Curves}\label{sec:normalization}

In this section, we present a constructive process for normalizing affine curves, following the framework of Trager~\cite{Tra84}, which itself is a function field analog of an algorithm by Ford and Zassenhaus~\cite{For78}. We start with an order contained in the coordinate ring and iteratively adjoin integral elements obtained from idealizers. Allowing purely inseparable extensions of bounded exponent, we continue until the resulting order is geometrically normal over its base field.

We revisit these arguments primarily because the bounds on coefficient complexity that we need are not provided in Trager~\cite{Tra84} or other literature. Moreover, our treatment differs in several respects from that of Trager. In particular, Trager uses a different method for computing the idealizer, which imposes restrictions on the characteristic of the base field. In contrast, our method uses \grobner bases and is characteristic-free.

Throughout this section, let $\KK=\FF(\vecY)=\FF(Y_1,\dots,Y_\ell)$, where $\FF$ is an algebraically closed field.

\subsection{Orders, Integral Bases, and Discriminants}\label{sec:ant-notions}

In this subsection, we introduce the core algebraic objects used to compute the integral closure. In particular, \Cref{thm:criterion} provides a useful criterion for determining when an order is integrally closed.

Throughout this subsection, let $\LL$ be a finite separable extension of $\KK(X)$ of degree $s$, and let $\order_{\LL}$ denote the integral closure of $\KK[X]$ in $\LL$.

\paragraph{Orders.} A subring $\order\subseteq\order_{\LL}$ is said to be a $(\KK[X], \LL)$-order if it is a finite $\KK[X]$-module and $\order\otimes_{\KK[X]} \KK(X)=\LL$ (i.e., the elements in $\order$ span $\LL$ over $\KK(X)$).

\paragraph{Integral bases.} It is known that every $(\KK[X], \LL)$-order  is a free $\KK[X]$-module of rank $[\LL:\KK(X)]=s$.
For a $(\KK[X], \LL)$-order $\order$, we say $\vecb=(b_1,\dots,b_s)\in\order^s$ is an \emph{integral basis} of $\order$ if $b_1,\dots,b_s$ form a basis of $\order$ over $\KK[X]$.

\paragraph{Discriminants.}

Let $\vecb=(b_1,\dots,b_s)\in\LL^s$. 
As $\LL/\KK(X)$ is separable, by Galois theory, $\LL$ has $s$ distinct embeddings $\sigma_1,\dots,\sigma_s$ into the Galois closure of $\LL$ over $\KK(X)$.
The \emph{discriminant} of $\vecb$ is defined to be
\[
\disc{\vecb}
=\det\begin{pmatrix}
    \sigma_1(b_1) & \sigma_1(b_2) & \cdots & \sigma_1(b_s)\\
    \sigma_2(b_1) & \sigma_2(b_2) & \cdots & \sigma_2(b_s)\\
    \vdots & \vdots & \ddots & \vdots\\
    \sigma_s(b_1) & \sigma_s(b_2) & \cdots & \sigma_s(b_s)
\end{pmatrix}^2.
\]
It is fixed by the Galois group of the Galois closure of $\LL/\KK(X)$, and hence is an element of $\KK(X)$.

We also need the definition of the discriminant of a univariate polynomial.
For simplicity, we only give the definition for univariate monic polynomials.
Let $f(T)=\sum_{i=0}^d a_i T^i\in R[T]$ be a univariate monic (i.e., $a_d=1$) polynomial of degree $d$ over a ring $R$. The \emph{discriminant} $\disc{f}$ of $f$ is defined to be
$(-1)^{d(d-1)/2}\mathrm{Res}(f,f')\in R$, where $\mathrm{Res}(f,f')$ denotes
the resultant of $f$ and $f'=\sum_{i=0}^{d-1}(i+1)a_{i+1} T^i$, given by
\begin{equation*}
\mathrm{Res}(f,f')=
\det\begin{pmatrix}
a_d & a_{d-1} & a_{d-2} & \cdots 
& a_0 & 0 & \cdots & 0 \\
0 & a_d & a_{d-1} & \cdots 
& a_1 & a_0 & \cdots & 0 \\
\vdots & \vdots & \vdots & \ddots & 
\vdots & \vdots & \ddots & \vdots \\
0 & \cdots & 0 & a_d & a_{d-1} 
& a_{d-2} & \cdots & a_0 \\
da_d & (d-1)a_{d-1} & (d-2)a_{d-2} & \cdots & a_1 & 0 & 
\cdots & 0 \\
0 & da_d & (d-1)a_{d-1} & \cdots & 2a_2 & a_1 & 
\cdots & 0 \\
\vdots & \vdots & \vdots & \ddots & \vdots & \vdots &  \ddots & \vdots \\
0 & \cdots & 0 & da_d & (d-1)a_{d-1} & (d-2)a_{d-2} & \cdots  & a_1
\end{pmatrix}.
\end{equation*}

The following lemma follows straightforwardly from the definition of $\disc{f}$ and \cref{lem:frac-deg-bound}.

\begin{lemma}\label{lem:disc-poly}
Let $R=\KK[X]$.
Suppose $f\in R[T]$ is monic, and is in $\KK_{\leq d'}[X,T]_{\leq d}$ when viewed as a polynomial in both $X$ and $T$ over $\KK$.
Then $\disc{f}\in \KK_{\leq D'}[X]_{\leq D}$ for some $D=O_{d}(1)$ and $D'=O_{d,d'}(1)$.
\end{lemma}

We now list some facts about the discriminant of a tuple, the discriminant of a polynomial, and their relations. These facts can be found in, e.g., \cite{Tra84}.

\begin{lemma}\label{lem:disc-facts}
We have the following facts:
\begin{enumerate}[(1)] 
\item\label{item:disc1} Let $\vecb=(b_1,\dots,b_s)\in\LL^s$. Then $\disc{\vecb}\neq 0$ if and only if $b_1,\dots,b_s$ are linearly independent over $\KK(X)$.
\item\label{item:disc2} Let $\vecb=(b_1,\dots,b_s)\in\LL^s$ 
such that $b_1,\dots,b_s$ are integral over $\KK[X]$. Then $\disc{\vecb}\in \KK[X]$. 
\item\label{item:disc3} Let $\vecb,\vecb'\in\LL^s$, where $\disc{\vecb'}\neq 0$.
Suppose $\vecb=A\cdot \vecb'$ for some $A\in\KK[X]^{s\times s}$. Then $\disc{\vecb}=\det(A)^2\cdot \disc{\vecb'}$. And $\vecb$ and $\vecb'$ generate the same $\KK[X]$-module if and only if $A$ is invertible as a matrix over $\KK[X]$, i.e., $\det(A)\in\KK[X]^\times=\KK^\times$.
\item\label{item:disc4} Suppose $\LL=\KK(X)(\alpha)$ and $\alpha\in\LL$ is integral over $\KK[X]$. Let $f$ be the minimal polynomial of $\alpha$ over $\KK(X)$, which is a monic polynomial with all coefficients in $\KK[X]$ \cite[Theorem~9.2]{Mat89}. 
Then $\disc{1,\alpha,\dots,\alpha^{s-1}}=\pm\disc{f}$.
\end{enumerate}
\end{lemma}

Recall that $\order_\LL$ denotes the integral closure of $\KK[X]$ in $\LL$.
We have the following statement, whose proof can be found in, e.g., \cite[Proof of Proposition~13.14]{Eis13}.

\begin{lemma}\label{lem:disc-square}
Let $\order\subseteq \order_\LL$ be a $(\KK[X], \LL)$-order with an integral basis $\vecb$.
Then
\[
\order_\LL\subseteq \frac{1}{\disc{\vecb}} \order\subseteq  \frac{1}{\disc{\vecb}} \order_\LL.
\]
\end{lemma}

\begin{definition}[Discriminant ideal]
Let $\order\subseteq \order_\LL$ be a $(\KK[X], \LL)$-order with an integral basis $\vecb$.
Define $\mathfrak{D}_{\order/\KK[X]}$ to be the ideal of $\KK[X]$ generated by $\disc{\vecb}$, called the \emph{discriminant ideal} of $\order$. By  \cref{lem:disc-facts}\,\eqref{item:disc3}, $\mathfrak{D}_{\order/\KK[X]}$ is well-defined and does not depend on the choice of $\vecb$.
\end{definition}

The following theorem gives a criterion for $\order$ being integrally closed.

\begin{theorem}[\cite{Tra84}]\label{thm:criterion}
$\order$ is integrally closed if and only if the idealizer $\id{\order}{\mm}$ of every maximal ideal $\mm$ of $\order$ containing $\mathfrak{D}_{\order/\KK[X]}$ equals $\order$. 
\end{theorem}

The next lemma states that taking the idealizer does not introduce non-integral elements.

\begin{lemma}[\cite{Tra84}]\label{lem:idealizer-integral}
 The idealizer $\id{\order}{I}$ is contained in the integral closure of $\order$ for any nonzero ideal $I$ of $\order$.
\end{lemma}

A $k$-algebra $B$ is said to be \emph{geometrically normal} over $k$
if $B\otimes_k k'$ is normal for every field extension $k'/k$.
We will use the following criterion.

\begin{lemma}\label{lem:geometric-normality}
Let $k$ be a field, let $E/k(x)$ be a finite
separable extension, and let $B$ be a $(k[x],E)$-order.
Let $Q(x)$ generate its discriminant ideal.
Suppose $B$ is integrally closed and $B/\mathfrak m$ is
separable over $k$ for every maximal ideal $\mathfrak m$
of $B$ containing $Q(x)$.
Then $B$ is geometrically normal over $k$.
\end{lemma}

\begin{proof}
The discriminant is invertible after inverting $Q$, so
$B_Q$ is finite \'etale over $k[x]_Q$
\cite[\href{https://stacks.math.columbia.edu/tag/0BJF}
{Tag 0BJF}]{stacks-project}, and hence smooth over $k$.
At each maximal ideal $\mathfrak m$ containing $Q$, the
local ring $B_{\mathfrak m}$ is a normal one-dimensional
Noetherian local domain, hence a discrete valuation ring.
It is therefore regular; together with the separability
of $B/\mathfrak m$ over $k$, this implies smoothness at
$\mathfrak m$
\cite[\href{https://stacks.math.columbia.edu/tag/00TV}
{Tag 00TV}]{stacks-project}.
It follows that $\operatorname{Spec}(B)$ is smooth over $k$ \cite[\href{https://stacks.math.columbia.edu/tag/01V5}{Tag 01V5}]{stacks-project}, and is
therefore geometrically normal
\cite[\href{https://stacks.math.columbia.edu/tag/056T}
{Tag 056T}]{stacks-project}.
\end{proof}

\subsection{Finding the Integral Closure}\label{sec:ic}

The main result of this section is the following theorem, which gives a quantitative version of the normalization of affine curves, or equivalently, of the computation of the integral closure of an order in a function field of transcendence degree one. Specifically, we give explicit bounds of the form $O_{d,d',n}(1)$ on the complexity of the defining data of the integral closure, including its generators and relations.

A technical difficulty is handling inseparable field extensions, since the base field $\FF$ may have positive characteristic. To address this, we allow variables of the form $Y_i^{1/p}$, thereby extending the field $\KK$ to $\KK^{(e)}$, where $e$ is a parameter. As before, the precise dependence on $d$, $d'$, and $n$ is unimportant, since these parameters will be treated as constants in the main argument of the next section.

\begin{theorem}\label{thm:integral-closure}
Let $\KK[\vecX]=\KK[X_1,\dots,X_n]$ and $f_1,\dots,f_k\in \KK_{\leq d'}[\vecX]_{\leq d}$. 
Let $\alpha$ be an element in the $\FF$-linear span of $X_1,\dots,X_n$ such that
the natural ring homomorphism 
\[
\KK[\alpha]\to A:=\KK[\vecX]/\ideal{f_1,\dots,f_k}
\]
sending $\alpha$ to $\alpha+\ideal{f_1,\dots,f_k}$ is injective and makes $A$ a finite $\KK[\alpha]$-module.
Moreover, suppose the following hold:
\begin{enumerate}[(1)]
\item For any algebraic extension $\LL$ of $\KK$, $A_\LL:=\LL[\vecX]/\ideal{f_1,\dots,f_k}$ is an integral domain of Krull dimension one.
\item $\fr{A}$ is a finite separable extension of $\KK(\alpha)$.
\end{enumerate} 
Then there exist $D,D',m,k',e\in\NN$ and  $g_1,\dots,g_{n},h_1,\dots,h_{k'}\in\KK^{(e)}[\vecZ]=\KK^{(e)}[Z_1,\dots,Z_m]$ such that the following hold:
\begin{enumerate}[(1)]
\item 
$D,D',m,k',p^e=O_{d,d',n}(1)$
and $g_1,\dots,g_{n},h_1,\dots,h_{k'}\in \KK^{(e)}_{\leq D'}[\vecZ]_{\leq D}$.
\item The map
\begin{align*}
\phi: A_{\KK^{(e)}}&\to \KK^{(e)}[\vecZ]/\ideal{h_1,\dots,h_{k'}}
\\
X_i+\ideal{f_1,\dots,f_k}&\mapsto g_i + \ideal{h_1,\dots,h_{k'}}, \quad i=1,2,\dots, n
\end{align*} 
defines an injective $\KK^{(e)}$-linear ring homomorphism.
\item $\KK^{(e)}[\vecZ]/\ideal{h_1,\dots,h_{k'}}$ is isomorphic to the integral closure of $A_{\KK^{(e)}}$, and this isomorphism composed with $\phi$ is the natural inclusion of $A_{\KK^{(e)}}$ in its integral closure.
\item $\KK^{(e)}[\vecZ]/\ideal{h_1,\dots,h_{k'}}$
is geometrically normal over $\KK^{(e)}$.
\end{enumerate}
\end{theorem}

\begin{remark}
The subring $\KK[\alpha]$ in \cref{thm:integral-closure} can be obtained via Noether normalization, though we defer this step to the next section. For our purposes, $\alpha$ is only in the $\FF$-linear span of $X_1, ..., X_n$, which is independent of the $Y$ variables. 
\end{remark}
\begin{remark}
The assumptions that $A_\LL$ is an integral domain and that $\fr{A}$ is separable over $\KK(\alpha)$ are imposed for simplicity; they hold in our setting.  It may be possible to eliminate these assumptions, but doing so could require additional tools, such as primary decomposition.
\end{remark}

In the following, we adopt the notations and assumptions of \cref{thm:integral-closure}. 

By assumption, $\fr{A}$ is a finite separable extension of $\KK(\alpha)$. 
We view $\alpha$ as a variable as it is transcendental over $\KK$.
By the primitive element theorem \cite[\S II.9, Theorem~19]{ZS75}, we can fix $X=\sum_{i=1}^n c_i X_i$ with $c_i\in \FF$ such that $\overline{X}:=X+\ideal{f_1,\dots,f_k}$ generates the field $\fr{A}$ over $\KK(\alpha)$. 
Let $s:=[\fr{A}:\KK(\alpha)]$ and $\vecb_0=(1,\overline{X},\dots,\overline{X}^{s-1})$.

To prove \cref{thm:integral-closure}, we first prove several lemmas. First, the following lemma bounds the complexity of the discriminant of $\vecb_0$.

\begin{lemma}\label{lem:disc-bound}
$\disc{\vecb_0}\in \KK_{\leq d_2}[\alpha]_{\leq d_1}$ for some $d_1=O_{d,n}(1)$ and $d_2=O_{d,d',n}(1)$.
\end{lemma}

\begin{proof}
By assumption, $A$ is a finite $\KK[\alpha]$-module.
So by \cref{lem:finite-vs-integral}, it is integral over $\KK[\alpha]$. Therefore, by \cref{lem:disc-facts}\,\eqref{item:disc4}, the minimal polynomial $F$ of $\overline{X}$ over $\KK(\alpha)$ is a monic polynomial over $\KK[\alpha]$, and $\disc{\vecb_0}=\pm\disc{F}$.

We first find $F$ as follows: Introduce new variables $T$ and $Z$.
Let $\preceq_{T,Z}$ be the elimination order for $T$ on $\mathcal{M}_{T,Z}$.
Let $\preceq=(\preceq_\vecX, \preceq_{T,Z})$ be an elimination order for $\vecX$ on $\mathcal{M}_{\vecX, T, Z}$.
Let $G$ be a \grobner basis for the ideal $I=\ideal{f_1,\dots,f_k,T-X,Z-\alpha}$ of $\KK[\vecX, T, Z]$ with respect to $\preceq$.
By \cref{lem:grobner-bounds}, we may assume $G\subseteq
\KK_{\leq d_4}[\vecX,T,Z]_{\leq d_3}$ for some $d_3=O_{d,n}(1)$ and $d_4=O_{d,d',n}(1)$.
By \cref{lem:elim-ideal}, $G\cap \KK[T,Z]$ is a \grobner basis for $I\cap \KK[T,Z]$ with respect to $\preceq_{T,Z}$. 
Choose $g\in G\cap\KK[T,Z]$ whose leading monomial
is $T^s$, where $s=\deg_T(F)$.
Such an element exists because $F$ is monic in $T$
and belongs to $I\cap\KK[T,Z]$.
Since $I\cap\KK[T,Z]$ is generated by $F(T,Z)$,
this $g$ is a nonzero scalar multiple of $F(T,Z)$.
Let $u\in \KK^\times$ be the leading coefficient of $g$. Then $u^{-1}g(T,\alpha)\in(\KK[\alpha])[T]$ is exactly the minimal polynomial $F$ of $\overline{X}$ over $\KK[\alpha]$.
By \cref{lem:frac-deg-bound}, $F\in 
 \KK_{\leq d_5}[T,\alpha]_{\leq d_3}$ for some $d_5=O_{d,d',n}(1)$.
Finally, by \cref{lem:disc-poly}, we have $\disc{\vecb_0}=\pm\disc{F}\in  \KK_{\leq d_2}[\alpha]_{\leq d_1}$, where $d_1=O_{d,n}(1)$ and $d_2=O_{d,d',n}(1)$.
\end{proof}

Keep the notations and assumptions in \cref{thm:integral-closure}.
For an algebraic extension $\LL$ of $\KK$, we denote by $I_\LL$ the ideal $\ideal{f_1,\dots,f_k}$ of $\LL[\vecX]$.

Suppose $A_{\LL}$ and $\order$ are both $(\LL[\alpha], \fr{A_{\LL}})$-orders and $A_{\LL}\subseteq \order$. By \cref{lem:disc-square} and the fact that $A_{\LL}=\LL[\vecX]/I_\LL$, every element in $\order$ is of the form $\frac{a+I_\LL}{\disc{\vecb_0}}$ for some $a\in \LL[\vecX]$. The following lemma shows how to translate a finite set of $\LL[\alpha]$-algebra generators of $\order$ in this form into an alternative description in terms of generators and relations over $\LL$, which is more convenient to work with. Here, we restrict our attention to the case where $\LL=\KK^{(e)}$.

\begin{lemma}\label{lem:ring-hom}
Let $\LL=\KK^{(e)}$ for some integer $e\geq 0$. Let $a_1,\dots,a_t\in\KK^{(e)}_{\leq d_2}[\vecX]_{\leq d_1}$, where $t=O_{d_1,n}(1)$. Assume that  $\frac{a_1+I_\LL}{\disc{\vecb_0}},\dots,\frac{a_t+I_\LL}{\disc{\vecb_0}} \in (A_\LL)_{\disc{\vecb_0}}$ generate an algebra $\order$ over $\LL[\alpha]$ that is a $(\LL[\alpha], \fr{A_{\LL}})$-order.
Then there exist $r_1,\dots,r_{t+1}\in\LL[\vecX, U]$ and $h_1,\dots,h_{k'}\in\LL[\vecZ]=\LL[Z_1,\dots,Z_{t+1}]$, where $k'=O_{d,d_1,n}(1)$, such that the following hold:
\begin{enumerate}[(1)]
\item $r_1,\dots,r_{t+1}\in 
\KK^{(e)}_{\leq d_2}[\vecX, U]_{\leq d_1+1}$ and $h_1,\dots,h_{k'}\in \KK^{(e)}_{\leq d_4}[\vecZ]_{\leq d_3}$, where $d_3=O_{d,d_1,n}(1)$ and $d_4=O_{d,p^e d',d_1,d_2,n}(1)$. Moreover, 
$r_{t+1}=\alpha$.
\item The map
\begin{align*}
\psi: \LL[\vecZ]/\ideal{h_1,\dots,h_{k'}}&\to \LL[\vecX, U]/\ideal{f_1,\dots,f_k,\disc{\vecb_0}\cdot U-1}\cong A_\LL[U]/\ideal{\disc{\vecb_0}\cdot U-1}
\\
Z_i+\ideal{h_1,\dots,h_{k'}}&\mapsto r_i + \ideal{f_1,\dots,f_k,\disc{\vecb_0}\cdot U-1}, \quad i=1,2,\dots, t+1
\end{align*} 
defines an injective $\LL$-linear ring homomorphism.
\item Let 
\[
\theta: A_\LL[U]/\ideal{\disc{\vecb_0}\cdot U-1}\to (A_\LL)_{\disc{\vecb_0}}
\]
be the $A_\LL$-linear ring isomorphism sending $U+\ideal{\disc{\vecb_0}\cdot U-1}$ to $1/\disc{\vecb_0}$ given by \cref{fact:inverse}. Then the image of $\theta\circ \psi$ is $\order$. 
\end{enumerate}
\end{lemma}

\begin{proof}
By assumption, $\frac{a_1+I_\LL}{\disc{\vecb_0}},\dots,\frac{a_t+I_\LL}{\disc{\vecb_0}}$ and $\alpha$ generate the algebra $\order$ over $\LL$.
Let $r_i=a_i\cdot U$ for $i\in [t]$ and $r_{t+1}=\alpha$. Then $r_1,\dots,r_{t+1}\in \KK^{(e)}_{\leq d_2}[\vecX, U]_{\leq d_1+1}$.  Also note that we may assume $t\le \binom{n+d_1}{d_1}$ by replacing ${a_1,\dots,a_t}$ with a linearly independent subset over $\LL$.

Identifying $\LL[\vecX, U]/\ideal{f_1,\dots,f_k,\disc{\vecb_0}\cdot U-1}$ with $(A_\LL)_{\disc{\vecb_0}}$ via $\theta$ and using the definition of $r_1,\dots,r_{t+1}$, we see that the image of the map $\LL[\vecZ]\to \LL[\vecX, U]/\ideal{f_1,\dots,f_k,\disc{\vecb_0}\cdot U-1}$ sending $Z_i$ to $r_i+\ideal{f_1,\dots,f_k,\disc{\vecb_0}\cdot U-1}$ for $i\in [t+1]$ is precisely $\order$. The kernel of this map is generated by a collection of polynomials $h_1,\dots,h_{k'}\in\LL[\vecZ]$.

By \cref{lem:ring-iso} and \cref{lem:disc-bound}, we may choose $h_1,\dots,h_{k'}$ to be in $ \KK^{(e)}_{\leq d_4}[\vecZ]_{\leq d_3}$ for sufficiently large $d_3=O_{d,d_1,n}(1)$ and $d_4=O_{d,p^e d',d_1,d_2,n}(1)$.
We may also assume $k'\leq \binom{t+1+d_3}{d_3}=O_{d,d_1,n}(1)$.
\end{proof}

The next lemma shows that, unless an order is already geometrically normal, we can obtain a strictly larger order by adjoining an integral element, possibly after a purely inseparable extension of bounded exponent.

\begin{lemma}\label{lem:add-new}
Let $\LL=\KK^{(e)}$ for some integer $e\geq 0$.
Let $a_1,\dots,a_t\in\KK^{(e)}_{\leq d_2}[\vecX]_{\leq d_1}$. Assume $\frac{a_1+I_\LL}{\disc{\vecb_0}},\dots,\frac{a_t+I_\LL}{\disc{\vecb_0}}$ generate an algebra $\order$ over $\LL[\alpha]$ that is a $(\LL[\alpha], \fr{A_{\LL}})$-order.
Assume moreover that
$\LL[\alpha,\overline{X}]\subseteq\order$.
Then there exist $d_3$, $d_4$, and $e'\geq e$ such that $p^{e'},d_3,d_4=O_{d,d',d_1,d_2,n,p^e}(1)$, and one of the following is true: 
\begin{enumerate}[(1)]
\item $\order$ is geometrically normal over $\LL$.
\item\label{item:add-new} There exists $a_{t+1}\in \KK^{(e')}_{\leq d_4}[\vecX]_{\leq d_3}$ such that $\frac{a_1+I_{\LL'}}{\disc{\vecb_0}},\dots,\frac{a_{t+1}+I_{\LL'}}{\disc{\vecb_0}}$  generate an algebra $\order'$ over $\LL'[\alpha]$ that is a $(\LL'[\alpha], \fr{A_{\LL'}})$-order strictly larger than $\order\otimes_\LL \LL'$, where $\LL'=\KK^{(e')}$.
\end{enumerate}
\end{lemma}

\begin{proof}
Assume that $\order$ is not geometrically normal over $\LL$.
We will prove that \cref{item:add-new} holds.

We may assume $t\le \binom{n+d_1}{d_1}$ by choosing a subset of linearly independent elements from $\{a_1,\dots,a_t\}$.
By \cref{lem:ring-hom}, there exist $k'=O_{d,d_1,n}(1)$, $r_1,\dots,r_{t+1}\in \KK^{(e)}_{\leq d_6}[\vecX, U]_{\leq d_5}$, and $h_1,\dots,h_{k'}\in \KK^{(e)}_{\leq d_6}[\vecZ]_{\leq d_5}$ for some  $d_5=O_{d,d_1,n}(1)$ and $d_6=O_{d,d',d_1,d_2,p^e,n}(1)$ such that $r_{t+1}=\alpha$, the ring homomorphism 
\begin{align*}
\psi: \LL[\vecZ]/\ideal{h_1,\dots,h_{k'}}&\to \LL[\vecX, U]/\ideal{f_1,\dots,f_k,\disc{\vecb_0}\cdot U-1}
\\
Z_i+\ideal{h_1,\dots,h_{k'}}&\mapsto r_i + \ideal{f_1,\dots,f_k,\disc{\vecb_0}\cdot U-1}, \quad i=1,2,\dots, t+1.
\end{align*} 
is injective, and the image of $\psi$ equals $\order$ if we identify $\LL[\vecX, U]/\ideal{f_1,\dots,f_k,\disc{\vecb_0}\cdot U-1}$ with $(A_\LL)_{\disc{\vecb_0}}$. We now use the above data to construct $a_{t+1}$.

If $\mathfrak{D}_{\order/\LL[\alpha]}$ were the unit ideal, then $\order$ would be finite \'etale over $\LL[\alpha]$
\cite[\href{https://stacks.math.columbia.edu/tag/0BJF}{Tag 0BJF}]{stacks-project},
hence smooth over $\LL$ and therefore geometrically normal over $\LL$
\cite[\href{https://stacks.math.columbia.edu/tag/056T}{Tag 056T}]{stacks-project},
contradicting our assumption.
So $1\notin\mathfrak{D}_{\order/\LL[\alpha]}$.
By \cref{lem:disc-facts}\,\eqref{item:disc3}, 
$\mathfrak{D}_{\order/\LL[\alpha]}$ is generated by some factor $Q(\alpha)\in \LL[\alpha]=\KK^{(e)}[\alpha]$ of $\disc{\vecb_0}$.
And by \cref{lem:disc-bound} and \cref{lem:factor},
$Q(\alpha)\in \KK^{(e)}_{\leq d_8}[\alpha]_{\leq d_7}$ for some $d_7=O_{d,n}(1)$ and $d_8=O_{d,d',n,p^e}(1)$.

Since $\order$ is finite free over $\LL[\alpha]$ and $Q(\alpha)$ is nonconstant, $\order/Q(\alpha)\order$ is a nonzero finite-dimensional $\LL$-algebra. So $Q(\alpha)\order$ is a zero-dimensional ideal of $\order$.
Identify $\LL[\vecZ]/\ideal{h_1,\dots,h_{k'}}$ with $\order$. Then $Q(Z_{t+1})+\ideal{h_1,\dots,h_{k'}}$ is identified with $Q(\alpha)$ since $r_{t+1}=\alpha$. 

In characteristic $p>0$, for each maximal ideal $\mm$
of $\order$ containing $Q(\alpha)$, let $e_\mm$ be the
largest nonnegative integer such that $p^{e_\mm}$ divides
$[\order/\mm:\LL]$, and let $e_0=\max_\mm e_\mm$.
By \cref{lem:bezout2},
\[
 p^{e_0}\leq
 \max\{\deg(Q),\deg(h_1),\dots,\deg(h_{k'})\}^{t+1}
 =O_{d,d_1,n,p^e}(1).
\]
In characteristic zero, let $e_0=0$.
Let $e'=e+e_0$, $\LL'=\KK^{(e')}$, and
$B=\order\otimes_\LL\LL'$.
Then $p^{e'}=O_{d,d_1,n,p^e}(1)$.

Since $A_{\LL'}$ is a domain, scalar extension identifies
$B$ with a subring of $(A_{\LL'})_{\disc{\vecb_0}}$.
In particular, $B$ is again an
$(\LL'[\alpha],\fr{A_{\LL'}})$-order.
Scalar extension preserves an integral basis and its
discriminant, so $Q(\alpha)$ also generates the
discriminant ideal of $B$.
Every maximal ideal $\mathfrak n$ of $B$
containing $Q(\alpha)$ contracts to a maximal ideal
$\mm$ of $\order$ containing $Q(\alpha)$.
Its residue field is generated over $\LL'$ by the image
of $\order/\mm$, and is separable over $\LL'$ by
\cref{lem:pth-root}; here we use that $\FF$ is perfect,
so $\LL'=\LL^{1/p^{e_0}}$ in positive characteristic.
In characteristic zero, separability is automatic.

If $B$ were integrally closed, then
\cref{lem:geometric-normality} would make it geometrically
normal over $\LL'$.
This would also make $\order$ geometrically normal over
$\LL$, since $\LL$ and $\LL'$ have the same perfect
closure
\cite[\href{https://stacks.math.columbia.edu/tag/037Y}
{Tag 037Y}]{stacks-project}, contradicting our assumption.
Hence $B$ is not integrally closed.
By \cref{thm:criterion} and \cref{lem:idealizer-integral},
there is a maximal ideal $\mm_{\LL'}$ of $B$
containing $Q(\alpha)$ whose idealizer is integral over
$B$ and strictly larger than it.
Its residue field is separable over $\LL'$ by the preceding
argument.

We now find an element in $\id{\order\otimes_{\LL}\LL'}{\mm_{\LL'}}\setminus (\order\otimes_{\LL}\LL')$. 
Let $c=Q(Z_{t+1})+\ideal{h_1,\dots,h_{k'}}\in \LL'[\vecZ]/\ideal{h_1,\dots,h_{k'}}$.
View $\mm_{\LL'}$ as a maximal ideal of $\LL'[\vecZ]/\ideal{h_1,\dots,h_{k'}}$ containing $c$. 
Let $\widetilde{\mm}_{\LL'}$ be the preimage of $\mm_{\LL'}$ in $\LL'[\vecZ]$ under the natural quotient map, which is a maximal ideal of $\LL'[\vecZ]$ and contains the zero-dimensional ideal $\ideal{h_1,\dots,h_{k'},Q(Z_{t+1})}$.
As $(\order\otimes_{\LL} \LL')/\mm_{\LL'}$ is separable, by \cref{lem:extract}, $\widetilde{\mm}_{\LL'}$ is generated by polynomials in $\KK^{(e')}_{\leq d_{10}}[\vecZ]_{\leq d_9}$, where $d_9=O_{d,d_1,n}(1)$ and $d_{10}=O_{d,d',d_1,d_2,n,p^e}(1)$.

As $Q(\alpha)$ generates $\mathfrak{D}_{\order/\LL[\alpha]}$,
identifying $\LL'[\vecZ]/\ideal{h_1,\dots,h_{k'}}$ with $\order\otimes_\LL \LL'$ and applying
 \cref{lem:disc-square} and \cref{lem:idealizer-integral} shows that $c \cdot \id{\LL'[\vecZ]/\ideal{h_1,\dots,h_{k'}}}{\mm_{\LL'}}\subseteq \LL'[\vecZ]/\ideal{h_1,\dots,h_{k'}}$.
Let $J\subseteq \LL'[\vecZ]$ be the preimage of $c \cdot \id{\LL'[\vecZ]/\ideal{h_1,\dots,h_{k'}}}{\mm_{\LL'}}$ under the natural quotient map $\LL'[\vecZ]\to\LL'[\vecZ]/\ideal{h_1,\dots,h_{k'}}$.
By \cref{lem:find-idealizer}, $J$ has a \grobner basis $G_1$ contained in $\KK^{(e')}_{\leq d_{12}}[\vecZ]_{\leq d_{11}}$, where $d_{11}=O_{d,d_1,n}(1)$ and $d_{12}=O_{d,d',d_1,d_2,n,p^e}(1)$.

Let
$R(\alpha)=\disc{\vecb_0}/Q(\alpha)\in\LL[\alpha]$.
Choose $\gamma\in G_1$ whose image in
$\order\otimes_{\LL}\LL'$ does not belong to
$Q(\alpha)(\order\otimes_{\LL}\LL')$.
Then
\[
 \eta:=\frac{\psi_{\LL'}(\gamma)}{Q(\alpha)}
 \in
 \id{\order\otimes_{\LL}\LL'}{\mm_{\LL'}}
 \setminus(\order\otimes_{\LL}\LL'),
\]
where $\psi_{\LL'}$ denotes the map $\psi\otimes_{\LL}\LL'$, and we identify its image with $\order\otimes_\LL\LL'$.

By \cref{lem:idealizer-integral}, $\eta$ is integral over $\order\otimes_\LL\LL'$, and hence over $\LL'[\alpha]$.
Applying \cref{lem:disc-square} to the initial order
$\LL'[\alpha,\overline{X}]$ gives
\[
 \disc{\vecb_0}\eta
 \in\LL'[\alpha,\overline{X}]\subseteq A_{\LL'}.
\]
Its image in $(A_{\LL'})_{\disc{\vecb_0}}$ is represented by $R(\alpha)\gamma(r_1,\dots,r_{t+1})$.
Let $G_2$ be a Gr\"obner basis of
$\ideal{f_1,\dots,f_k,\disc{\vecb_0} U-1}$ with respect to an
elimination order for $U$ that is degree-compatible in
$\vecX$. Define $a_{t+1}$ to be the remainder of
$R(\alpha)\gamma(r_1,\dots,r_{t+1})$ modulo $G_2$.
Since $R(\alpha)\gamma(r_1,\dots,r_{t+1})$ represents $\disc{\vecb_0}\eta$, which belongs to $A_{\LL'}$, $a_{t+1}$ is congruent
modulo $\ideal{f_1,\dots,f_k,\disc{\vecb_0}U-1}$ to some
$b\in\LL'[\vecX]$.
The remainder of $b$ modulo $G_2$ must also be $a_{t+1}$.
As the monomial order eliminates $U$, reducing $b$ cannot
introduce $U$, so $a_{t+1}\in \LL'[\vecX]$.

The bounds from \cref{lem:factor},
\cref{lem:grobner-bounds}, and
\cref{lem:reduction-alg2} show that
$a_{t+1}\in\KK^{(e')}_{\leq d_4}[\vecX]_{\leq d_3}$
for suitable
$d_3,d_4=O_{d,d',d_1,d_2,n,p^e}(1)$.
Moreover,
$\frac{a_{t+1}+I_{\LL'}}{\disc{\vecb_0}}=\eta$.
Adjoining this element therefore gives the required
strictly larger order.
\end{proof}

Finally, we prove \cref{thm:integral-closure} by repeatedly applying \cref{lem:add-new}. The key idea is that the ``difference'' between the initial order and the maximal order (i.e., the integral closure) can be bounded via the discriminant, which in turn yields an upper bound on the number of steps (see \cref{lem:disc-facts}\,\eqref{item:disc3} and \cref{lem:disc-square}). Since the increase in the parameters at each step is bounded by \cref{lem:add-new}, the final parameters can be bounded as well.

\begin{proof}[Proof of \cref{thm:integral-closure}]
We start from the order $\order_0=\KK[\alpha,\overline{X}]$. At each step, apply \cref{lem:add-new}. Unless $\order_{i-1}$ is geometrically normal over $\KK^{(e_{i-1})}$, this gives an order $\order_i$ that strictly contains $\order_{i-1}\otimes_{\KK^{(e_{i-1})}}\KK^{(e_i)}$. We stop when the order is geometrically normal over its current base field. Let $\tau$ denote the number of steps; the argument below bounds $\tau$ and proves termination.

For $i=0,1,\dots,\tau$, the generators of $\order_i$ are $(\alpha_1+I_{\KK^{(e_i)}})/\disc{\mathbf{b}_0},\dots, (\alpha_{i+1}+I_{\KK^{(e_i)}})/\disc{\mathbf{b}_0}$, where $\alpha_1,\dots,\alpha_{i+1} \in \KK^{(e_i)}_{\leq d_i'}[\vecX]_{\leq d_i}$.
By \cref{lem:add-new}, the parameters $d_i,d_i',p^{e_i}$ satisfy
\[
d_0,d_0',p^{e_0} = O_{d,d',n}(1)
\]
and
\[
d_i,d_i',p^{e_i} \leq F_{d,d',n}(d_{i-1},d'_{i-1},p^{e_{i-1}})\quad \text{ for } i>0,
\]
where $F_{d,d',n}$ is some non-decreasing function depending only on $d$, $d'$, and $n$.\footnote{We say a function $F$ is non-decreasing if it is non-decreasing in each variable when the other variables are fixed. Note that for any function $F_0(\cdot,\cdot,\cdot)$, we can find a non-decreasing function $F(\cdot,\cdot,\cdot)$ that bounds $F_0$ from above by letting $F(a,b,c)=\sup_{a'\leq a,b'\leq b,c'\leq c} F_0(a',b',c')$.}
Define $F(0)=\max(d_0,d_0',p^{e_0})$ and $F(i)=F_{d,d',n}(F(i-1),F(i-1),F(i-1))$.
Then $d_i,d'_i,p^{e_i}\leq F(i)$.

Scalar extension preserves the discriminant of an integral basis and hence its degree.
By \cref{lem:disc-facts}\,\eqref{item:disc3}, replacing $\order_{i-1}$ by $\order_{i}$ decreases the degree of the generator of the discriminant ideal by at least two.
So the process must terminate with a geometrically normal order after at most
$\tau\leq\deg_\alpha(\disc{\vecb_0})/2=O_{d,n}(1)$ steps: otherwise \cref{lem:add-new} would give another strict enlargement and eventually force a negative discriminant degree.
Therefore, the final parameters $d_\tau$, $d'_\tau$, and $p^{e_\tau}$ are bounded by $F(\tau)=O_{d,d',n}(1)$.  It also follows by construction that $m,k'=O_{d,d',n}(1)$.

Note that  $\order_\tau$ equals the integral closure of $A_{\KK^{(e_\tau)}}$. This follows from \cref{cor:integral-subring} and the fact that $A_{\KK^{(e_\tau)}}$ is integral over $\KK^{(e_\tau)}[\alpha]$.

To find the polynomials $g_i$, use the map $\psi$ associated with $\order_\tau$ given by \cref{lem:ring-hom}, and then use \cref{lem:ring-iso}\,\eqref{item:poly3} to compute the inverse of $X_i+\ideal{f_1,\dots,f_k,\disc{\vecb_0}\cdot U-1}$ under $\psi$ for $i\in [n]$.
\end{proof}

\section{PIT via Normalization}\label{sec:pit}

In this section, let $\FF$ be a field and
let $\KK=\FF(\vecY)=\FF(Y_{0,1},\dots, Y_{0,n},Y_{1,1},\dots,Y_{1,n},Y_{2,1},\dots,Y_{2,n})$.

\subsection{Restricting to a Generic Affine Plane}

We start by discussing the algebra resulting from restricting to a generic affine plane.

\begin{definition}[Restriction to a generic affine plane]\label{def:res}
For a polynomial $f\in\FF[\vecX]=\FF[X_1,\dots,X_n]$ over a field $\FF$, define $\res(f)$ to be the polynomial
\begin{equation}\label{eq:res}
f\left(Y_{0,1}+Y_{1,1} Z_1+Y_{2,1} Z_2, \dots, Y_{0,n}+Y_{1,n}Z_1+Y_{2,n} Z_2\right)\in \FF[\vecY][Z_1,Z_2]\subseteq \KK[Z_1,Z_2].
\end{equation}
\end{definition}

\begin{lemma}\label{lem:irreducibility}
Let $n\geq 2$.
Suppose $\FF$ is algebraically closed and $f\in\FF[\vecX]=\FF[X_1,\dots,X_n]$ is irreducible over $\FF$. Then $\res(f)\in \KK[\vecZ]$ is absolutely irreducible over $\KK=\FF(\vecY)$.
\end{lemma}

For a proof of \cref{lem:irreducibility}, see \cite[Lemma~7]{Kal95}.\footnote{The statement \cite[Lemma~7]{Kal95} is slightly different, where the expression $Y_{0,1}+Y_{1,1} Z_1+Y_{2,1} Z_2$ in \eqref{eq:res} is replaced by $Y_{0,1}+Z_1$, resulting in another polynomial $\varphi_2$. However, the proof can be adapted to prove \cref{lem:irreducibility}.
Alternatively, we can recover \eqref{eq:res} from $\varphi_2$ by performing the invertible $\KK$-linear variable substitution $Z_1\mapsto Y_{1,1} Z_1+Y_{2,1} Z_2$, followed by applying the field automorphism of $\KK$ that maps $Y_{1,i}\mapsto Y_{1,i}/Y_{1,1}$ and $Y_{2,i}\mapsto Y_{2,i}-Y_{2,1}Y_{1,i}/Y_{1,1}$ for $i\in [n]\setminus\{1\}$. Both transformations preserve absolute irreducibility.}

The next lemma states that for any algebraic extension $\LL$ of $\KK$, the ring $\LL[Z_1,Z_2]/\ideal{\res(f)}$ is an integral domain whose field of fractions is a finite separable extension of the function field $\LL(\overline{X}_{c_1,c_2})$, where $\overline{X}_{c_1,c_2}$ is defined below, for $(c_1,c_2)\in(\FF^\times)^2$. Together with Noether normalization, this lemma will later be used to show that the assumptions of \cref{thm:integral-closure} can always be satisfied.

\begin{lemma}\label{lem:irred-sep}
Let $n\geq 2$.
Suppose $\FF$ is algebraically closed and $f\in\FF[\vecX]=\FF[X_1,\dots,X_n]$ is an irreducible polynomial over $\FF$ of degree $d>0$.
Let $\LL$ be an algebraic extension of $\KK=\FF(\vecY)$.
Then $\LL[Z_1,Z_2]/\ideal{\res(f)}$ is an integral domain.
Moreover, for every 
$(c_1,c_2)\in(\FF^\times)^2$, the field of fractions $\EE$ of $\LL[Z_1,Z_2]/\ideal{\res(f)}$ is a finite separable extension of $\LL(
\overline{X}_{c_1,c_2})$, where $\overline{X}_{c_1,c_2}:=c_1 Z_1+c_2 Z_2+\ideal{\res(f)}\in \LL[Z_1,Z_2]/\ideal{\res(f)}\subseteq\EE$.
\end{lemma}

\begin{proof}
By \cref{lem:irreducibility}, $\res(f)$ is absolutely irreducible over $\KK$. So it is irreducible over $\LL$. Therefore, $\LL[Z_1,Z_2]/\ideal{\res(f)}$ is an integral domain.

Consider $(c_1,c_2)\in(\FF^\times)^2$ and let $X_{c_1,c_2}:=c_1 Z_1+c_2 Z_2$.
Then $Z_2=c_2^{-1}(X_{c_1,c_2}-c_1 Z_1)$. So $\LL[Z_1,Z_2]=\LL[Z_1,X_{c_1,c_2}]$. 
Performing the substitution $Z_2=c_2^{-1}(X_{c_1,c_2}-c_1 Z_1)$ in \eqref{eq:res}, we have
\[
\res(f)=f(T_1,\dots, T_n),
\]
where 
\begin{equation}\label{eq:Ti}
T_i:=Y_{0,i}+Y_{1,i}Z_1+c_2^{-1}Y_{2,i}(X_{c_1,c_2}-c_1 Z_1),\quad i=1,2,\dots,n
\end{equation}
View $\res(f)$ as a univariate polynomial in $Z_1$ over $\LL[X_{c_1,c_2}]$ (which is fine as $Z_1$ and $X_{c_1,c_2}$ are algebraically independent).
Its derivative, which we denote by $D_{\res(f)}$, can be determined by the chain rule:
\begin{equation}\label{eq:res-partial}
D_{\res(f)}=\sum_{i=1}^n \frac{\partial f}{\partial X_i}(T_1,\dots,T_n)\frac{\partial T_i}{\partial Z_1}=\sum_{i=1}^n \frac{\partial f}{\partial X_i}(T_1,\dots,T_n)\cdot (Y_{1,i}-c_2^{-1}c_1Y_{2,i}).
\end{equation}

As $\FF$ is algebraically closed (and hence a perfect field) and $f$ is irreducible over $\FF$, we know $\frac{\partial f}{\partial X_i}\neq 0$ for some $i\in [n]$.
Choose $i_0\in [n]$ and a monomial $m=\prod_{i=1}^n X_i^{e_i}$ that appears in $\frac{\partial{f}}{\partial X_{i_0}}$ such that $\deg_{\vecX}(m)$ is maximized over all choices of $i_0$ and $m$. 
Then by \eqref{eq:Ti} and \eqref{eq:res-partial}, $D_{\res(f)}$, viewed as a polynomial in $\vecY$, $Z_1$, and $X_{c_1,c_2}$ over $\FF$, has a monomial $X_{c_1,c_2}^{\deg_{\vecX}(m)}\left(\prod_{i=1}^n Y_{2,i}^{e_i}\right) Y_{1,i_0}$ 
that comes from the term $\frac{\partial f}{\partial X_{i_0}}(T_1,\dots,T_n)\cdot Y_{1,i_0}$ in  \eqref{eq:res-partial} and is not canceled by other monomials.
So $D_{\res(f)}\neq 0$.

As the degree of $D_{\res(f)}$ is smaller than that of $\res(f)$, both viewed as univariate polynomials in $Z_1$ over $\LL[X_{c_1,c_2}]$, we have $D_{\res(f)}\not\in \ideal{\res(f)}$.

Let $f_d=f_d(X_1,\dots,X_n)$ be the homogeneous degree-$d$ component of $f$, where $d=\deg(f)$. 
View $\res(f)$ as a univariate polynomial in $Z_1$ over $\LL[X_{c_1,c_2}]$. 
Then we can write $\res(f)=\sum_{i=0}^d a_i Z_1^i$, where $a_i\in\LL[X_{c_1,c_2}]$.
By $\eqref{eq:Ti}$, we have 
\[
a_d=f_d(Y_{1,1}-c_2^{-1}c_1Y_{2,1},\dots,Y_{1,n}-c_2^{-1}c_1Y_{2,n})\neq 0
\]
We also have $a_d\not\in\ideal{\res(f)}$ as $d>0$ and $a_d$ is independent of $Z_1$.

Note that $\LL[X_{c_1,c_2}]\cap \ideal{\res(f)}=0$ since $d>0$.
So we have an inclusion 
\[
\LL[X_{c_1,c_2}]\cong \LL[X_{c_1,c_2}]/(\LL[X_{c_1,c_2}]\cap \ideal{\res(f)})\hookrightarrow \LL[Z_1,Z_2]/\ideal{\res(f)} 
\]
Taking the fields of fractions, we see that $\LL(X_{c_1,c_2})$ may be identified with $\LL(\overline{X}_{c_1,c_2})\subseteq \EE$.
And the minimal polynomial of $Z_1$ over $\LL(X_{c_1,c_2})$ is $F(T):=\sum_{i=0}^d (a_i/a_d) T^i\in \LL(X_{c_1,c_2})[T]$.
We also have $F'(Z_1+\ideal{\res(f)})=(D_{\res(f)}+\ideal{\res(f)})/(a_d+\ideal{\res(f)})\neq 0$ since $a_d,D_{\res(f)}\not\in\ideal{\res(f)}$.
It follows that $Z_1+\ideal{\res(f)}\in\EE$ is separable over  $\LL(
\overline{X}_{c_1,c_2})$.
A symmetric argument shows that $Z_2+\ideal{\res(f)}\in \EE$ is   separable over $\LL(
\overline{X}_{c_1,c_2})$.
As $\EE$ is generated by $Z_1+\ideal{\res(f)}$ and $Z_2+\ideal{\res(f)}$, we conclude that $\EE$ is a finite separable extension of $\LL(
\overline{X}_{c_1,c_2})$. 
\end{proof}

\subsection{Proof of the Main Theorem}\label{sec:proof-main}

For convenience, we state our main theorem again.

\begin{theorem}[Main theorem, homogeneous version]\label{thm:main}
Let $C_{n,d,k,\delta,\FF}$ be the set of polynomials $F\in\FF[\vecX]=\FF[X_1,\dots,X_n]$ over a field $\FF$ satisfying the following conditions:
\begin{enumerate}[(1)]
    \item $F$ can be expressed as a sum $F=\sum_{i=0}^{k_0-1} F_i$, where $k_0\leq k$, $F_i=\prod_{j=1}^{m_i} f_{i,j}$ for $i\in \{0,1,\dots,k_0-1\}$, and each $f_{i,j}\in \FF[\vecX]$ is a nonzero homogeneous polynomial of degree at most $\delta$.
    \item $\deg(F_i)=d_0$ for some $d_0\leq d$ and all $i\in \{0,1,\dots,k_0-1\}$.
    \item\label{item:hom3} $F_i$ is squarefree over $\overline{\FF}$ for some $i\in \{0,1,\dots,k_0-1\}$, meaning that the irreducible factors of $F_i$ over $\overline{\FF}$ are distinct.
\end{enumerate}
Then there exists an explicit $(nd)^{O_{\delta}(1)}$-sized hitting set $\mathcal{H}\subseteq \overline{\FF}^n$ for $C_{n,d,3,\delta,\FF}$.
\end{theorem}

\paragraph{Assumptions.}

We make some further assumptions to simplify the discussion, and briefly justify them:

\begin{enumerate}[(1)]
\item \label{item:assumption1} $F$ is non-constant. Otherwise, any nonempty set will be a hitting set for $F$ by definition.
\item $\FF$ is algebraically closed. This can be guaranteed by extending $\FF$ to $\overline{\FF}$.
\item $k_0=3$. This is because \cref{thm:main} is known to hold when $k_0=1$ or $k_0=2$ even without the condition of squarefreeness. So $F=F_0+F_1+F_2$.
\item $n\geq 3$. This is because there exists an explicit hitting set of size $d^{O(1)}$ for constant-variate polynomials of degree at most $d$. This follows from either Kronecker substitutions or explicit hitting set constructions for sparse polynomials \cite{KS01}.
\item $F_0$ is squarefree. 
\cref{item:hom3} of \cref{thm:main} states that some $F_i$ is squarefree. 
So this assumption can be guaranteed by permuting the summands $F_i$.
\item The GCD of $F_0$, $F_1$, and $F_2$ is $1$. This is because we can take out the GCD of $F_0$, $F_1$, and $F_2$, if it is nontrivial, from each $F_i$. See \cite{Gup14} for a more detailed discussion.
\item \label{item:assumption6} $f_{i,j}$ is an irreducible non-constant polynomial for $i=0,1,2$ and $j\in [m_i]$. This can be guaranteed by replacing $f_{i,j}$ with its irreducible factors over $\FF$, and absorbing the constant $f_{i,j}$ into other factors.
\item \label{item:assumption7}  There exists $j_0\in [m_0]$ such that $f_{0,j_0}$ does not divide any polynomial in $\{F_1, F_2, F_1+F_2\}$. By permuting the factors of $F_0$, we may assume, without loss of generality, that $j_0=1$.
\end{enumerate}

Assumption~\eqref{item:assumption7}  is justified by the following lemma.
  
\begin{lemma}\label{lem:justify}
There exists an explicit $(nd)^{O_\delta(1)}$-sized hitting set for all $F\in C_{n,d,3,\delta,\FF}$ satisfying Assumptions \eqref{item:assumption1}--\eqref{item:assumption6} and the additional assumption that every irreducible factor of $F_0$ over $\FF$ divides some polynomial in $\{F_1, F_2, F_1+F_2\}$. 
\end{lemma}

\begin{proof}
If $F_1+F_2=cF_0$ for some $c\in\FF$, then $F=F_0+F_1+F_2=(c+1)F_0$, which is a product of polynomials of degree at most $\delta$. This is the case where $k_0=1$, which is already solved as mentioned. So we assume this is not the case. In particular, $F_1+F_2\neq 0$. Then, $\deg(F_1+F_2)=\deg(F_0)=d_0$.
As $F_1+F_2$ is not of the form $c F_0$ with $c\in \FF$ but has  degree $d_0=\deg(F_0)$, we see that $F_0$ does not divide $F_1+F_2$. As $F_0=\prod_{j=1}^{m_0} f_{0,j}$ is squarefree, we conclude that for some $j\in [m_0]$, the factor $f_{0,j}$ does not divide $F_1+F_2$. Fix such $j$. 

By assumption, $f_{0,j}$ divides $F_1$ or $F_2$, but not both since it does not divide $F_1+F_2$.
By symmetry, we may assume $f_{0,j}$ divides $F_1$ but not $F_2$. By Assumption~\eqref{item:assumption6}, $V(f_{0,j})$ is an irreducible hypersurface of $\aff^n$. 
As $f_{0,j}$ divides $F_1$, we have $V(f_{0,j}, F_1)=V(f_{0,j})$ and therefore its codimension in $\aff^n$ equals one.
On the other hand, as $f_{0,j}$ does not divide $F_2=\prod_{i=1}^{m_2} f_{2,i}$, the codimension of $V(f_{0,j}, f_{2,i})$ in $\aff^n$ equals two for all $i\in [m_2]$.
We also have $\deg(V(f_{0,j}))=\deg(f_{0,j})\leq \delta$ and $\deg(V(f_{0,j}, f_{2,i}))\leq \delta^2$ for $i\in [m_2]$ by \bezout's inequality.

Guo \cite[Theorem~1.6]{Guo24} showed how to explicitly construct 
a set $\mathcal{H}$ of affine planes in $\aff^n$ of size at most $n^{O(\delta^2)}d^{O(1)}$ such that for at least one $P\in \mathcal{H}$, we have 
$\dim (V(f_{0,j})\cap P)=1$ and 
$\dim (V(f_{0,j}, f_{2,i})\cap P)\leq 0$ for all $i\in [m_2]$.
Fix such $P$. Then
\[
\dim (V(f_{0,j},F_1)\cap P)=\dim (V(f_{0,j})\cap P)=1
\] and 
\[
\dim (V(f_{0,j},F_2)\cap P)
=\dim\left(\bigcup_{i=1}^{m_2}( V(f_{0,j}, f_{2,i})\cap P)\right)
=\max_{i\in [m_2]}(\dim (V(f_{0,j}, f_{2,i})\cap P)) \leq 0.
\]

Note that this implies that $F_1+F_2$ is not identically zero when restricting to $V(f_{0,j})\cap P$, since otherwise we would have 
\[
V(f_{0,j},F_1)\cap P=V(f_{0,j})\cap P\cap V(F_1)=V(f_{0,j})\cap P\cap V(-F_2)
=V(f_{0,j},F_2)\cap P,
\]
contradicting the fact that $\dim (V(f_{0,j},F_1)\cap P)\neq \dim(V(f_{0,j},F_2)\cap P)$.

As $f_{0,j}$ divides $F_0$, we have $F=F_0+F_1+F_2\equiv F_1+F_2\not\equiv 0$ when restricting to $V(f_{0,j})\cap P$.
So $F$ is not identically zero on $P$. As $F$ restricted to $P$ is a bivariate polynomial of degree $\deg(F)=d_0\leq d$, we know how to construct an explicit hitting set of size $d^{O(1)}$ on $P$ for $F$. Finally, while we do not know which affine plane $P\in\mathcal{H}$ works, we could construct an explicit hitting set for each affine plane in $\mathcal{H}$ and then take their union as the final hitting set, whose size is at most $|\mathcal{H}|\cdot d^{O(1)}\leq n^{O(\delta^2)}d^{O(1)}$.
\end{proof}

From now on, we assume Assumptions \eqref{item:assumption1}--\eqref{item:assumption7}.
Next, we introduce the projective analogue of \cref{def:res}.

\begin{definition}[Restriction to a generic projective plane]\label{def:pres}
Let $f\in\FF[\vecX]=\FF[X_1,\dots,X_n]$ be a homogeneous polynomial.
Define $\pres(f)$ to be the homogeneous polynomial
\begin{equation}\label{eq:pres}
f\left(Y_{0,1}\widehat{Z}_0+Y_{1,1} \widehat{Z}_1+Y_{2,1} \widehat{Z}_2, \dots, Y_{0,n} \widehat{Z}_0+Y_{1,n}\widehat{Z}_1+Y_{2,n} \widehat{Z}_2\right)\in \FF[\vecY][\widehat{Z}_0,\widehat{Z}_1,\widehat{Z}_2]\subseteq \KK[\widehat{Z}_0,\widehat{Z}_1,\widehat{Z}_2].
\end{equation}
\end{definition}

Consider the projective space $\proj^2_{\KK}$ over $\KK=\FF(Y_{0,1},Y_{1,1},Y_{2,1},\dots,Y_{0,n},Y_{1,n},Y_{2,n})$ with homogeneous coordinates $\widehat{Z}_0$, $\widehat{Z}_1$, and $\widehat{Z}_2$. For $i=0,1,2$, let $U_i\cong \aff^2_\KK$ be the affine open chart of $\proj^2_{\KK}$ defined by $\widehat{Z}_i\neq 0$. 

The homogeneous polynomial $\pres(f_{0,1})$ defines a projective curve $C\subseteq \proj^2_{\KK}$.
Identify $U_0$ with $\aff^2_\KK$ and let $Z_1=\widehat{Z}_1/\widehat{Z}_0$ and $Z_2=\widehat{Z}_2/\widehat{Z}_0$ be the coordinates of $U_0$. By \cref{def:res} and \cref{def:pres}, $C\cap U_0\subseteq U_0$ is defined precisely by the polynomial $\res(f_{0,1})$.
By Assumption~\eqref{item:assumption6}, $f_{0,1}$ is irreducible over $\FF$.
So by \cref{lem:irred-sep}, $\res(f_{0,1})$ is absolutely irreducible.
Thus, the affine curve $C\cap U_0$ and the projective curve $C$ are both absolutely irreducible.

Let $g=\pres(F_1)/\pres(F_2)$, which is a homogeneous rational function of degree $\deg(F_1)-\deg(F_2)=0$ on the projective space $\proj^2_{\KK}$.
By Assumption~\eqref{item:assumption7}, $f_{0,1}$ divides neither $F_1$ nor $F_2$.

\begin{claim}\label{claim:indivisibility}
$\pres(f_{0,1})$ divides neither $\pres(F_1)$ nor $\pres(F_2)$.
\end{claim}
\begin{proof}[Subproof of \cref{claim:indivisibility}]
Write $f=f_{0,1}$ and $R=\FF[\vecY]$, so that
$\KK=\operatorname{Frac}(R)$.
Since $f$ is irreducible, $f(\vecY_0)$ remains irreducible
after adjoining the variables $\vecY_1,\vecY_2,Z_1,Z_2$.
The substitution
$Y_{0,j}\mapsto Y_{0,j}+Y_{1,j}Z_1+Y_{2,j}Z_2$,
fixing all other variables, is an $\FF$-algebra automorphism
of $R[Z_1,Z_2]$ sending $f(\vecY_0)$ to $\res(f)$.
Thus $\res(f)$ is irreducible in $R[Z_1,Z_2]$.
Since it also has positive degree in $(Z_1,Z_2)$, its
coefficients have no common nonunit factor in $R$;
that is, it is primitive over $R$.

Suppose that $\pres(f)\mid\pres(F_i)$ for some $i\in\{1,2\}$.
Setting $\widehat Z_0=1$ gives
$\res(f)\mid\res(F_i)$ in $\KK[Z_1,Z_2]$.
By Gauss's lemma, divisibility by a primitive polynomial
over a UFD can be checked over its fraction field.
Hence this divisibility already holds in $R[Z_1,Z_2]$.
Setting $Z_1=Z_2=0$ and $\vecY_1=\vecY_2=0$ then yields
$f(\vecY_0)\mid F_i(\vecY_0)$, contradicting Assumption~\eqref{item:assumption7}.
\end{proof}

We divide the proof of the main theorem into two cases. 

\paragraph{The easier case.}

In the easier case, $\pres(F_1)$ and $\pres(F_2)$ have the same order of vanishing at every point of the normalized curve $\widetilde{C}$, so that the quotient $\pres(F_1)/\pres(F_2)$ is well-defined everywhere on $\widetilde{C}$. We show that in this case, $g=\pres(F_1)/\pres(F_2)$ is a constant on $\widetilde{C}$, and in fact, $F_1/F_2$ is a constant modulo $f_{0,1}$, allowing us to merge $F_1$ and $F_2$ in the analysis. The following lemma and its proof make this precise.

\begin{lemma}\label{lem:easy-case}
There exists an explicit set $S\subseteq\FF^n$ of size $n^{O(\delta^2)}d^{O(1)}$ independent of $F$ such that, if the restriction\footnote{The restriction of $g$ to $\widetilde{C\cap U_i}$ means first restricting $g$ to $C\cap U_i$, and then viewing it as a (rational) function on $\widetilde{C\cap U_i}$. Note that $C\cap U_i$ and $\widetilde{C\cap U_i}$ share the same function field, $\fr{\KK[C\cap U_i]}$.} of $g=\pres(F_1)/\pres(F_2)$ to the normalization $\widetilde{C\cap U_i}$ of $C\cap U_i$ is regular for $i=0,1,2$, then $S$ is a hitting set for $F$.
\end{lemma}

\begin{proof}
Suppose the restriction of $g=\pres(F_1)/\pres(F_2)$ to $\widetilde{C\cap U_i}$ is regular for $i=0,1,2$. These normalizations $\widetilde{C\cap U_i}$ with $i=0,1,2$ glue together to form the normalization $\widetilde{C}$ of $C$; see \cite[Proposition~4.13]{Eis13} and the discussion thereafter. Then $g|_C$ is a regular function on $\widetilde{C}$.

Let $A_i=\KK[C\cap U_i]$. By \cref{lem:irreducibility}, even after changing the base field $\KK$ to $\overline{\KK}$, the coordinate ring $A_i\otimes_{\KK}\overline{\KK}=\overline{\KK}[C\cap U_i]$ is an integral domain.
As $\overline{\KK}$ is free over $\KK$, we have an injection
\[
 \widetilde{A_i}\otimes_{\KK}\overline{\KK}
 \hookrightarrow
 \fr{A_i}\otimes_{\KK}\overline{\KK}.
\]
The ring on the right is a localization of
$A_i\otimes_{\KK}\overline{\KK}$ and hence is a domain.
Thus $\widetilde{A_i}\otimes_{\KK}\overline{\KK}$
is a domain as well, proving that $\widetilde C$
is \emph{geometrically integral} \cite[\href{https://stacks.math.columbia.edu/tag/05DW}{Tag 05DW}, \href{https://stacks.math.columbia.edu/tag/0366}{Tag 0366}]{stacks-project}.
Moreover, $\widetilde{C}$ is projective and hence \emph{proper} over $\KK$.\footnote{Properness is a version of compactness in algebraic geometry. See \cite[\S11.4]{Vak24}. Here we use the fact that regular functions on a geometrically integral and proper variety are constants, which is an analog of the fact in complex geometry that holomorphic functions on a connected complex compact manifold are constants.} Any regular function on a geometrically integral and proper variety lives in the base field \cite[\href{https://stacks.math.columbia.edu/tag/0BUG}{Tag 0BUG}]{stacks-project}. So $g|_{\widetilde{C}}\in\KK$. It follows that $g|_{C}\in\KK=\FF(\vecY)$. 

View $F_0,F_1,F_2$ as polynomials on $\aff^{n}_\FF$.
Let $H$ be the irreducible hypersurface in $\aff^{n}_\FF$ defined by $f_{0,1}$.
As $f_{0,1}$ divides neither $F_1$ nor $F_2$, we know $F_1/F_2$ restricts to a nonzero rational function $(F_1/F_2)|_H\in \FF(H)$ on $H$.

We claim that $(F_1/F_2)|_H \in \FF$.
Assume that this is not the case. Then for a general point $\vecx=(\veca,\vecb,\vecc)\in\FF^{3n}$, we have
\begin{enumerate}
\item $\veca, \vecb,\vecc\in\FF^n$ are linearly independent.
\item $g|_C\in \KK=\FF(\vecY)$ is regular at $\vecx$.
\item The affine line passing through $\veca$ and $\vecb$ intersects $H$ at some point $\vecu$.
\item The affine line passing through $\veca$ and $\vecc$ intersects $H$ at some point $\vecv$.
\item $F_1/F_2$ is regular at $\vecu$ and $\vecv$ but $(F_1/F_2)(\vecu)\neq (F_1/F_2)(\vecv)$.
\end{enumerate}
Here, the last item holds because, for a fixed general vector $\veca$, both $\vecu$ and $\vecv$ can vary independently over a dense open subset of $H$, depending on $\vecb$ and $\vecc$, respectively.

Fix such $\vecx=(\veca,\vecb,\vecc)$.
We may write $\vecu=\alpha \veca+(1-\alpha)\vecb$ and
$\vecv=\beta\veca+(1-\beta)\vecc$ for some $\alpha,\beta\in\FF$. 
By definition, as $g|_C(\vecx)\in\KK$ does not depend on $(\widehat{Z_0},\widehat{Z_1},\widehat{Z_2})$, we have 
\[
\left.(F_1/F_2)(\vecu)=\frac{\pres(F_1)}{\pres(F_2)}\right\vert_{\vecY=\vecx, (\widehat{Z_0},\widehat{Z_1},\widehat{Z_2})=(\alpha,1-\alpha,0)}=g|_C(\vecx),
\]
and similarly,
\[
\left.(F_1/F_2)(\vecv)=\frac{\pres(F_1)}{\pres(F_2)}\right\vert_{\vecY=\vecx, (\widehat{Z_0},\widehat{Z_1},\widehat{Z_2})=(\beta,0,1-\beta)}=g|_C(\vecx).
\]
So $(F_1/F_2)(\vecu)=(F_1/F_2)(\vecv)$, contradicting the fact that $(F_1/F_2)(\vecu)\neq (F_1/F_2)(\vecv)$.

Therefore, $(F_1/F_2)|_H \in \FF$. 
Denote $(F_1/F_2)|_H$ by $\gamma$. 
The facts that $(F_1/F_2)|_H=\gamma$, $H$ is defined by $f_{0,1}$, and $f_{0,1}$ divides $F_0$ imply that 
\begin{equation}\label{eq:congruence}
F=F_0+F_1+F_2\equiv F_1+F_2\equiv (1+\gamma) F_2=(1+\gamma)\prod_{i=1}^{m_2} f_{2,i}\pmod {f_{0,1}}.
\end{equation}
If $1+\gamma=0$, then by \eqref{eq:congruence}, $F_1+F_2$ would be divisible by $f_{0,1}$, contradicting Assumption~\eqref{item:assumption7}. 
So $1+\gamma\neq 0$.
Therefore, by \eqref{eq:congruence},
\[
V(f_{0,1}, F_1+F_2)=\bigcup_{i=1}^{m_2} V(f_{0,1}, f_{2,i}),
\]
where the degree of each $V(f_{0,1}, f_{2,i})$ in $\aff^n_\FF$ is bounded by $\delta^2$ by \bezout's inequality.

As $f_{0,1}$ does not divide $F_2$, the codimension of each $V(f_{0,1}, f_{2,i})$ is exactly two. 
Apply the same plane-selection argument as in
\cref{lem:justify}. It gives an explicit family $\mathcal H$ containing a plane $P$ such that $\dim(V(f_{0,1})\cap P)=1$ and $\dim(V(f_{0,1},f_{2,i})\cap P)\leq 0$ for all $i\in[m_2]$.
Hence the one-dimensional variety $V(f_{0,1})\cap P$ is not contained in the union of these finite sets.
This implies that $\prod_{i=1}^{m_2} f_{2,i}$ is nonzero when restricted to $V(f_{0,1})\cap P$.
Combining this with \eqref{eq:congruence} and the fact that $1+\gamma\neq 0$, we have $F\equiv (1+\gamma)\prod_{i=1}^{m_2} f_{2,i}\not\equiv 0$ when restricted to $V(f_{0,1})\cap P$. In particular, $F$ is nonzero when restricted to $P$.

The rest of the proof is the same as the last part of the proof of \cref{lem:justify}:  As $F$ restricted to $P$ is a bivariate polynomial of degree $\deg(F)=d_0\leq d$, we know how to construct an explicit hitting set of size $d^{O(1)}$ on $P$ for $F$. While we do not know which affine plane $P\in\mathcal{H}$ works, we could construct an explicit hitting set for each affine plane in $\mathcal{H}$ and then take their union as the final hitting set, whose size is at most $|\mathcal{H}|\cdot d^{O(1)}\leq n^{O(\delta^2)}d^{O(1)}$.
\end{proof}

\paragraph{The harder case.}

Next, we address the harder case, where $\pres(F_1)$ and $\pres(F_2)$ vanish at some point $\mm\in\widetilde{C}$ with orders $k_1,k_2\ge0$, respectively, and $k_1\ne k_2$. The basic idea is to specialize the projective plane (i.e., by assigning $\veca=(a_{0,1},\dots,a_{2,n})\in\FF^{3n}$ to $(Y_{0,1},\dots,Y_{2,n})$) so that the two orders of vanishing remain distinct after specialization. This, in turn, ensures that the polynomial $F$ remains nonzero when restricted to the specialized plane.

To achieve this, we exploit the structure of $F_1$ and $F_2$. Recall that $F_1$ and $F_2$ are products of polynomials $f_{i,j}$ of constant degree. Since $\widetilde{C}$ is normal and the function $\ord_\mm(\cdot)$, which maps a function to its order of vanishing at $\mm$, is additive on products (as it is a discrete valuation), we can ``decompose'' $F_1$ and $F_2$ into their constant-degree factors $f_{i,j}$ in our analysis. This decomposition plays a key role in efficiently derandomizing the choice of $\veca$.

We now provide the details:

\begin{lemma}\label{lem:hard-case}
For $i=0,1,2$, there exists an explicit set $\mathcal{H}_i\subseteq\FF^n$ of size at most $(nd)^{O_\delta(1)}$ independent of $F$ such that, if the restriction of $g=\pres(F_1)/\pres(F_2)$ to the normalization $\widetilde{C\cap U_i}$ of $C\cap U_i$ is not regular, then $\mathcal{H}_i$ is a hitting set for $F$.
\end{lemma}

\begin{proof}
By symmetry, it suffices to explicitly construct $\mathcal{H}_0$ and prove that this set satisfies the lemma.
Let $C_0=C\cap U_0$.
Suppose the restriction of $g$ to $\widetilde{C_0}$ is not regular.
Identifying $U_0$ with the affine plane $\aff^2_\KK$ using the coordinates $Z_1=\widehat{Z}_1/\widehat{Z}_0$ and $Z_2=\widehat{Z}_2/\widehat{Z}_0$, one can see that the affine curve $C_0$ is defined by $\res(f_{0,1})$ in $\aff^2_\KK$, and $g=\pres(F_1)/\pres(F_2)$ restricted to $\aff^2_\KK$ becomes the rational function $\res(F_1)/\res(F_2)$.
So $(\res(F_1)/\res(F_2))|_{\widetilde{C}_0}$ is not regular. 

The coordinate ring of $C_0$ is the integral domain $A:=\KK[Z_1,Z_2]/\ideal{\res(f_{0,1})}$, and the coordinate ring of $\widetilde{C}_0$ is its integral closure $\widetilde{A}$.
Define 
\begin{equation}\label{eq:def-of-g0}
g_0:=\frac{\res(F_1)+\ideal{\res(f_{0,1})}}{\res(F_2)+\ideal{\res(f_{0,1})}}\in \fr{A}.
\end{equation}
As mentioned above, $(\res(F_1)/\res(F_2))|_{\widetilde{C}_0}$ is not a regular function. By definition, this means that it does not belong to the coordinate ring of $\widetilde{C}_0$, that is, $g_0\notin\widetilde{A}$.

\paragraph{Step 1: Representing the normalization.}

We start by representing $\widetilde{A}$, or more precisely, the integral closure $\widetilde{A_{\KK^{(e)}}}$ of $A_{\KK^{(e)}}=A\otimes_\KK\KK^{(e)}$ for some $e\geq0$, in terms of generators and relations. Under this representation, the subring $A_{\KK^{(e)}}$ can also be represented explicitly. This is achieved via \cref{thm:integral-closure}.

By Noether normalization (\cref{lem:noether}), there exist $c_1,c_2\in\FF^\times$ such that for $\alpha=c_1 Z_1+c_2 Z_2$, the natural ring homomorphism $\KK[\alpha]\to A$ sending $\alpha$ to $\alpha+\ideal{\res(f_{0,1})}$ is injective and makes $A$ a finite $\KK[\alpha]$-module.
Moreover, by \cref{lem:irreducibility}, for any algebraic extension $\LL$ of $\KK$, $A_\LL:=\LL[Z_1,Z_2]/\ideal{\res(f_{0,1})}$ is an integral domain of Krull dimension one.
And by \cref{lem:irred-sep}, $\fr{A_{\LL}}$ is a finite separable extension of $\LL(\alpha)$.
By \cref{def:res} and the fact that $\deg(f_{0,1})\leq \delta$, we have $\res(f_{0,1})\in \KK_{\leq \delta}[Z_1,Z_2]_{\leq \delta}$.
(Similarly, $\res(f_{i,j})\in \KK_{\leq \delta}[Z_1,Z_2]_{\leq \delta}$ for all $i=0,1,2$ and $j\in [m_i]$.)
So, by \cref{thm:integral-closure}, 
there exist integers $D,D',m,k',e$ satisfying $D,D',m,k',p^e=O_{\delta}(1)$, and  $g_1,g_2,h_1,\dots,h_{k'}\in\KK^{(e)}[\vecT]=\KK^{(e)}[T_1,\dots,T_m]$ such that the following hold: 
\begin{enumerate}[(1)]
\item $g_1,g_2,h_1,\dots,h_{k'}\in \KK^{(e)}_{\leq D'}[\vecT]_{\leq D}$.
\item The map
\begin{align*}
\phi: A_{\KK^{(e)}}&\to \KK^{(e)}[\vecT]/\ideal{h_1,\dots,h_{k'}}
\\
Z_i+\ideal{\res(f_{0,1})}&\mapsto g_i + \ideal{h_1,\dots,h_{k'}}, \quad i=1,2
\end{align*} 
defines an injective $\KK^{(e)}$-linear ring homomorphism.
\item $\KK^{(e)}[\vecT]/\ideal{h_1,\dots,h_{k'}}$ is isomorphic to the integral closure $\widetilde{A_{\KK^{(e)}}}$ of $A_{\KK^{(e)}}$, and this isomorphism composed with $\phi$ is the natural inclusion of $A_{\KK^{(e)}}$ in its integral closure.
\item $\KK^{(e)}[\vecT]/\ideal{h_1,\dots,h_{k'}}$
is geometrically normal over $\KK^{(e)}$.
\end{enumerate}

We claim that $g_0$ is not in the integral closure $\widetilde{A_{\KK^{(e)}}}$ of $A_{\KK^{(e)}}$ either.
To see this, assume to the contrary that $g_0\in \widetilde{A_{\KK^{(e)}}}$. Then it is integral over $A_{\KK^{(e)}}$. And $A_{\KK^{(e)}}=A\otimes_\KK \KK^{(e)}$ is integral over $A$. So by \cref{lem:integral-transitivity}, $g_0$ is integral over $A$, contradicting the fact that $g_0\not\in\widetilde{A}$. Therefore, $g_0\not\in \widetilde{A_{\KK^{(e)}}}$.

\paragraph{Step 2: Choosing a valuation detecting non-regularity.}

By \cref{lem:intersection-of-localization}, we have
\begin{equation}\label{eq:intersection-of-local}
\widetilde{A_{\KK^{(e)}}}=\bigcap_{\textup{maximal ideal }\mm\subseteq \widetilde{A_{\KK^{(e)}}}} (\widetilde{A_{\KK^{(e)}}})_\mm.\footnote{Alternatively, \eqref{eq:intersection-of-local} follows from \emph{Algebraic Hartogs's Lemma} \cite[13.5.19]{Vak24}, where the intersection is taken over all prime ideals of codimension one. This extends to higher-dimensional normal varieties.}
\end{equation}

By \eqref{eq:intersection-of-local}, for some maximal ideal $\mm$ of $\widetilde{A_{\KK^{(e)}}}$, we have $g_0\not\in (\widetilde{A_{\KK^{(e)}}})_\mm$. Fix such $\mm$. As $\widetilde{A_{\KK^{(e)}}}$ is integrally closed, the one-dimensional local ring $(\widetilde{A_{\KK^{(e)}}})_\mm$ is a discrete valuation ring and is equipped with a normalized valuation $\ord_\mm(\cdot)$. An element $f\in\fr{A_{\KK^{(e)}}}$ is in $(\widetilde{A_{\KK^{(e)}}})_\mm$ if and only if $\ord_\mm(f)\geq 0$.

For $i=1,2$ and $j\in [m_i]$, let 
\begin{equation}\label{eq:valinfo}
k_{i,j}=\ord_\mm(\res(f_{i,j})+\ideal{\res(f_{0,1})}) \quad\text{and}\quad n_i=\ord_\mm(\res(F_i)+\ideal{\res(f_{0,1})}).
\end{equation}
Then for $i=1,2$, since $F_i=\prod_{j=1}^{m_i} f_{i,j}$, we have
\begin{equation}\label{eq:decompose}
n_i=\ord_\mm(\res(F_i)+\ideal{\res(f_{0,1})})=\sum_{j\in [m_i]}\ord_\mm(\res(f_{i,j})+\ideal{\res(f_{0,1})})=\sum_{j\in [m_i]} k_{i,j}.
\end{equation}
By \eqref{eq:def-of-g0} and the fact that $g_0\not\in (\widetilde{A_{\KK^{(e)}}})_\mm$, we have
\begin{equation}\label{eq:inequality}
0\leq n_1<n_2.
\end{equation}

\paragraph{Step 3: Representing the valuation.}

Next, we realize the valuation via polynomials of bounded complexity.

We know that $\widetilde{A_{\KK^{(e)}}}$ may be identified with $\KK^{(e)}[\vecT]/\ideal{h_1,\dots,h_{k'}}$ such that $\phi$ becomes the inclusion $A_{\KK^{(e)}}\hookrightarrow \widetilde{A_{\KK^{(e)}}}$. We now consider the latter ring $\KK^{(e)}[\vecT]/\ideal{h_1,\dots,h_{k'}}$ for computational purposes. 

View $\mm$ as a maximal ideal of $\KK^{(e)}[\vecT]/\ideal{h_1,\dots,h_{k'}}$.
Let $\widehat{\mm}\subseteq\KK^{(e)}[\vecT]$ be the preimage of $\mm$ under the natural quotient map, which is a maximal ideal of $\KK^{(e)}[\vecT]$.

Next, we find a \grobner basis of $\widehat{\mm}$ whose members have low complexity. 
For $i\in\{0,1,2\}$ and $j\in [m_i]$, since  $\res(f_{i,j})\in \KK_{\leq \delta}[Z_1,Z_2]_{\leq \delta}$ and $g_1,g_2\in \KK^{(e)}_{\leq D'}[\vecT]_{\leq D}$, where $D,D',p^e=O_{\delta}(1)$, we have $\res(f_{i,j})(g_1,g_2)\in \KK^{(e)}_{\leq d_2}[\vecT]_{\leq d_1}$ for some $d_1,d_2=O_{\delta}(1)$. 
%By definition, we have $\phi(\res(f_{i,j}))=\res(f_{i,j})(g_1,g_2)+\ideal{h_1,\dots,h_{k'}}$.
By \eqref{eq:decompose} and \eqref{eq:inequality}, there exists $j_0\in [m_2]$ such that $k_{2,j_0}>0$.
So $\res(f_{2,j_0})(g_1,g_2)\in \widehat{\mm}+\ideal{h_1,\dots,h_{k'}}$.
Therefore, $\widehat{\mm}$ contains the zero-dimensional ideal $\ideal{h_1,\dots,h_{k'},\res(f_{2,j_0})(g_1,g_2)}$.

Let $\LL=\KK^{(e)}$ and
$C=\LL[\vecT]/\ideal{h_1,\dots,h_{k'}}$.
By \cref{lem:bezout2}, $[C/\mm:\LL]=O_\delta(1)$.
In characteristic $p>0$, let $r$ be the largest integer
such that $p^r$ divides this degree, and let
$\LL'=\KK^{(e+r)}$; in characteristic zero, let $r=0$ and $\LL'=\LL$.
Let $\mm'$ be the unique maximal ideal of
$C_{\LL'}:=C\otimes_\LL\LL'$ above $\mm$.
By \cref{lem:pth-root}, its residue field is separable
over $\LL'$.

By the geometric normality obtained in Step~1,
$C_{\LL'}$ remains normal. It is integral over
$A_{\LL'}$ and has the same fraction field, so it is
the normalization of $A_{\LL'}$.
The inclusion of discrete valuation rings
$C_\mm\subseteq(C_{\LL'})_{\mm'}$ implies that the
new normalized valuation restricts to a positive integer
multiple of the old one: old units remain units, and an
old uniformizer has positive valuation in the new ring.
Replace $e,\mm,\widehat{\mm}$ by their counterparts
over $\LL'$, and redefine $k_{i,j}$ and $n_i$ using this
new valuation. Then \eqref{eq:decompose} and
\eqref{eq:inequality} still hold.
Since $p^r=O_\delta(1)$, all complexity bounds remain
$O_\delta(1)$ after this scalar extension. 
By \cref{lem:extract}, $\widehat{\mm}$ admits a \grobner basis $G\subseteq \KK^{(e)}_{\leq d_4}[\vecT]_{\leq d_3}$ for some $d_3,d_4\in O_{\delta}(1)$.

As $\KK^{(e)}[\vecT]/\ideal{h_1,\dots,h_{k'}}$ is isomorphic to (and hence identified with) $\widetilde{A_{\KK^{(e)}}}$, the local ring
\[
B:=(\KK^{(e)}[\vecT]/\ideal{h_1,\dots,h_{k'}})_{\mm}
\]
is a discrete valuation ring, naturally identified with $(\widetilde{A_{\KK^{(e)}}})_\mm$.
By a slight abuse of notation, we continue to denote its normalized valuation by $\ord_\mm$.
As $\ord_\mm$ is normalized and $G$ generates the ideal $\widehat{\mm}$, there exists $u\in G$ such that $\overline{u}:=u+\ideal{h_1,\dots,h_{k'}}$ satisfies $\ord_\mm(\overline{u})=1$.

\paragraph{Step 4: Deriving polynomial identities.} We now encode the valuation information in \eqref{eq:valinfo} as a collection of low-complexity polynomial identities.

Fix $i\in\{1,2\}$ and $j\in[m_i]$, and write
$I=\ideal{h_1,\dots,h_{k'}}$,
$C=\KK^{(e)}[\vecT]/I$, and
$P=\res(f_{i,j})(g_1,g_2)$.
Let $k=k_{i,j}$.
In the discrete valuation ring $B=C_\mm$, we have
$P+I=\overline{u}^{k}v$ for some unit $v\in B^\times$.
Write $v=(\rho+I)/(\sigma+I)$ with
$\rho,\sigma\in\KK^{(e)}[\vecT]\setminus\widehat{\mm}$.
Since $C\hookrightarrow B$ is injective,
$\sigma P-u^k\rho\in I$.
Choose $\tau\in\KK^{(e)}[\vecT]$ such that
$\tau\rho\equiv1\pmod{\widehat{\mm}}$.
Taking $t_{i,j}=\tau\sigma$ gives
\begin{equation}\label{eq:tij-relation}
 t_{i,j}P-u^k\in u^k\widehat{\mm}+I.
\end{equation}

Moreover, $k=O_\delta(1)$: the local quotient
$B/(P+I)$ has length $k$, so
$k\leq\dim_{\KK^{(e)}}(C/(P+I))=O_\delta(1)$
by \cref{lem:bezout2}.
Since \eqref{eq:tij-relation} implies
$u^k\in\ideal{P,\{u^k g:g\in G\},h_1,\dots,h_{k'}}$,
\cref{lem:ideal-membership-coeff} lets us choose
$t_{i,j}$ and witness polynomials of bounded complexity
such that
\begin{equation}\label{eq:relation1}
 t_{i,j}\res(f_{i,j})(g_1,g_2)
 =u^{k_{i,j}}\left(1+\sum_{g\in G}a_{i,j,g}g\right)
  +\sum_{\ell=1}^{k'}b_{i,j,\ell}h_\ell,
\end{equation}
where
$t_{i,j}\in\KK^{(e)}_{\leq d_6}[\vecT]_{\leq d_5}$,
$a_{i,j,g},b_{i,j,\ell}\in
\KK^{(e)}_{\leq d_8}[\vecT]_{\leq d_7}$, and
$d_5,d_6,d_7,d_8=O_\delta(1)$.

Cancelling $\overline{u}^{k}$ in $B$ and reducing
\eqref{eq:relation1} modulo $\mm B$ gives
$t_{i,j}v\equiv1\pmod{\mm B}$.
Thus this bounded-complexity choice of $t_{i,j}$ is
still invertible modulo $\widehat{\mm}$.
Applying \cref{lem:ideal-membership-coeff} to
$1\in\ideal{t_{i,j}}+\widehat{\mm}$ gives
\begin{equation}\label{eq:relation2}
 s_{i,j}t_{i,j}=1+\sum_{g\in G}c_{i,j,g}g,
\end{equation}
with
$s_{i,j},c_{i,j,g}\in
\KK^{(e)}_{\leq d_{10}}[\vecT]_{\leq d_9}$ and
$d_9,d_{10}=O_\delta(1)$.

Similarly, as $u\in\widehat{\mm}$ and $\ideal{h_1,\dots,h_{k'}}\subseteq\widehat{\mm}$, by \cref{lem:ideal-membership-coeff}, we can write
\begin{equation}\label{eq:relation3}
u=\sum_{g\in G} d_{g} g \quad\text{and}\quad h_\ell=\sum_{g\in G} e_{\ell,g} g \text{~~for~~} \ell\in[k']
\end{equation}
with $d_g, e_{\ell,g}\in \KK^{(e)}_{\leq d_{12}}[\vecT]_{\leq d_{11}}$ for all $g\in G$ and $\ell\in [k']$, where $d_{11},d_{12}\in O_{\delta}(1)$.

\paragraph{Step 5: Constructing specialization conditions.}

We adopt the following notation:
Suppose $Q$ is a polynomial over $\KK^{(e)}$, and $\veca\in\FF^{3n}$ is a point such that none of the denominators of the coefficients of $Q$, which are polynomials in $\vecY^{1/p^e}$, vanishes after assigning $\veca$ to $\vecY^{1/p^e}$. Then the resulting polynomial after the assignment $\vecY^{1/p^e}\gets \veca$ is a well-defined polynomial over $\FF$. Denote this polynomial by $Q|_{\veca}$, called the specialization of $Q$ at $\veca$.

For convenience, for a set $S\subseteq \FF[\vecY^{1/p^e}]$ and $\veca\in\FF^{3n}$, we say $\veca$ is a \emph{common non-root} of $S$ if every $Q\in S$ is non-vanishing at $\veca$.

We will construct a set $S^* \subseteq \KK^{(e)}_{\leq D_0}\cap \FF[\vecY^{1/p^e}]$ of size $O_{\delta}(d)$, where $D_0=O_\delta(1)$, such that $\res(F)|_\veca\neq 0$ for any common non-root $\veca$ of $S^*$.

First, we construct $S_0$ by collecting the denominators of all
coefficients of $g_1,g_2,u,h_1,\dots,h_{k'}$, all elements of $G$,
and all polynomials occurring in \eqref{eq:relation1},
\eqref{eq:relation2}, and \eqref{eq:relation3}, including the
witness polynomials $t_{i,j},\ s_{i,j},\ a_{i,j,g},\ b_{i,j,\ell},\
c_{i,j,g},\ d_g,\ e_{\ell,g}$
for all relevant indices.
Then $|S_0|=O_\delta(d)$, and
$S_0\subseteq
\KK^{(e)}_{\leq D_0}\cap\FF[\vecY^{1/p^e}]$
for sufficiently large $D_0=O_\delta(1)$.

We also need the following three claims, whose proofs are given in the appendix.

\begin{restatable}{claim}{proper}\label{claim:proper}
There exists a set $S_1\subseteq \KK^{(e)}_{\leq D_1}\cap \FF[\vecY^{1/p^e}]$ of size $O_{\delta}(1)$, where $D_1=O_\delta(1)$, such that for any common non-root $\veca\in\FF^{3n}$ of $S_0\cup S_1$, 
it holds that $1\not\in \widehat{\mm}|_\veca$. 
\end{restatable}

\begin{restatable}{claim}{nonzerodivisor}\label{claim:non-zerodivisor}
There exists a set $S_2\subseteq \KK^{(e)}_{\leq D_2}\cap \FF[\vecY^{1/p^e}]$ of size $O_{\delta}(1)$, where $D_2=O_\delta(1)$, such that for any common non-root $\veca\in\FF^{3n}$ of $S_0\cup S_2$, the element $u|_{\veca}+\ideal{h_1|_\veca,\dots,h_{k'}|_\veca}$ is a non-zerodivisor of $\FF[\vecT]/\ideal{h_1|_\veca,\dots,h_{k'}|_\veca}$.
\end{restatable}

\begin{restatable}{claim}{injectivity}\label{claim:injectivity}
There exists a set $S_3\subseteq \KK^{(e)}_{\leq D_3}\cap \FF[\vecY^{1/p^e}]$ of size $O_{\delta}(d)$, where $D_3=O_\delta(1)$, such that for any common non-root $\veca\in\FF^{3n}$ of $S_0\cup S_3$, the polynomial $\res(f_{i,j})|_\veca\in\FF[Z_1,Z_2]$ is well-defined for all $i\in\{0,1,2\}$ and $j\in [m_i]$, and the kernel of the $\FF$-linear ring homomorphism $\FF[Z_1,Z_2]\to \FF[\vecT]/\ideal{h_1|_\veca,\dots,h_{k'}|_\veca}$ sending $Z_i$ to $g_i|_\veca+\ideal{h_1|_\veca,\dots,h_{k'}|_\veca}$ contains $\ideal{\res(f_{0,1})|_\veca}$. 
\end{restatable}

Let $S^*=S_0\cup S_1\cup S_2\cup S_3$. Then $|S^*|=O_\delta(d)$, and $S^* \subseteq \KK^{(e)}_{\leq D_0}\cap \FF[\vecY^{1/p^e}]$ for sufficiently large $D_0=O_\delta(1)$.

\paragraph{Step 6: Specialization.}

Consider a common non-root $\veca$ of $S^*$. We will show that
$\res(F)|_\veca\neq 0$.
The remainder of this step is devoted to proving this.

Define the ideal 
\[
\widehat{\mm}|_\veca:=\ideal{\{g|_\veca: g\in G\}}
\]
of $\FF[\vecT]$.
By \eqref{eq:relation1}, \eqref{eq:relation2}, and \eqref{eq:relation3}, and the fact that $S_0\subseteq S^*$, we have
\begin{equation}\label{eq:reduced1}
\left(t_{i,j}(\res(f_{i,j})(g_1,g_2))\right)|_\veca= \left.\left(u^{k_{i,j}} \left(1+\sum_{g\in G} a_{i,j,g} g\right)+\sum_{\ell=1}^{k'} b_{i,j,\ell} h_\ell\right)\right\vert_\veca,
\end{equation}
\begin{equation}\label{eq:reduced2}
(s_{i,j} t_{i,j})|_\veca=\left.\left(1+\sum_{g\in G} c_{i,j,g} g\right)\right\vert_\veca,
\end{equation}
and
\begin{equation}\label{eq:reduced3}
u|_\veca\in \widehat{\mm}|_\veca \quad\text{and}\quad \ideal{h_1|_\veca,\dots,h_{k'}|_\veca}\subseteq \widehat{\mm}|_\veca.
\end{equation}

Taking the product of the equations \eqref{eq:reduced1} over all $j\in[m_i]$ and applying \eqref{eq:decompose}, we obtain, for $i=1,2$,
\begin{equation}\label{eq:product}
\left.\left(\left(\prod_{j\in [m_i]} t_{i,j}\right) \res(F_i)(g_1,g_2)\right)\right\vert_\veca - u^{n_i}|_\veca
\in u^{n_i}|_\veca \cdot \widehat{\mm}|_\veca + \ideal{h_1|_\veca,\dots,h_{k'}|_\veca}.
\end{equation}

Let 
\[
\mm|_\veca:=\widehat{\mm}|_\veca/\ideal{h_1|_\veca,\dots,h_{k'}|_\veca},
\]
which is an ideal of $R_\veca:=\FF[\vecT]/\ideal{h_1|_\veca,\dots,h_{k'}|_\veca}$.
By \cref{claim:proper}, we have $1\not\in \mm|_\veca$, i.e., $\mm|_\veca\subsetneq R_\veca$.

Let $\overline{u}|_\veca:=u|_\veca+\ideal{h_1|_\veca,\dots,h_{k'}|_\veca}\in R_\veca$.
By \cref{claim:non-zerodivisor}, $\overline{u}|_\veca$ is a non-zerodivisor of $R_\veca$, and so are its powers.
Therefore, for any integer $r\geq 0$, the $(R_\veca/\mm|_\veca)$-module homomorphism
\begin{align*}
\times (\overline{u}|_\veca)^r: R_\veca/\mm|_\veca &\to \ideal{(\overline{u}|_\veca)^r}/((\overline{u}|_\veca)^r\mm|_\veca)
\\
x+\mm|_\veca &\mapsto (\overline{u}|_\veca)^r x+(\overline{u}|_\veca)^r \mm|_\veca.
\end{align*} 
is injective. It is also obviously surjective. So $\times (\overline{u}|_\veca)^r$ is a $(R_\veca/\mm|_\veca)$-module isomorphism.

As $1\not\in \mm|_\veca$, using the above isomorphism, we see that $(\overline{u}|_\veca)^{r}$ is in $\ideal{(\overline{u}|_\veca)^r}$ but not in $(\overline{u}|_\veca)^r\cdot \mm|_\veca$ for $r\geq 0$.

Let $U\subseteq R_\veca$ be the set of all finite products of elements in $\{t_{2,j}|_\veca+\ideal{h_1|_\veca,\dots,h_{k'}|_\veca}: j\in [m_2]\}$. By \eqref{eq:reduced2},  elements in $U$ are already invertible modulo $\mm|_\veca$. So localizing $R_\veca/\mm|_\veca$ and $\ideal{(\overline{u}|_\veca)^r}/((\overline{u}|_\veca)^r\mm|_\veca)$ with respect to the multiplicatively closed set $U$ does not change these two modules. It follows that $(\overline{u}|_\veca)^{r}$ is not in the localization $U^{-1}((\overline{u}|_\veca)^r\mm|_\veca)$ for $r\geq 0$.

In particular, by \eqref{eq:product}, we have 
\[
\left.\left(\left(\prod_{j\in [m_1]} t_{1,j}\right) \res(F_1)(g_1,g_2)\right)\right\vert_\veca+\ideal{h_1|_\veca,\dots,h_{k'}|_\veca}\not\in U^{-1}((\overline{u}|_\veca)^{n_1}\cdot \mm|_\veca).
\]
It follows that 
\begin{equation}\label{eq:F1}
\res(F_1)(g_1,g_2)|_\veca+\ideal{h_1|_\veca,\dots,h_{k'}|_\veca}\not\in U^{-1}((\overline{u}|_\veca)^{n_1}\cdot \mm|_\veca).
\end{equation}

As $n_2>n_1$ and $\overline{u}|_\veca\in \mm|_\veca$ (which holds by \eqref{eq:reduced3}), it also follows from \eqref{eq:product} that
\[
\left.\left(\left(\prod_{j\in [m_2]} t_{2,j}\right) \res(F_2)(g_1,g_2)\right)\right\vert_\veca+\ideal{h_1|_\veca,\dots,h_{k'}|_\veca}\in\ideal{(\overline{u}|_\veca)^{n_2}}\subseteq (\overline{u}|_\veca)^{n_1}\cdot \mm|_\veca.
\]
Localizing with respect to $U$ makes $\left(\prod_{j\in [m_2]} t_{2,j}\right)|_\veca+\ideal{h_1|_\veca,\dots,h_{k'}|_\veca}$ invertible.
Therefore,
\begin{equation}\label{eq:F2}
\res(F_2)(g_1,g_2)|_\veca+\ideal{h_1|_\veca,\dots,h_{k'}|_\veca}\in U^{-1}((\overline{u}|_\veca)^{n_1}\cdot \mm|_\veca).
\end{equation}
It follows from \eqref{eq:F1} and \eqref{eq:F2} %\SWnote{and the fact that $\overline{u}|_\veca$ is not a zero divisor} 
that 
\begin{equation}\label{eq:not-in-ideal}
\res(F_1+F_2)(g_1,g_2)|_{\veca}\not\in \ideal{h_1|_\veca,\dots,h_{k'}|_\veca}.
\end{equation}
 
By \cref{claim:injectivity}, $\res(F_0)|_\veca$, $\res(F_1)|_\veca$, and $\res(F_2)|_\veca$ are well-defined.
Again by \cref{claim:injectivity}, we have $\res(f_{0,1})(g_1,g_2)|_\veca\in \ideal{h_1|_\veca,\dots,h_{k'}|_\veca}$. 
This together with \eqref{eq:not-in-ideal} implies that 
\[
\res(F_1+F_2)|_\veca\not\in  \ideal{\res(f_{0,1})|_\veca}.
\]
Therefore,
\[
\res(F)|_\veca=\res(F_0+F_1+F_2)|_\veca\equiv \res(F_1+F_2)|_\veca\not\equiv 0\pmod{\ideal{\res(f_{0,1})|_\veca}}.
\]
In particular, $\res(F)|_\veca\neq 0$. 

\paragraph{Step 7: Concluding the hitting set construction.}

Write $q=p^e$, with $q=1$ in characteristic zero.
Since specialization assigns $\veca$ to $\vecY^{1/q}$, it assigns $\veca^{\,q}$ to $\vecY$, where powers are taken
coordinatewise. Thus the preceding step shows that, for any common non-root $\veca$ of $S^*$, the polynomial $F\circ L_{\veca,q}$ is nonzero, where $L_{\veca,q}(z_1,z_2)
:=\left(a_{0,j}^{q}+a_{1,j}^{q}z_1+a_{2,j}^{q}z_2\right)_{j\in [n]}$.

Choose constants $M_\delta,C_\delta,D_0=O_\delta(1)$,
independent of $F$, such that $q\leq M_\delta$,
$|S^*|\leq C_\delta d$, and every element of $S^*$ has
degree at most $D_0$ in $\vecY^{1/q}$.
Using \cite{KS01}, construct an $\epsilon$-hitting set
$\mathcal H\subseteq\FF^{3n}$ for all nonzero
$3n$-variate polynomials of degree at most $D_0$, where
$\epsilon<1/(C_\delta d)$.
Its size is $d^{O(1)}n^{O_\delta(1)}$.
By the union bound, $\mathcal H$ contains a common
non-root of $S^*$.

Let $\mathcal Q$ consist of all powers $p^r\leq M_\delta$
with $r\geq0$; in characteristic zero, let $\mathcal Q=\{1\}$.
Fix a set $T\subseteq\FF$ of $d+1$ distinct elements and define
\[
 \mathcal H_0
 =\bigcup_{q\in\mathcal Q}\;
   \bigcup_{\veca\in\mathcal H} L_{\veca,q}(T^2).
\]
For the value of $q$ used in the preceding argument, some
$\veca\in\mathcal H$ makes $F\circ L_{\veca,q}$ a nonzero
bivariate polynomial of degree at most $d$, so $T\times T$ hits it.
Therefore $\mathcal H_0$ hits $F$.
The construction is independent of $F$, and
\[
 |\mathcal H_0|
 \leq |\mathcal Q|\,|\mathcal H|\,(d+1)^2
 =d^{O(1)}n^{O_\delta(1)}
 \leq (nd)^{O_\delta(1)}
\]
as desired.
\end{proof}

Now we finish the proof of \cref{thm:main}.

\begin{proof}[Proof of \cref{thm:main}]
As discussed at the beginning of \cref{sec:proof-main}, we may assume  Assumptions \eqref{item:assumption1}--\eqref{item:assumption7}.
Under these assumptions, the statement follows from \cref{lem:easy-case} and \cref{lem:hard-case}. The final hitting set is the union of the hitting set in \cref{lem:easy-case} and the three hitting sets $\mathcal{H}_0,\mathcal{H}_1,\mathcal{H}_2$ in \cref{lem:hard-case}. The resulting size remains $(nd)^{O_\delta(1)}$.
\end{proof}

% \section{Conclusions}\label{sec:conclusion}

% In this paper, we gave a deterministic polynomial-time black-box PIT algorithm for $\Sigma^{[3]}\Pi\Sigma\Pi^{[\delta]}$ circuits over arbitrary fields, under a squarefreeness assumption. Beyond the specific result, our main contribution is a new approach to PIT based on normalization of varieties, providing a field-independent alternative to the Sylvester--Gallai framework.

% A central open problem is to remove the squarefreeness assumption and obtain PIT algorithms for the full class of $\Sigma^{[3]}\Pi\Sigma\Pi^{[\delta]}$ circuits over arbitrary fields. Our current proof uses normalization to recover factorization-like behavior after restricting to a hypersurface. In the general case, however, relevant multiplicity information may be lost upon passing to the reduced structure. Understanding how to retain or recover such higher-order information may be an important step toward further progress.

% Another direction is to extend the normalization-based framework to larger top fan-in. It would also be interesting to investigate whether the techniques developed here can be combined with other structural ideas, including those arising from the Sylvester--Gallai-based literature.

\paragraph{AI disclosure.}
During revisions, we used generative AI to assist with proof checking and corrections, including
the geometric-normality and specialization arguments. The authors take responsibility for all content.

\bibliographystyle{alphaurl}
\bibliography{ref}

\appendix

\section{Proof of \texorpdfstring{\cref{lem:factor}}{Lemma 3.6}}\label{sec:appendix}

We now prove \cref{lem:factor}, which we restate below for convenience.

\factor*

First, we introduce some notation. Let $R=\FF[\vecY]\subseteq \KK=\FF(\vecY)$.
For an irreducible polynomial $a\in R$ and $b\in \KK^\times$, define $\ord_a(b)$ to be the unique integer $r$ such that $b$ can be written as $b=a^r \frac{s}{t}$, where $s\in R$ and $t\in R\setminus\{0\}$ are not divisible by $a$.
For a univariate polynomial $f\in \KK[X]\setminus\{0\}$ and an irreducible polynomial $a\in R$ in $\vecY$, define $\ord_a(f):=\min_c(\ord_a(c))$, where $c$ ranges over the nonzero coefficients of $f$.
Finally, denote by $\cont{f}$ the \emph{content} of $f$, defined as
\[
\cont{f}:=\prod_a a^{\ord_a(f)}
\]
with the product taken over the irreducible polynomials $a\in R$ such that $\ord_a(f)\neq 0$. If two such polynomials differ by a unit factor, only one representative is included in the product.
Note that $\cont{f}$ is well-defined up to multiplication by a unit of $R$. Also note that by definition, $f\in R[X]$ as long as $\cont{f}\in R$.

We need the following version of Gauss's lemma, which holds more generally when $R$ is a UFD \cite[Theorem~2.1]{Lan12}.

\begin{lemma}[Gauss's lemma]\label{lem:gauss}
$\cont{gh}=\cont{g}\cont{h}$ for nonzero univariate polynomials $g,h\in \KK[X]$.
\end{lemma}

\begin{proof}[Proof of \cref{lem:factor}] 
Consider the Kronecker map $\phi:\KK[X_1,\dots,X_n]\to \KK[X]$, which is the $\KK$-linear ring homomorphism sending $X_i$ to $X^{(d+1)^{i-1}}$ for $i$. We have $\phi(f)=\phi(g)\cdot \phi(f/g)$. The terms of $f$ (resp. $g$) correspond one-to-one to the terms of $\phi(f)$ (resp. $\phi(g)$), preserving coefficients.
And the degree of $\phi(f)$ is at most $(d+1)^n=O_{d,n}(1)$.
This reduces the problem to the univariate case.

Now we prove the lemma for the univariate case. Let $f\in \KK_{\leq d'}[X]_{\leq d}$ be a nonzero polynomial.
Recall that $R=\FF[\vecY]$.
Let $s$ be the least common multiple of all the denominators of the nonzero coefficients of $f$, so that $sf\in R[X]$.
Let $t=\cont{sf}\in R\setminus\{0\}$. Then $\cont{\frac{s}{t}f}=1$. Note that the coefficients of $\frac{s}{t}f$ are polynomials in $\FF[\vecY]$ of degree at most $d'':=(d+1)d'$.
Let $c=\frac{1}{\cont{g}}\in\KK^\times$. Then $\cont{c g}=1$. 
Let $h=\frac{s}{ct}\cdot \frac{f}{g}$, so that $\frac{s}{t}f=(cg)h$. As $\cont{\frac{s}{t}f}=1$ and $\cont{cg}=1$, we have $\cont{h}=1$ by \cref{lem:gauss}. In particular, $\frac{s}{t}f$, $cg$, and $h$ are all in $R[X]$.

Write $cg=\sum_{i=0}^{\deg(cg)} a_i X^i$ and $h=\sum_{j=0}^{\deg(h)} b_j X^j$ with $a_i,b_j\in R$. Let $D=\max_{0\leq i\leq \deg(cg)} \deg_\vecY(a_i)$ and $D'=\max_{0\leq j\leq \deg(h)} \deg_\vecY(b_j)$.
Choose the maximal $i_0\in \{0,1,\dots,\deg(cg)\}$ such that $a_{i_0}$ has a nonzero degree-$D$ homogeneous component, and let $H$ be this homogeneous component.
Similarly, choose the maximal $j_0\in \{0,1,\dots,\deg(h)\}$ such that $b_{j_0}$ has a nonzero degree-$D'$ homogeneous component, and let $H'$ be this homogeneous component.
By the maximality of $D$, $D'$, $i_0$, and $j_0$, the degree-$(D+D')$ homogeneous component of the coefficient of $X^{i_0+j_0}$ in $(cg)\cdot h=\frac{s}{t}f$ is exactly $H\cdot H'\neq 0$. However, we know that the coefficients of $\frac{s}{t}f$ all have degrees at most $d''$. So 
\[
D+D'\leq d''.
\]
Therefore, the coefficients of $cg$ are all polynomials in $\FF[\vecY]$ of degree at most $D\leq d''$. In particular, $cg\in \KK_{\leq d''}[X]_{\leq d}$, where $d''=(d+1)d'\in O_{d,d'}(1)$. 
\end{proof}

\section{Proofs of Other Claims}\label{sec:appendix2}

We now prove \cref{claim:proper}, \cref{claim:non-zerodivisor}, and \cref{claim:injectivity}.

\proper*

\begin{proof}
Let $D=O_\delta(1)$ bound the degrees of the polynomials in $G$.
By \cref{lem:ideal-membership}, we can choose
$d_0=O_\delta(1)$ such that, for every common non-root
$\veca$ of $S_0$, if $1\in\widehat{\mm}|_\veca$, then $1=\sum_{g\in G} h_g(g|_\veca)$ holds
for some $h_g\in\FF[\vecT]$ of degree at most $d_0$.
This choice of $d_0$ works for all $\veca$, since specialization
does not increase the degrees or the number of generators.

Construct a matrix $M$ over $\KK^{(e)}$ as follows.
Its rows are indexed by pairs $(\mu,g)$, where $g\in G$
and $\mu$ is a monomial in $\vecT$ of degree at most $d_0$.
Its columns are indexed by all monomials in $\vecT$
of degree at most $d_0+D$.
The row indexed by $(\mu,g)$ is the coefficient vector
of $\mu g$.
Let $M^+$ be obtained from $M$ by adjoining the
coefficient vector of the constant polynomial $1$.

For any common non-root $\veca$ of $S_0$, define
$M|_\veca$ and $M^+|_\veca$ by specializing their entries,
with the same row and column indices.
Thus the rows of $M|_\veca$ are precisely the coefficient
vectors of $\mu(g|_\veca)$, including any zero or repeated
rows. By the choice of $d_0$, we have
\[
  1\in\widehat{\mm}|_\veca
  \quad\Longleftrightarrow\quad
  \mathrm{rank}(M^+|_\veca)
  =\mathrm{rank}(M|_\veca).
\]

Let $r=\mathrm{rank}(M)$.
Since $1\notin\widehat{\mm}=\ideal{G}$, the adjoined
row does not belong to the row span of $M$.
Hence $\mathrm{rank}(M^+)=r+1$.
Choose an $(r+1)\times(r+1)$ submatrix $N$ of $M^+$
with nonzero determinant, and write
$\det(N)=\frac{Q}{P}$, where $P,Q\in\FF[\vecY^{1/p^e}]\setminus\{0\}$.
The dimensions of $M^+$ and the complexities of its
entries are $O_\delta(1)$, so we may choose $P$ and $Q$
of degree $O_\delta(1)$.
Set $S_1=\{P,Q\}$.

Now let $\veca$ be a common non-root of $S_0\cup S_1$.
Then
\[
  \det(N|_\veca)=\frac{Q(\veca)}{P(\veca)}\neq0,
\]
and therefore $\operatorname{rank}(M^+|_\veca)\ge r+1$.
On the other hand, all $(r+1)\times(r+1)$ minors of $M$
vanish identically and remain zero after specialization.
Consequently, $\operatorname{rank}(M|_\veca)\le r$.
Thus the two specialized matrices have different ranks,
which implies $1\notin\widehat{\mm}|_\veca$.
The required size and degree bounds for $S_1$ follow
from its construction.
\end{proof}

\nonzerodivisor*

\begin{proof}
The key observation is that it suffices to consider zerodivisor witnesses of bounded degree.
Indeed,
$u|_{\veca}+\ideal{h_1|_\veca,\dots,h_{k'}|_\veca}$ is a zerodivisor of
$\FF[\vecT]/\ideal{h_1|_\veca,\dots,h_{k'}|_\veca}$
if and only if there exists $v\in\FF[\vecT]$ such that
\[
u|_\veca\cdot v\in\ideal{h_1|_\veca,\dots,h_{k'}|_\veca}
\quad\text{but}\quad
v\notin\ideal{h_1|_\veca,\dots,h_{k'}|_\veca}.
\]
The set of $v$ satisfying the first condition is the ideal quotient
$(\ideal{h_1|_\veca,\dots,h_{k'}|_\veca}:\ideal{u|_\veca})$,
whose \grobner basis can be computed using
\cite[Lemmas~2.3.10 and~2.3.11]{AL94}.
By the degree bounds for the \grobner bases involved, if such a $v$ exists, then one exists with $\deg_{\vecT}(v)\le d_0$
for some $d_0=O_\delta(1)$.

Similarly,
$u+\ideal{h_1,\dots,h_{k'}}$ is a zerodivisor of
$\KK^{(e)}[\vecT]/\ideal{h_1,\dots,h_{k'}}$
if and only if there exists $v\in\KK^{(e)}[\vecT]$ such that
\[
u\cdot v\in\ideal{h_1,\dots,h_{k'}}
\quad\text{but}\quad
v\notin\ideal{h_1,\dots,h_{k'}},
\]
and if such a $v$ exists, then one exists with $\deg_{\vecT}(v)\le d_0$.

Let $W$ be the $\KK^{(e)}$-vector space of polynomials in $\KK^{(e)}[\vecT]$ of degree at most $d_0$. Define
\[
W_1:=\{v\in W:u\cdot v\in\ideal{h_1,\dots,h_{k'}}\},
\qquad
W_2:=\{v\in W:v\in\ideal{h_1,\dots,h_{k'}}\}.
\]
Similarly, for a specialization $\veca$, let $W|_\veca$ be the $\FF$-vector space of polynomials in $\FF[\vecT]$ of degree at most $d_0$, and define
\[
W_1|_\veca
:=
\{v\in W|_\veca:u|_\veca\cdot v\in\ideal{h_1|_\veca,\dots,h_{k'}|_\veca}\},\quad
W_2|_\veca
:=
\{v\in W|_\veca:v\in\ideal{h_1|_\veca,\dots,h_{k'}|_\veca}\}.
\]

Clearly $W_2\subseteq W_1$. Since $\KK^{(e)}[\vecT]/\ideal{h_1,\dots,h_{k'}}$
is an integral domain and
$u+\ideal{h_1,\dots,h_{k'}}\ne0$,
the element $u+\ideal{h_1,\dots,h_{k'}}$ is a non-zerodivisor. Hence $W_1\subseteq W_2$, and therefore
\begin{equation}\label{eq:W1-W2}
W_1=W_2.
\end{equation}

We now express $W_1$ and $W_2$ by finite systems of linear equations. Choose $d_1=O_\delta(1)$, using
\cref{lem:ideal-membership}, uniformly for the identities $uv=\sum_i w_i h_i$ and $v=\sum_i w_i h_i$, both over
$\KK^{(e)}$ and after every specialization at a common non-root of $S_0$. This is possible because the numbers of variables and generators and all relevant degrees are bounded in terms of $\delta$.
So for $v\in W_1$, there exist
$w_1,\dots,w_{k'}\in\KK^{(e)}[\vecT]_{\leq d_1}$ with
\[
 uv=\sum_{i=1}^{k'}w_i h_i.
\]
Viewing the coefficients of $v,w_1,\dots,w_{k'}$ as
unknowns, this identity becomes a homogeneous linear system over $\KK^{(e)}$. Let $M_1$ be its coefficient matrix. Computing a basis of the solution space of this system, projecting each solution to the coordinates corresponding to $v$, and selecting a linearly independent subset if necessary yields a basis of $W_1$. Let $B_1$ be the matrix whose rows are these basis vectors.

Similarly, viewing the coefficients of $v,w_1,\dots,w_{k'}$ as unknowns, the identity
\[
v=\sum_{i=1}^{k'} w_i h_i
\]
gives another homogeneous linear system over $\KK^{(e)}$. Let $M_2$ denote its coefficient matrix.
Computing a basis of the solution space of this system, projecting each solution to the coordinates corresponding to $v$, and selecting a linearly independent subset if necessary yields a basis of $W_2$.
Let $B_2$ be the matrix whose rows are these basis vectors.

By \eqref{eq:W1-W2}, the row spans of $B_1$ and $B_2$ are equal.
We now define $S_2$. Add to $S_2$ all denominators of the entries of
$M_1,M_2,B_1,B_2$ and all intermediate matrices used to compute and project the kernel bases. Also add the numerators and denominators of the nonzero pivot minors used in these Gaussian eliminations and in selecting the bases $B_1,B_2$. Finally, add the denominators of the coefficients expressing each row of $B_1$ as a linear combination of the rows of $B_2$, and vice versa.
Since all the linear systems above have size $O_\delta(1)$ and all their entries have complexity $O_\delta(1)$, we have $S_2\subseteq \KK^{(e)}_{\le D_2}\cap \FF[\vecY^{1/p^e}]$, where
$|S_2|=O_\delta(1)$ and
$D_2=O_\delta(1)$.

Now let $\veca$ be a common non-root of $S_0\cup S_2$.
The conditions imposed by $S_2$ ensure that all the matrices above specialize well and that the same Gaussian eliminations remain valid after specialization. Therefore
$B_1|_\veca$ and $B_2|_\veca$ are bases of $W_1|_\veca$ and $W_2|_\veca$, respectively.
Moreover, the linear relations expressing the rows of $B_1$ in terms of the rows of $B_2$, and conversely, also specialize well. Hence the row spans of $B_1|_\veca$ and $B_2|_\veca$ are equal, i.e., $W_1|_\veca=W_2|_\veca$.

Suppose, for contradiction, that $u|_{\veca}+\ideal{h_1|_\veca,\dots,h_{k'}|_\veca}$
is a zerodivisor in $\FF[\vecT]/\ideal{h_1|_\veca,\dots,h_{k'}|_\veca}$.
By the bounded-degree reduction above, there exists $v\in W|_\veca$ such that
$u|_\veca\cdot v\in\ideal{h_1|_\veca,\dots,h_{k'}|_\veca}$ but $v\notin\ideal{h_1|_\veca,\dots,h_{k'}|_\veca}$,
or equivalently, $v\in W_1|_\veca$ but
$v\notin W_2|_\veca$.
But this contradicts $W_1|_\veca=W_2|_\veca$.
\end{proof}

\injectivity*

\begin{proof}
First, for all $i\in\{0,1,2\}$ and $j\in[m_i]$, add to $S_3$ the denominators of all coefficients of $\res(f_{i,j})$. This ensures that the polynomials $\res(f_{i,j})|_\veca$ are well-defined.

As $\res(f_{0,1})(g_1,g_2)\in \ideal{h_1,\dots,h_{k'}}$, by \cref{lem:ideal-membership-coeff}, we may write 
\[
\res(f_{0,1})(g_1,g_2)=\sum_{i=1}^{k'} q_{i} h_i,
\] where $q_i\in  \KK^{(e)}_{\leq d_0'}[\vecT]_{\leq d_0}$ for some $d_0,d_0'\in O_{\delta}(1)$.
Add to $S_3$ the denominators of all the coefficients of $q_1,\dots,q_{k'}$. 
Then $\res(f_{0,1})(g_1,g_2)|_\veca=\sum_{i=1}^{k'} q_{i}|_\veca h_i|_\veca$. So $\res(f_{0,1})(g_1,g_2)|_\veca\in \ideal{h_1|_\veca,\dots,h_{k'}|_\veca}$.
 It follows that the $\FF$-linear ring homomorphism $\FF[Z_1,Z_2]\to \FF[\vecT]/\ideal{h_1|_\veca,\dots,h_{k'}|_\veca}$ sending $Z_i$ to $g_i|_\veca+\ideal{h_1|_\veca,\dots,h_{k'}|_\veca}$ maps $\res(f_{0,1})|_\veca$ to zero, as desired. 
 And by construction, we have $S_3\subseteq\KK^{(e)}_{\leq D_3}\cap \FF[\vecY^{1/p^e}]$ and $|S_3|=O_\delta(d)$, where $D_3=O_\delta(1)$.
\end{proof}

\end{document}